\documentclass[11pt]{article}
\usepackage[centertags]{amsmath}
\usepackage{amsfonts}
\usepackage{amssymb}
\usepackage{amsthm}
\usepackage{newlfont}
\usepackage{verbatim}
\usepackage[a4paper=true]{hyperref}

\hypersetup{
    colorlinks = true,
    linkcolor = magenta,
    anchorcolor = red,
    citecolor = blue,
    filecolor = red,
    pagecolor = red,
    urlcolor = red
}

\newtheorem{thm}{Theorem}[section]
\newtheorem*{thm*}{Theorem}
\newtheorem{cor}[thm]{Corollary}

\newtheorem{lem}[thm]{Lemma}

\newtheorem{prop}[thm]{Proposition}
\newtheorem*{prop*}{Proposition}

\newtheorem*{conj*}{Conjecture}

\newtheorem*{dfn*}{Definition}
\theoremstyle{definition}
\newtheorem{rem}[thm]{\textbf{Remark}}
\newtheorem*{rmk*}{Remark}
\newtheorem*{fact*}{Fact}

\newcommand{\norm}[1]{\left\Vert#1\right\Vert}

\newcommand{\abs}[1]{\left\vert#1\right\vert}
\newcommand{\set}[1]{\left\{#1\right\}}
\newcommand{\brac}[1]{\left(#1\right)}

\newcommand{\Real}{\mathbb{R}}
\newcommand{\eps}{\varepsilon}
\newcommand{\K}{\mathcal{K}}
\newcommand{\I}{\mathcal{I}}

\newcommand{\vol}{\mathrm{vol}}
\newcommand{\Var}{\mathrm{Var}}

\newcommand{\Ent}{\mathrm{Ent}}
\newcommand{\E}{\mathbb E}
\renewcommand{\P}{\mathbb P}
\def \F {\mathcal{F}}

\newlength{\defbaselineskip}
\newcommand{\setlinespacing}[1]           {\setlength{\baselineskip}{#1 \defbaselineskip}}

\numberwithin{equation}{section}

\begin{document}

\title{Transference Principles for Log-Sobolev and Spectral-Gap with Applications to Conservative Spin Systems}
\author{Franck Barthe\textsuperscript{1} and Emanuel Milman\textsuperscript{2}}

\date{}

\footnotetext[1]{Institut de Math\'ematiques de Toulouse, CNRS UMR 5219, Universit\'e Paul Sabatier, 31062 Toulouse Cedex 09, France. Email:  barthe@math.univ-toulouse.fr.}

\footnotetext[2]{Department of Mathematics, Technion - Israel
Institute of Technology, Haifa 32000, Israel. Supported by ISF, GIF, BSF and the Taub Foundation (Landau Fellow). Email:
emilman@tx.technion.ac.il.}

\maketitle

\begin{abstract}
We obtain new principles for transferring log-Sobolev and Spectral-Gap inequalities from a source metric-measure space to a target one, 
when the  curvature of the target space is bounded from below. 
As our main application, we obtain explicit estimates for the log-Sobolev and Spectral-Gap constants of various conservative spin system models, consisting of non-interacting and weakly-interacting particles, constrained to conserve the mean-spin.  When the self-interaction is a perturbation of a strongly convex potential, this partially recovers and partially extends previous results of Caputo, Chafa\"{\i}, Grunewald, Landim, Lu, Menz, Otto, Panizo, Villani, Westdickenberg and Yau. When the self-interaction is only assumed to be (non-strongly) convex, as in the case of the two-sided exponential measure, we obtain sharp estimates on the system's spectral-gap as a function of the mean-spin, independently of the size of the system. 
\end{abstract}

\section{Introduction}

The log-Sobolev and spectral-gap (or Poincar\'e) inequalities are among the most fundamental and useful functional inequalities for the analysis of equilibrium and non-equilibrium statistical mechanical systems. For instance, in the context of $n$-particle spin-systems, it is known that under various typical conditions (see \cite{YoshidaLogSobEquivToDecayOfCorrelations,StroockZegarlinskiLogSobImpliesDSCondition,BodineauHelfferSurvey,LedouxSpinSystemsRevisited}), the existence of a uniform lower bound on the spectral-gap or log-Sobolev constant associated to a spin-system, independent of boundary conditions and system-size, is equivalent to the exponential decay of spin-spin correlations in the distance between sites, and is thus a strong manifestation of the existence of a unique phase in the thermodynamic limit. However, for many natural non-trivial models, it is by no means an easy task to obtain uniform or other quantitative bounds on the spectral-gap or log-Sobolev constants. To elucidate this point, let us start by introducing our protagonists.

Let $(\Omega,d,\mu)$ denote a measure-metric space, meaning that $(\Omega,d)$ is a separable metric space and $\mu$ is a Borel probability measure on $(\Omega,d)$. 
Let $\F = \F(\Omega,d)$ denote the space of functions which are Lipschitz on every ball in $(\Omega,d)$. Given $f \in \F$, define $|\nabla f|$ as the following Borel function:
\[
 \abs{\nabla f}(x) := \limsup_{d(y,x) \rightarrow 0+} \frac{|f(y) - f(x)|}{d(x,y)} ~.
\]
(and we define it as 0 if $x$ is an isolated point - see \cite[pp. 184,189]{BobkovHoudre} for more details). In the smooth Euclidean setting, $\abs{\nabla f}$ of course coincides with the Euclidean length of the gradient of $f$. The study of the \emph{log-Sobolev inequality} was initiated in the works of Federbush \cite{Federbush-LogSobolev} and Gross \cite{GrossIntroducesLogSob} (cf. Stam \cite{Stam-LogSobolev}), and is a key feature of the Gaussian measure (see e.g. \cite{Ledoux-Book} for a general introduction and applications):

\begin{dfn*}
$(\Omega,d,\mu)$ is said to satisfy a log-Sobolev (LS) inequality with constant $\rho>0$ ($LSI(\rho)$) if:
\begin{equation} \label{eq:LS-inq-def}
\forall f \in \F \;\;\;\;  \frac{\rho}{2} \, \Ent_\mu(f^2) \leq \int |\nabla f|^2 d\mu ~,
\end{equation}
where $\Ent_\mu(g)$ denotes the entropy of a non-negative function $g$:
\[
 \Ent_\mu(g) := \int g \log g \, d\mu - \left(\int g \, d\mu  \right)\log\left(\int g \, d\mu\right) ~.
\]
The best possible constant $\rho$ above is denoted by $\rho_{LS} = \rho_{LS}(\Omega,d,\mu)$.
\end{dfn*}

\begin{dfn*}
$(\Omega,d,\mu)$ is said to satisfy a spectral-gap (SG) inequality with constant $\rho>0$ ($SG(\rho)$) if:
\begin{equation} \label{eq:SG-inq-def}
\forall f \in \F \;\;\;\;  \rho \, \Var_\mu(f) \leq \int |\nabla f|^2 d\mu ~,
\end{equation}
where $\Var_\mu(f)$ denotes the variance of $f$:
\[
\Var_\mu(f) := \int f^2 d\mu - \left(\int f\,  d\mu\right)^2 ~.
\]
The best possible constant $\rho$ above is denoted by $\rho_{SG} = \rho_{SG}(\Omega,d,\mu)$.
\end{dfn*}

By linearizing the LS inequality around constant functions, it is easy to verify that $\rho_{SG} \geq \rho_{LS}$ (see e.g. \cite{Ledoux-Book}).  The spectral-gap $\rho_{SG}$ controls the rate of convergence of an appropriate diffusion to the stationary measure $\mu$ in the variance sense, whereas the stronger LS constant $\rho_{LS}$ controls the rate of convergence in the entropy sense. By the tensorization property and the Bakry--\'Emery criterion, both types of inequalities are well-known to hold on products of spaces satisfying the corresponding inequality and on strongly convex spaces, respectively (see Appendix for details). 

One natural way of obtaining (lower) bounds on $\rho_{LS}$ or $\rho_{SG}$, is to start from a well-understood space $(\Omega,d,\mu)$, and to \emph{transfer} the $LS$ or $SG$ inequality from that space to a perturbation thereof. To this end, a \emph{transference principle} or \emph{stability result} for these inequalities is required. The most common type of perturbation is when the underlying metric space $(\Omega,d)$ remains fixed and only the measure $\mu$ is perturbed, and this is the case we will consider here. A very well-known transference principle for the log-Sobolev inequality is given by the Holley--Stroock lemma \cite{HolleyStroockPerturbationLemma}, which states that if $\mu_{1}$ and $\mu_{2}$ are two mutually absolutely continuous Borel probability measures on $(\Omega,d)$, then:
\begin{equation} \label{eq:HS}
\norm{\frac{d\mu_{2}}{d\mu_{1}}}_{L^{\infty}} \leq L_{2} ~,~  \norm{\frac{d\mu_{1}}{d\mu_{2}}}_{L^{\infty}} \leq L_{1} \;\;\; \Rightarrow \;\;\; \rho_{LS}(\Omega,d,\mu_{2}) \geq \frac{1}{L_{1}L_{2}} \rho_{LS}(\Omega,d,\mu_{1}) ~.
\end{equation}
A completely analogous statement trivially holds for the $SG$ inequality. Unfortunately, although being an extremely useful tool, a naive application of the Holley--Stroock lemma in a high-dimensional situation (consider simultaneously perturbing each single-site potential) will typically lead to an exponential degradation of $\rho_{LS}$ in the dimension, severely obstructing any hope of obtaining uniform bounds. 

Our aim in this work is to present several alternative general transference principles for the $LS$ and $SG$ inequalities, and as an application, to test their performance on a conservative spin model with and without weak-interactions. In this Introduction, we will put more emphasis on describing the application to conservative spin models, but before proceeding with this application, we briefly describe our general transference principles.

\subsection{The Transference Principles - A Brief Taste}

The Holley-Stroock perturbation principle (\ref{eq:HS})   requires uniform upper and lower bounds on  $\log(d\mu_{2}/d\mu_{1})$.  Our transference principles   rely  on an upper bound on  $\norm{d\mu_{2}/d\mu_{1}}_{L^{p}(\mu_{1})}$ for some $p > 1$, allowing $d\mu_{2}/d\mu_{1}$ to vanish or to explode on a $\mu_{1}$-small part of the space. They also rely on geometric assumptions on the spaces involved:
 we require  that $(\Omega,d)$ be given by a complete oriented smooth connected Riemannian manifold $(M,g)$, endowed with its natural geodesic distance $d$, and so that $(M,g,\mu_{2})$ satisfies an appropriate \emph{curvature lower-bound condition}. For the $SG$ inequality, we require the curvature to be non-negative, whereas for the $LS$ inequality we allow it to be bounded below by $-\kappa$, but in addition require that $\rho_{LS}(M,g,\mu_{1})$ is big enough with respect to $\kappa$. Our method is based on the equivalence between isoperimetric and concentration inequalities in the latter setting due to the second named author \cite{EMilman-RoleOfConvexity,EMilmanGeometricApproachPartI}, and extends the transference principles obtained in \cite{EMilman-RoleOfConvexity,EMilmanGeometricApproachPartII}. Let us state now a sample result and refer to Section \ref{sec:trans} for a more comprehensive account:

\begin{thm}[log-Sobolev Transference under Curvature Lower Bound - Euclidean Setting] \label{thm:intro-ls-trans}
Let $\mu_i = \exp(-V_i(x)) \, dx$ ($i=1,2$) denote two Borel probability measures on Euclidean space $(\Real^{n},|\cdot|)$, and assume that $V \in C^2(\Real^{n})$ and that $ \mathrm{Hess} V_{2} \geq -\kappa Id$ ($\kappa \geq 0$). Assume that $(\Real^{n},|\cdot|,\mu_1)$ satisfies a strong-enough log-Sobolev inequality:
\begin{equation} \label{eq:intro-ls-assump}
\rho  = \rho_{LS}(\Real^{n},|\cdot|,\mu_1) > \frac{4p}{p-1} \kappa ~,
\end{equation}
for some $p > 1$, and that:
\[
\brac{\int \brac{\frac{d\mu_2}{d\mu_1}}^p d\mu_1}^{1/p} \leq L ~.
\] 
Then $(\Real^{n},|\cdot|,\mu_2)$ satisfies a log-Sobolev inequality:
\[
\rho_{LS}(\Real^{n},|\cdot|,\mu_2) \geq C(\rho,\kappa,L,p)  ~,
\]
where:
\[
C(\rho,\kappa,L,p) := c \; \rho \; \frac{p-1}{p} \; \exp(-C (1 + \log(L)) / \theta) ~~,~~ \theta := 1 - \frac{4p \kappa}{(p-1) \rho} ~,
\]
and $c,C>0$ are universal numeric constants. Moreover, when $\kappa = 0$, one may in fact use:
\begin{equation} \label{eq:intro-L-dep-convex}
C(\rho,0,L,p) = c \; \rho \; \frac{p-1}{p} \frac{1}{1+ \log(L)} ~.
\end{equation}
\end{thm}

Note the logarithmic dependence on the perturbation parameter $L$ in (\ref{eq:intro-L-dep-convex}) when $\kappa = 0$ (i.e. when $V_2$ is convex), improving over the linear dependence in (\ref{eq:HS}).  A similar result holds for the $SG$ inequality when $\kappa = 0$ (and in fact much more may be said, see Subsection \ref{subsec:SG-trans}). See Remark \ref{rem:constants} for more on the constant $4$ in (\ref{eq:intro-ls-assump}).

Obtaining  dimension-free estimates of the malleable quantity $\norm{d\mu_{2}/d\mu_{1}}_{L^{p}(\mu_{1})}$ is already a more feasible task for many natural models, as we demonstrate in the case of the conservative spin model. Furthermore, the convenience of controlling this quantity is especially apparent when superimposing several perturbations of $\mu_1$. For instance, if $\mu_2 = f g \mu_1 / \int f g d\mu_1$ with $f g > 0$, then by several applications of Cauchy--Schwartz we may estimate:
\begin{eqnarray}
\label{eq:superimpose}\nonumber  & & \norm{\frac{d\mu_2}{d\mu_1}}^p_{L^p(\mu_1)} = \frac{\int f^p g^p d\mu_1}{ (\int f g d\mu_1)^p} \leq 
\brac{\int f^{2p} \; d\mu_1}^{1/2} \brac{\int g^{2p} \; d\mu_1}^{1/2} \brac{\int \frac{1}{fg} \; d\mu_1}^p  \\
& \leq & \brac{\int f^{2p} \; d\mu_1}^{1/2} \brac{\int g^{2p} \; d\mu_1}^{1/2} \brac{\int f^{-2} \; d\mu_1}^{p/2} \brac{\int g^{-2} \; d\mu_1}^{p/2} ~,
\end{eqnarray}
and so it is enough to analyze each perturbation ($f$ and $g$) of $\mu_1$ separately. We illustrate this point by adding weak-interactions to our conservative spin model in Section \ref{sec:inter}. 

\subsection{The Conservative Spin Model}

Let $V \in C^2(\Real)$ denote a single-site potential, so that $\mu = \exp(-V(x))\,  dx$ is a probability measure on $\Real$ with barycenter at the origin. Let $\mu_n := \mu^{\otimes n} = \exp(-H(x))\,  dx$  denote the Gibbs measure on $\Real^n$ corresponding to the grand canonical ensemble of non-interacting spins,
where $H(x)$ denotes the non-interacting Hamiltonian $H(x) = \sum_{i=1}^n V(x_i)$. By the well-known tensorization property of LS and SG inequalities (see Appendix), it follows that $\rho_{I}(\Real^n,|\cdot|,\mu_n) = \rho_{I}(\Real,|\cdot|,\mu)$, for $I = LS , SG$, and so the product space $(\Real^n,|\cdot|,\mu_n)$ is well-understood.  Here and elsewhere, $\Real^n$ is equipped with its standard Euclidean structure $|\cdot|$. 

The measure $\mu_{E_s}$, corresponding to the canonical ensemble having mean-spin $s \in \Real$, is obtained by restricting $\mu_n$ onto the hyperplane $E_s$ given by $\frac{1}{n} \sum_{i=1}^n x_i = s$, namely:
\begin{equation} \label{eq:intro-mu}
\mu_{E_s} = \frac{1}{Z_{E_s}} \exp(-H(x))\,  d\vol_{E_s}(x) ~,
\end{equation}
where we denote by $\vol_{E_s}$ the induced Lebesgue measure on the hyperplane $E_s$, and  by $Z_{E_s} > 0$ a normalization term. 

In \cite{VaradhanNonlinearDiffusionLimit}, Varadhan asked to characterize those potentials $V$ so that $(E_s,|\cdot|,\mu_{E_s})$ satisfies a SG inequality, uniformly in the system-size $n \geq 2$ and mean-spin $s \in \Real$.  Varadhan's question naturally extends to the LS inequality case as well. When $V$ is a $C^2$ perturbation of a (strongly convex) super-quadratic potential and under some mild technical assumptions, a positive answer to Varadhan's question for SG was obtained by Caputo \cite{CaputoUniformPoincareForPerturbedStricltyConvexPonentials}. A similar answer cannot be expected in the sub-quadratic case, as follows from the work of the first named author and Wolff. In \cite{BartheWolffGammaDistributions}, these authors showed that when $\mu$ is the exponential measure on $\Real_+$ (and more generally, when $\mu$ is the Gamma distribution), then $\rho_{SG}(E_s,|\cdot|,\mu_{E_s})$ behaves like $1 / s^2$ (albeit independently of $n$), and that $\rho_{LS}(E_s,|\cdot|,\mu_{E_s})$ behaves like $1 / n   s^2$. As for the corresponding questions in the LS case, various authors have addressed (using different methods) the case when $V$ is a $C^2$ perturbation of a quadratic \cite{LandimPanizoYauLSIForPerturbationOfQuadraticPotential,ChafaiConservativeSpinSystems,GOVWTwoScaleApproachForLSI}, starting from the work of Lu and Yau \cite{LuYauLSIForKawasakiAndGlauber}. In particular, Grunewald, Otto, Villani and Westdickenberg analyzed the LS case in \cite{GOVWTwoScaleApproachForLSI} using a two-scale (macroscopic - microscopic) approach, and established the system's convergence to the hydrodynamic limit. Building on this two-scale approach and extending it to a multi-scale one, Menz and Otto succeeded in \cite{MenzOttoLSIForPerturbationsOfSuperQuadraticPotential} to obtain a positive answer to the LS version of Varadhan's question when $V$ is a $C^1$ perturbation of a (strongly convex) super-quadratic potential.  We have recently learned that Fathi and Menz (in preparation) have succeeded to reprove this result by using the original two-scale approach.  

Subsequently, Menz \cite{MenzLSIwithWeakInteraction} obtained a positive answer for LS in the presence of weak-interactions between spins, when $V$ is a $C^2$ perturbation of a quadratic potential. Given a symmetric $n$ by $n$  matrix $A = \set{a_{i,j}}$ with zero diagonal ($a_{i,j} = a_{j,i}$ and $a_{i,i} = 0$) and a vector $b \in \Real^n$, the  spin-interaction term $-I_A(x)$ and boundary contribution term $B_b(x)$ are defined and added to the non-interacting Hamiltonian $H$ as follows (see Section \ref{sec:mean-spin} for more details):
\[
H_{A,b}(x) = H(x)  + B_b(x) - I_{A}(x) ~,~ B_b(x) = \sum_{i=1}^n b_i x_i ~,~ I_{A}(x) := \sum_{i,j =1}^n a_{i,j} x_i x_j ~.
\]
The corresponding weakly-interacting conservative Gibbs measure $\mu_{A,E_s,b}$ is defined as in (\ref{eq:intro-mu}) by replacing $H$ with $H_{A,b}$; when $b=0$, we simply denote it by $\mu_{A,E_s}$.

\subsection{Applying the Transference Principles}

In this work, we apply our transference principles to partially recover and partially extend the above mentioned results. Before stating a sample of our results, let us briefly describe our method. By a standard application of Cram\'er's trick, which is the key ingredient in the Cram\'er Theorem on Large Deviations - a central tool in the approaches of \cite{GOVWTwoScaleApproachForLSI,MenzOttoLSIForPerturbationsOfSuperQuadraticPotential,MenzLSIwithWeakInteraction} - we reduce the analysis of the general mean-spin case $s \in \Real$, to the case $s = 0$. The idea is to replace $\mu$ with $\mu^{\wedge a} := c_a \mu \exp(a x)$, for an $a = a(s)$ appropriately chosen so that the barycenter of $\mu^{\wedge a(s)}$ is at $s$, and to note that $(E_s,|\cdot|,\mu_{E_s})$ is isometric (as a metric-measure space) to $(E,|\cdot|,(\mu^{a(s)})_{E})$, where $\mu^a$ is the translation of $\mu^{\wedge a}$ having barycenter at $0$, and where we denote $E := E_0$. This trick also applies in the presence of weak-interactions and boundary contribution, but there the argument is more subtle, and requires the existence of a solution to a system of non-linear equations, which is guaranteed by an application of Banach's fixed point theorem (see Lemma \ref{lem:inter-invariance}). 

Next, note that $\mu_{E}$ is not absolutely continuous with respect to $\mu_n$, and so to be able to apply our transference results, we define $\mu_{E,w}$ by ``thickening" $\mu_{E}$ uniformly in the diagonal direction $D$ by a width of $w > 0$ from each side:
\[
d\mu_{E,w}(x) = \frac{1}{Z_{E,w}} \exp\brac{-H(\pi_{E}(x))} 1_{|\pi_D(x)| \leq w}\,  dx ~,~ Z_{E,w} = 2 w Z_{E} ~,
\]
where $\pi_F$ denotes orthogonal projection onto the subspace $F$.  The measure $\mu_{A,E,b,w}$ is defined analogously, replacing $H$ with $H_{A,b}$ above. 

\subsubsection{Obtaining log-Sobolev inequalities}

 By controlling $\norm{d\mu_{E,w}/d\mu_n}_{L^p(\mu_n)}$ for e.g. $p=4$ and applying Theorem \ref{thm:intro-ls-trans}, the LS inequality is transferred from $(\Real^n,|\cdot|,\mu_n)$ onto $(\Real^n,|\cdot|,\mu_{E,w})$ under appropriate assumptions on $V$; and from the latter space, which is a product of $(E,|\cdot|,\mu_{E})$ and the uniform measure on an interval of length $2w$ on $D$, the LS inequality is immediately established on $(E,|\cdot|,\mu_{E})$, and it remains to optimize over $w>0$. 
 
 To control $\norm{d\mu_{E,w}/d\mu_n}_{L^p(\mu_n)}$ for say $p > 0$, 
we require upper estimates, exponentially decaying in $u$, of:
\begin{eqnarray*}
&   & \mu_n\set{x \in \Real^n \; ; \; \abs{\pi_D(x)} \leq w ~,~ H(x) - H(\pi_E(x)) \geq u} \\
& = & \mu_n \set{x \in \Real^n \; ; \; \abs{\pi_D(x)} \leq w ~,~ \sum_{i=1}^n V(x_i) - \sum_{i=1}^n V(\pi_E(x)_i) \geq u } \\
& \leq & \mu_n \set{x \in \Real^n \; ; \; \exists t \in [-w,w] ~,~ \sum_{i=1}^n V(x_i) - \sum_{i=1}^n V(x_i + t/\sqrt{n}) \geq u} ~.
\end{eqnarray*}
Applying Taylor's theorem and expanding $V(x_i + t / \sqrt{n})$ to the second order, it remains to control the large deviation of:
\[
 \P \brac{\frac{w}{\sqrt{n}} \abs{\sum_{i=1}^n V'(X_i)} \geq u_1 } ~,~ \P \brac{ \frac{w^2}{2n} \sum_{i=1}^{n} (Y_i - \E(Y_0)) \geq u_2 } ~,
\]
where $X_1,\ldots,X_n$ and $Y_1,\ldots,Y_n$ are independent random variables having distribution identical to that of (respectively) $X_0$ and $Y_0 := \sup_{\xi \in [X_0-\delta,X_0+\delta]} |V''(\xi)|$ for a fixed $\delta > 0$, where $X_0$ is distributed according to $\mu$. Since $\E(V'(X_0)) = \int V'(x) \exp(-V(x)) dx = 0$, it is clear by the Central-Limit Theorem and the Law of Large Numbers that both terms above converge as $n \rightarrow \infty$. However, to obtain estimates valid for each individual $n$, we resort to other classical large deviation tools such as Bernstein's theorem. Furthermore, if $V''$ is bounded below, then control over the second order term is automatic. Finally, to control the normalization term $Z_E$, we apply the Berry--Esseen Theorem or a local version thereof. 

Superimposing weak-interaction on our model is also easily handled, by separately analyzing it using (\ref{eq:superimpose}). To this end, we use a non-symmetric version of the  Hanson--Wright order-two sub-Gaussian chaos deviation estimates \cite{HansonWright}, kindly communicated to us by Rafal Lata{\l}a, to whom we are grateful (see Appendix).     

\medskip

The following notation will be used throughout this work. Recall that the $\Psi_1$ norm of a random-variable $Y$, denoted $\norm{Y}_{L_{\Psi_1}}$, is defined as follows:
\[
\norm{Y}_{L_{\Psi_1}} := \inf\set{ \lambda > 0 \; ; \; \E \exp(|Y|/\lambda) \leq e} ~.
\]
Observe that $\norm{\cdot}_{L_{\Psi_1}}$ is indeed a norm, thanks to the convexity of the exponential function, and that our normalization ensures that $\norm{1}_{L_{\Psi_1}} = 1$. Set:
\begin{itemize}
\item $M_p = M_p(\mu) := \E |X_0|^p$. 
\item $D_{1,\Psi_1} = D_{1,\Psi_1}(\mu) := \norm{V'(X_0)}_{L_{\Psi_1}}$.
\item $D_{2,\Psi_1} = D^\delta_{2,\Psi_1} = D^\delta_{2,\Psi_1}(\mu) := \norm{Y_0}_{L_{\Psi_1}}$. 
\item $D_2 = D^\delta_2 = D^\delta_2(\mu) := \E(Y_0)$.
\end{itemize}
Finally, we denote the dependence of a certain parameter $\alpha$ on other parameters $\beta,\gamma,\delta,\eps$ as follows: $\alpha^{\beta,\gamma}_{\delta,\eps}$ means that $\alpha$ only depends on \emph{upper bounds} on $\beta,\eps$ and on \emph{lower bounds} on $\delta,\eps$.

\medskip

We now provide a few samples to illustrate the types of results we obtain:

\begin{thm} \label{thm:intro-model-zero-spin-2}
Let $\mu= \exp(-V(x))\,  dx$ denote a probability measure on $\Real$ with barycenter at $0$ so that $D_{1,\Psi_1} = D_{1,\Psi_1}(\mu) < \infty$ and
$\lambda := \norm{d\mu/dx}_{L^\infty} < \infty$. Assume in addition that $(\Real,|\cdot|,\mu)$ satisfies $LSI(\rho)$, and that:
\begin{equation} \label{eq:intro-model-kappa-2}
-\kappa := \inf_{x \in \Real} V''(x) \geq -\frac{\rho}{8} ~.
\end{equation}
Then for any integer $n$ greater than $\eta = \eta^{M_3/M_2^{3/2},M_3 \lambda^3}$, the conservative zero-mean spin system $(E,|\cdot|,\mu_E)$ satisfies LSI with constant:
\[
\rho_{LS}(E,|\cdot|,\mu_E) \geq c \frac{\rho}{Q^C} ~,
\]
where $c,C>0$ are universal constants and $Q$ is the following scale-invariant quantity:
\[
Q:= \max(1, (\kappa + D_{1,\Psi_1}^2)M_2) ~.
\]
\end{thm}

\begin{thm} \label{thm:intro-inter-zero-spin}
Let $\mu = \exp(-V(x))\,  dx$ denote a probability measure on $\Real$ with barycenter at the origin, so that $D_{1,\Psi_1} = D_{1,\Psi_1}(\mu) < \infty$ and $D_{2,\Psi_1} = D^\delta_{2,\Psi_1}(\mu) < \infty$ for some $\delta > 0$. Let $A$ denote an $n$ by $n$ symmetric matrix with zero diagonal.
Assume that $(\Real,|\cdot|,\mu)$ satisfies $LSI(\rho)$, and that:
\begin{itemize}
\item
\begin{equation} \label{eq:intro-inter-kappa}
-\kappa := \inf_{x \in \Real} V''(x) \geq -\frac{\rho}{8} ~;
\end{equation}
\item
\[ \norm{A}_{op} \leq c \rho ~,
\] where $c>0$ is an appropriate universal constant.
\end{itemize}
Then for any integer $n$ greater than $\eta = \eta^{M_3 / M_2^{3/2}, D_{1,\Psi_1},D_{2,\Psi_1}}_{\rho,\delta}$, the weakly-interacting conservative zero-boundary zero-mean spin system
$(E,|\cdot|,\mu_{A,E})$ satisfies LSI with constant:
\[
\rho_{LS}(E,|\cdot|,\mu_{A,E}) \geq c \frac{\rho}{Q^C} ~,
\]
where $c,C>0$ are universal constants and $Q$ is the following scale-invariant quantity:
\[
Q := \max\brac{1 , M_2 (D_2 + D_{1,\Psi_1}^2)} \exp(\norm{A}_{HS}^2 / \rho^2) ~.
\]
\end{thm}

Here and throughout we denote the operator and Hilbert-Schmidt norms of $A$ by $\norm{A}_{op}$ and $\norm{A}_{HS}$, respectively. To demonstrate the desired uniformity in system-size $n \geq 2$, mean-spin $s \in \Real$ and boundary contribution $b \in \Real^n$ for a concrete class of measures $\mu$, we consider the following:
\begin{dfn*}
The probability measure $\mu = \exp(-V(x))\, dx$ is called $(\alpha,\beta,\omega)$ \emph{weakly Gaussian} if we may decompose $V = V_{\text{conv}} + V_{\text{pert}}$ so that:
\begin{itemize}
\item $V_{\text{conv}}, V_{\text{pert}} \in C^2(\Real)$.
\item $V_{\text{conv}}'' \geq \alpha > 0$.
\item $\sup V_{\text{pert}} - \inf V_{\text{pert}} \leq \omega < \infty$.
\item $- \kappa := - \frac{1}{8} \alpha \exp(-\omega) \leq V'' \leq \beta < \infty$. \end{itemize}
\end{dfn*}

\begin{thm}
Let $\mu$ be a $(\alpha,\beta,\omega)$ weakly Gaussian measure, and let $A$ denote an $n$ by $n$ symmetric matrix with zero diagonal satisfying:
\[
\norm{A}_{op} \leq c \alpha \exp(-\omega) ~,
\]
for an appropriate universal constant $c > 0$. 
Then the canonical ensemble $(E_s,\abs{\cdot},\mu_{A,E_s,b})$ with weak-interaction $A$, satisfies a LSI, uniformly in the system size $n \geq 2$, mean-spin value $s \in \Real$ and boundary contribution $b \in \Real^n$, depending solely on a positive lower bound on $\alpha$ and upper bounds on $\beta$, $\omega$ and $\norm{A}_{HS} / (\alpha \exp(-\omega))$.
\end{thm}

\subsubsection{Obtaining Spectral-Gap inequalities}

We now turn our attention to the weaker SG inequalities. Unfortunately, our transference principle for the SG inequality only applies (at least in the Euclidean case) when the target measure is \emph{log-concave}. Recall that a measure $\nu = \exp(-W(x))\, dx$ on $\Real^n$ is called log-concave if $W : \Real^n \rightarrow \Real \cup \set{\infty}$ is convex. We naturally consider a scenario where no LS inequality is possible, when the log-concave probability measure $\mu = \exp(-V(x))\, dx$ does not have sub-Gaussian tail decay and the potential $V$ is not necessarily strongly convex:

\begin{thm} \label{thm:intro-SG-any-mean}
Let $\mu = \exp(-V(x)) \, dx$ denote a log-concave probability measure on $\Real$. Assume that either:
\begin{enumerate}
\item $V$ is Lipschitz with constant $L$. 
\item $V \in C^1(\Real)$ and $V'$ is Lipschitz with constant $L^2$. 
\end{enumerate}
Given $s \in \Real$, denote $\rho_s := \rho_{SG}(\Real,|\cdot|,\mu^{a(s)})$. Then for any integer $n \geq 2$ and mean-spin $s \in \Real$:
\begin{equation} \label{eq:intro-SG-bound}
\rho_{SG}(E_s,|\cdot|,\mu_{E_s}) \geq c \frac{\rho_s}{\log(2 + L^2/\rho_s)^2} ~,
\end{equation}
where $c>0$ is a  universal constant.
\end{thm}

We show in Section \ref{sec:SG} that any log-concave probability measure satisfies the opposite inequality:
$
\rho_{SG}(E_s,|\cdot|,\mu_{E_s}) \leq C \rho_s$,
for some universal constant $C>0$. A well-known conjecture in convexity of Kannan, Lov\'asz and Simonovits \cite{KLS} predicts that the logarithmic term in (\ref{eq:intro-SG-bound}) and the restrictions on the log-concave measure $\mu$  in Theorem \ref{thm:intro-SG-any-mean} may be removed. 
In the special case of the two-sided exponential measure $\nu = \frac{1}{2} \exp(-|x|) dx$, we are indeed able to resolve this conjecture, which translates into:

\begin{thm} \label{thm:intro-two-sided-exp}
Let $\nu = \frac{1}{2} \exp(-|x|)\,  dx$. Then the canonical ensemble $(E_s,|\cdot|,\nu_{E_s})$ satisfies:
\[
\rho_{SG}(E_s,|\cdot|,\nu_{E_s})  \simeq \frac{1}{1+s^2} ~,
\]
uniformly in the system-size $n \geq 2$ and mean-spin $s \in \Real$. 
\end{thm}

Here and throughout, we use $A \simeq B$  to denote that $c_1 \leq A/B \leq c_2$ for some two universal constants $c_1,c_2 > 0$. Theorem \ref{thm:intro-two-sided-exp} answers a question suggested to us by Pietro Caputo, whom we would like to thank, which served as motivation for this entire work.  His question was inspired by the relationship
between this specific canonical ensemble  and a continuous version of the one-dimensional
Solid-on-Solid model. See \cite{SOS} for recent results and references.

\subsection{Comparison With Previous Works}

Theorem \ref{thm:intro-ls-trans} should be compared with the Otto--Reznikoff criterion for LSI \cite{OttoReznikoffLSICriterion}. Roughly speaking, the latter result  applies on a product space, where the $i$-th factor satisfies $LSI(\rho_i)$, and a smallness condition (relative to the $\set{\rho_i}$'s) on the interaction between different factors is required, in the form of a lower bound on the cross-factors' Hessian. Our Theorem \ref{thm:intro-ls-trans}, which applies to an arbitrary not-necessarily-product space, also requires a smallness condition on the lower bound of the Hessian of the target measure $\mu_2$'s potential, relative to the LS constant of the source measure $\mu_1$. However, this smallness is required not just for the ``cross-terms", but also for the ``diagonal-terms" in some sense. 

Our approach to functional inequalities for conservative spin systems, based on transference principles, is novel 
and has several advantages. Firstly, it follows closely the natural heuristics: if $\int x\, d\mu(x)=0$
and the $X_i$'s are independent copies distributed according to $\mu$, then by the Law of Large Numbers and Central Limit Theorem, conditioning on $\sum_{i=1}^n X_i=0$
should be less and less stringent as $n$ increases, and hence the measure $\mu_E$ is expected to behave similarly to the product measure $\mu_n$; our proof precisely mimics this heuristic argument. Secondly, our approach is global and allows to deduce estimates on the ergodic constants for an individual 
value of the mean spin $s$ which hold uniformly in the dimension of the system. In contrast, previous
methods relied on a kind of induction on sub-systems, which seems to work well only when 
the ergodic constants are bounded independently of the dimension $n$ \emph{and} the spin $s$. For this reason, they 
could not be applied to sub-quadratic potentials, where a dependence in $s$ is expected. 
Finally, our approach is soft and flexible. It  adapts to weak interactions and is likely 
to be useful for other models.

Next, let us comment on our results in comparison to previous ones.
Our spectral-gap estimates for the conservative spin-model with convex self-interactions are new and almost sharp. Apart from the work \cite{BartheWolffGammaDistributions} on log-concave gamma distributions, which used very specific properties of these laws, our results are the first dimension-free estimates for potentials with a sub-quadratic growth.

Clearly, our results on the LS inequality for the conservative spin-models, being based on our general transference principle, do not fully recover the best known results described earlier. Our main limitation lies in the requirements (\ref{eq:intro-model-kappa-2})
and (\ref{eq:intro-inter-kappa}) that the LSI constant of $(\Real,|\cdot|,\mu)$ is strong-enough relative to the curvature lower-bound satisfied by the potential $V$. Consequently, we cannot handle the case of an \emph{arbitrarily large} $C^2$ perturbation of a quadratic potential as in \cite{LuYauLSIForKawasakiAndGlauber,LandimPanizoYauLSIForPerturbationOfQuadraticPotential,ChafaiConservativeSpinSystems,GOVWTwoScaleApproachForLSI}, let alone $C^1$ perturbations as in \cite{MenzOttoLSIForPerturbationsOfSuperQuadraticPotential}.
However we are able to handle  any \emph{weakly Gaussian measure}, i.e. small $C^2$ perturbations of not just a quadratic potential as in \cite{LuYauLSIForKawasakiAndGlauber,LandimPanizoYauLSIForPerturbationOfQuadraticPotential,ChafaiConservativeSpinSystems,GOVWTwoScaleApproachForLSI,MenzLSIwithWeakInteraction}, but of any potential $V$ with $0 < \alpha \leq V'' \leq \beta < \infty$. Furthermore, if no upper bound on  $V''$ is assumed, we can still handle the zero mean-spin case under a very mild technical assumption (essentially, that $V$ grows sub-exponentially); handling arbitrary mean-spin values via our method seems technically involved, but not impossible. Therefore in the non-interacting case, we propose a soft and robust approach to known results. Our approach avoids in particular the delicate and technical proofs of local versions of Cram\'er's Theorem, which are essential ingredients in the approaches of \cite{GOVWTwoScaleApproachForLSI,MenzOttoLSIForPerturbationsOfSuperQuadraticPotential,MenzLSIwithWeakInteraction}, and which must be reproved for every small variation of the underlying model (and in some cases, as in the weakly-interacting model of \cite{MenzLSIwithWeakInteraction}, are obtained by an indirect perturbation argument).

In the presence of weak-interactions, we partially recover and partially extend the work of Menz  \cite{MenzLSIwithWeakInteraction}, which deals with arbitrary $C^2$ bounded perturbations of a quadratic potential. Passing to a more general potential $V$ (even one with $0 < \alpha \leq V'' \leq \beta < \infty$) is a genuine issue for the approach of \cite{MenzLSIwithWeakInteraction}, due to the perturbative nature of the argument. In contrast, we can also deal with weakly Gaussian potentials if the interaction is small enough.
Furthermore, the result of \cite{MenzLSIwithWeakInteraction} requires that $\norm{\, |A| \,}_{op}$ be small-enough, where $|A| := \set{|a_{i,j}|}$ for $A = \set{a_{i,j}}$. Our estimates only 
require the smallness of  $\norm{A}_{op}$, which is always smaller (in the worst case by a factor $\sqrt n$):
$$  \norm{A}_{op} \leq \norm{\,|A|\,}_{op} \leq \norm{\,|A|\,}_{HS} = \norm{A}_{HS}\le  \sqrt n \, \norm{A}_{op}.$$
However the constants appearing in our estimates also depend on the value of $\norm{A}_{HS}$. 
So this advantage would only be apparent in  mean-field type situations, or at least when 
$\norm{A}_{HS}/ \norm{A}_{op}$ is upper-bounded independently of the dimension. 

To summarize, we present in this work a single framework for simultaneously handling several spin-models, which may be easily extended to other models or scenarios. 
All of our estimates are stated in terms of explicit dependencies on several concrete parameters extracted from the initial single-site measure $\mu$.

\medskip

The rest of the paper is organized as follows. Section \ref{sec:trans} presents our transference principles for the LS and SG inequalities. In Section \ref{sec:model}, we deduce a uniform bound in the system size on the LSI constant of a conservative spin system with zero-mean spin, under various assumptions on the self-interaction. Section \ref{sec:inter} achieves similar results for a weakly-interacting model. In Section \ref{sec:mean-spin}, we extend the uniformity of our bounds for both models to arbitrary mean-spin values. In Section \ref{sec:SG}, we turn our attention to the SG inequality when the self-interaction is assumed (non-strictly) convex. We obtain uniform bounds on the system size, whose dependence on the mean-spin is essentially sharp. The Appendix collects various useful auxiliary statements. 

\medskip

\noindent \textbf{Acknowledgements.} We thank Rafal Lata{\l}a and Sasha Sodin for their help with references and estimates, and Pietro Caputo, Georg Menz and Cedric Villani for their comments and interest. E.M. would also like to thank Thierry Bodineau and Tom Spencer for introducing him to Spin Systems.

\section{The Transference Principle} \label{sec:trans}

\subsection{Concentration Transference}

Let $(\Omega,d,\mu)$ denote a measure-metric space, meaning that $(\Omega,d)$ is a separable metric space and $\mu$ is a Borel probability measure.
Given such a space, its concentration profile $\K = \K(\Omega,d,\mu) : \Real_+ \rightarrow [0,1/2]$ is defined by:
\[
\K(r) := \sup \set{ \mu(\Omega \setminus A^d_r) \; ; \; \mu(A) \geq 1/2} ~, ~ A^d_r := \set{x \in \Omega \; ; \; d(x,A) < r} ~.
\]
It is immediate to check that this is equivalent to requiring that $\K$ is minimal with:
\[
\mu(A) > \K(r) \;\; \Rightarrow \;\; \mu(A^d_r) > 1/2 ~.
\]
Using the ``pessimistic" convention for the inverse of a non-increasing function:
\[
\K^{-1}(\eps) := \inf \set{r > 0 \; ; \; \K(r) < \eps} ~.
\]
it follows that:
\begin{equation} \label{eq:conc-reversed}
\mu(A) \geq \eps \;\; \Rightarrow \;\; \mu\big(\overline{A^d_r}\big) \geq 1/2 \;\; \text{ for } r = \K^{-1}(\eps) ~.
\end{equation}
Concentration inequalities provide upper bounds on the decay of $\K(r)$ in the large-deviation regime, as $r \rightarrow \infty$. Note that trivially $\K(r) \leq 1/2$ for all $r > 0$, and so the concentration inequality $\K(r) \leq \alpha(r)$ is only meaningful when $r > \alpha^{-1}(1/2)$. 

The following crucial lemma is a generalization of \cite[Lemma 3.1]{EMilmanGeometricApproachPartII}. 

\begin{lem}[Concentration Transference] \label{lem:conc-trans}
Let $\mu_1$,$\mu_2$ be two probability measures on a common metric space $(\Omega,d)$, with $\mu_2 \ll \mu_1$. 
Let $\K_i = \K_{(\Omega,d,\mu_i)}$ denote the corresponding concentration profiles. Assume that there exists a right-continuous non-increasing function $M: (0,1/4] \rightarrow (0,\infty)$ so that:
$$ 
\forall \eps \in (0,1/4]\,, \quad \mu_2\left(\Big\{x\in\Omega;\; \frac{d\mu_2}{d\mu_1}(x)>M(\eps)\Big\} \right)\le \eps.
$$
Then for all $r \geq 2 \K_1^{-1}(\beta(1/4))$,
\begin{equation} \label{eq:conc-trans-2}
\K_2(r) \leq 2 \beta^{-1}\left(\K_1\big(r/2\big)\right) ~, \end{equation}
where $\beta : [0,1/4] \rightarrow [0,\beta(1/4)]$ denotes the increasing function
$\displaystyle
\beta(\eps) := \eps/M(\eps)$, with the convention that $\beta(0) = 0$.
\end{lem}
\begin{proof}
Let $A \subset \Omega$ with $\mu_2(A) \geq 1/2$. Given $\eps \in (0,1/4]$, set 
\begin{equation}\label{eq:M-cond}
\Omega_\eps:=\left\{x\in \Omega;\;  \frac{d\mu_2}{d\mu_1}(x)\le M(\eps)\right\}.
\end{equation}
By hypothesis $\mu_2(\Omega_\eps)\ge 1-\eps$. Denote $A^\eps := A \cap \Omega_\eps$, and observe that
$\mu_2(A^\eps) \geq 1/2 - \eps$ and hence $\mu_1(A^\eps) \geq \delta_\eps := (1/2 - \eps)/M(\eps)$. Denoting $r_0 = \K_1^{-1}(\delta_\eps)$, it follows from (\ref{eq:conc-reversed}) that $\mu_1(\overline{(A^\eps)_{r_1}}) \geq 1/2$, and hence for any $r > 0$:
\[
\mu_1(\Omega_\eps \setminus (A^\eps)_{r_1+r}) \leq \mu_1(\Omega \setminus (A^\eps)_{r_1+r}) \leq \K_1(r) ~.
\]
Using (\ref{eq:M-cond}) again, we obtain:
\[
\mu_2(\Omega_\eps \setminus A_{r_1+r}) \leq \mu_2(\Omega_\eps \setminus (A^\eps)_{r_1+r}) \leq M(\eps) \K_1(r) ~,
\]
and thus:
\[
\mu_2(\Omega \setminus A_{r_1+r}) \leq \mu_2(\Omega \setminus \Omega_\eps) + \mu_2(\Omega_\eps \setminus A_{r_1+r}) \leq \eps + M(\eps) \K_1(r) ~.
\]
It follows that for all $\eps\in (0,1/4]$, 
\begin{equation} \label{eq:conc-trans-1}
 \K_2\brac{r + \K_1^{-1}\brac{\frac{1/2 - \eps}{M(\eps)}}} \leq \eps + M(\eps) \K_1(r) \;\;\;\; \forall r > 0 ~,~ \forall \eps \in (0,1/4] ~.
\end{equation}

Now given $r \geq 2 \K_1^{-1}(\beta(1/4))$, we set $\eps_0 := \beta^{-1}(\K_1(r/2)) \in [0,1/4]$, ensuring (by the right-continuity of $\beta$) that $\eps_0 / M(\eps_0) \geq \K_1(r/2)$. Since:
\[
\K_1^{-1}\brac{\frac{1/2 - \eps_0}{M(\eps_0)}} \leq \K_1^{-1}\brac{\frac{\eps_0}{M(\eps_0)}} \leq \frac{r}{2} ~,
\]
we obtain from (\ref{eq:conc-trans-1}):
\[
\K_2(r) \leq \eps_0 + M(\eps_0) \K_1\brac{r - \K_1^{-1}\brac{\frac{1/2 - \eps_0}{M(\eps_0)}}} \leq \eps_0 + M(\eps_0) \K_1(r/2) \leq 2 \eps_0 ~,
\]
and (\ref{eq:conc-trans-2}) follows.
\end{proof}

When $M(\eps) = M$ for all $\eps \in (0,1/4]$, that is when $d\mu_2/d\mu_1$ is essentially upper bounded, 
 the above concentration transference result  appeared in \cite[Lemma 3.1]{EMilmanGeometricApproachPartII} with better numerical constants. However the previous lemma adapts to cases when the density is unbounded.

\begin{prop}[Concentration Transference - Integral Form] \label{prop:conc-trans-integral}
Let $\mu_1$, $\mu_2$ be two probability measures on a common metric space $(\Omega,d)$, verifying  $\mu_2 \ll \mu_1$. Let $\K_i = \K_{(\Omega,d,\mu_i)}$, $i\in\{1,2\}$ be their concentration profiles. Let $G : \Real_+ \rightarrow \Real_+$ denote a continuous function increasing to infinity, and set $F(x) = x G(x)$ on $\Real_+$. Assume that:
\begin{equation} \label{eq:F-integral}
\int F\brac{\frac{d\mu_2}{d\mu_1}} d\mu_1 = \int G\brac{\frac{d\mu_2}{d\mu_1}} d\mu_2 \leq L < \infty ~.
\end{equation}
Then:
\[
\K_2(r) \leq 2 \K_1(r/2) F^{-1}\big(L / \K_1(r/2)\big) \;\;\; \forall  r > 0 ~.
\]
\end{prop}
\begin{proof}
By the Markov-Chebyshev inequality:
\[
\mu_2\set{x \in \Omega \; ; \;  G\brac{\frac{d\mu_2}{d\mu_1}(x)} \geq \frac{L}{\eps}} \leq \frac{\int G\brac{\frac{d\mu_2}{d\mu_1}} d\mu_2}{L / \eps}  \leq \eps ~.
\]
It follows that the hypothesis of Lemma \ref{lem:conc-trans} is satisfied with $M(\eps) := G^{-1}(L/\eps)$. Observe that:
\[
F(M(\eps)) = M(\eps) G(M(\eps)) = \frac{L}{\beta(\eps)} ~,
\]
where as in the lemma, $\beta(\eps) := \eps / M(\eps)$. Consequently:
\[
\beta^{-1}(x) = M^{-1}(F^{-1}(L/x)) = \frac{L}{G(F^{-1}(L/x))} = x F^{-1}(L/x) ~,
\]
and the conclusion of Lemma \ref{lem:conc-trans} yields the desired assertion.
\end{proof}

\subsection{Concentration Vs. Isoperimetry}

In the ``semi-convex setting", when an additional lower-bound assumption on an appropriate (generalized Ricci) curvature associated to the target space $(\Omega,d,\mu_2)$ is imposed, it is possible, repeating the program put forth in \cite{EMilmanGeometricApproachPartII}, to utilize the results of \cite{EMilman-RoleOfConvexity,EMilmanGeometricApproachPartI} on the equivalence between isoperimetric and concentration inequalities in that setting, and translate the transference principle obtained above from the concentration to the isoperimetric level. As a by-product, a transference principle may be deduced for Sobolev-type inequalities in the semi-convex setting. We will only briefly describe this program here for the log-Sobolev and spectral-gap inequalities, as these seems the most interesting and relevant ones for applications, and refer to \cite{EMilmanGeometricApproachPartII} for results on the transference (or stability) of general isoperimetric and Sobolev-type inequalities and for missing details.

\medskip

Recall that Minkowski's (exterior) boundary measure of a
Borel set $A \subset \Omega$, which we denote here by $\mu^+(A)$, is
defined as $\mu^+(A) := \liminf_{\eps \to 0} \frac{\mu(A^d_{\eps}) -
\mu(A)}{\eps}$, where $A_{\eps}=A^d_{\eps} := \set{x \in \Omega ; \exists y
\in A \;\; d(x,y) < \eps}$ denotes the $\eps$ extension of $A$ with
respect to the metric $d$. The isoperimetric profile $\I =
\I_{(\Omega,d,\mu)}$ is defined as the pointwise maximal function $\I
: [0,1] \rightarrow \Real_+$, so that $\mu^+(A) \geq \I(\mu(A))$, for
all Borel sets $A \subset \Omega$. An isoperimetric inequality measures the relation between the boundary measure and the measure of a set, by providing a lower bound on $\I_{(\Omega,d,\mu)}$. Since $A$ and $\Omega \setminus A$ will typically have the same boundary measure, it will be convenient to also define $\tilde{\I} :[0,1/2] \rightarrow \Real_+$ as $\tilde{\I}(v) := \min(\I(v),\I(1-v))$.

The two main differences between isoperimetric and concentration inequalities are that the latter ones only measure the concentration around sets having measure $1/2$, and only provide large-deviation information on the measure of extensions of these sets, contrary to the infinitesimal information provided by the former ones. We refer to \cite{Ledoux-Book,MilmanSurveyOnConcentrationOfMeasure1988,EMilmanGeometricApproachPartII} for a wider exposition on these and related topics and for various applications.

It is known and easy to see that an isoperimetric inequality always implies a concentration inequality, simply by ``integrating'' along the isoperimetric differential inequality (see e.g. \cite{EMilmanGeometricApproachPartII}), but the converse implication is in general false, due to the possible existence of narrow ``necks'' in the geometry of the space $(\Omega,d)$ or the measure $\mu$. However, when such necks are ruled out by imposing some semi-convexity assumptions on the geometry and measure in the Riemannian-manifold-with-density setting, it was shown in the second named author's previous work \cite{EMilmanGeometricApproachPartI} (see also \cite{EMilmanSharpIsopInqsForCDD} and \cite{LedouxConcentrationToIsoperimetryUsingSemiGroups}) that general concentration inequalities imply back their isoperimetric counterparts, with quantitative estimates which \emph{do not depend} on the dimension of the underlying manifold. Semi-convexity thus serves as a bridge between the large-deviation and infinitesimal extension scales. The precise formulation is as follows.

\begin{dfn*} \label{def:CA}
We will say that our \emph{smooth $\kappa$-semi-convexity assumptions} are satisfied $(\kappa \geq 0)$ if:
\begin{itemize}
\item $(\Omega,d)$ is given by a complete smooth oriented Riemannian manifold $(M,g)$ with its induced geodesic distance.
\item $\mu$ is supported on the closure of a geodesically convex domain $S \subset M$ with (possibly empty) $C^2$ smooth boundary, on which $d\mu = \exp(-\psi) dvol_M|_S$ with $\psi \in C^2(\overline{S})$, and as tensor fields on $S$:
\[
  \mathrm{Ric}_g +  \mathrm{Hess}_g \psi \geq -\kappa g ~.
\]
\end{itemize}
We will say that our \emph{$\kappa$-semi-convexity assumptions} are satisfied if $\mu$ can be approximated in total-variation by measures $\set{\mu_m}$ so that each $(\Omega,d,\mu_m)$ satisfies our smooth $\kappa$-semi-convexity assumptions. \\
When $\kappa=0$, we will say that our \emph{convexity assumptions} are satisfied.
\end{dfn*}

\smallskip

Here $ \mathrm{Ric}_g$ denotes the Ricci curvature tensor of $(M,g)$, $ \mathrm{Hess}_g$ denotes the Riemannian Hessian, and $vol_M$ denotes the Riemannian volume form. $ \mathrm{Ric}_g +  \mathrm{Hess}_g \psi$ is the well-known Bakry--\'Emery curvature tensor, introduced in \cite{Lichnerowicz1970GenRicciTensorCRAS} and developed in \cite{BakryEmery} (in the more abstract framework of diffusion generators), which incorporates the curvature from both the geometry of $(M,g)$ and the measure $\mu$. When $\psi$ is sufficiently smooth and $S=M$, our $\kappa$-semi-convexity assumption is then precisely the Curvature-Dimension condition $CD(-\kappa,\infty)$ (see \cite{BakryEmery}). An important example to keep in mind is that of Euclidean space $(\Real^n,\abs{\cdot})$ equipped with a probability measure $\exp(-\psi(x)) dx$ with $ \mathrm{Hess} \; \psi \geq -\kappa Id$.

\begin{thm}[\cite{EMilmanGeometricApproachPartI}] \label{thm:main-equiv}
Let $\kappa \geq 0$ and let $\alpha: \Real_+ \rightarrow \Real \cup \set{+\infty}$ denote an increasing continuous function so that:
\begin{equation} \label{eq:alpha-cond}
\exists \delta_0 > 1/2 \;\;\; \exists r_0 \geq 0 \;\;\; \forall r \geq r_0 \;\;\; \alpha(r) \geq \delta_0 \kappa r^2 ~.
\end{equation}
Then under our $\kappa$-semi-convexity assumptions, the concentration inequality:
\[
\K(r) \leq \exp(-\alpha(r)) \;\;\; \forall r > 0 ~,
\]
implies the following isoperimetric inequality:
\begin{equation} \label{eq:i-2}
\tilde{\I}(v) \geq \min(c_{\delta_0} \; v \gamma(\log 1/v),c_{\kappa,\alpha}) \;\;\; \forall v \in [0,1/2] \;\;\; , \;\;\; \textrm{where} \;\; \gamma(x) = \frac{x}{\alpha^{-1}(x)} ~,
\end{equation}
and $c_{\delta_0},c_{\kappa,\alpha}>0$ are constants depending solely on their arguments. Moreover, if $\kappa = 0$, we may take $c_{\delta_0} = c$ and $c_{0,\alpha} = \frac{c}{4} \gamma(\log 4)$ for some universal constant $c>0$. If $\kappa > 0$, the dependence of $c_{\kappa,\alpha}$ on $\alpha$ may be expressed only via $\delta_0$ and $\alpha(r_0)$.
\end{thm}

\subsection{log-Sobolev Transference under Curvature Lower Bound}

\begin{thm}[log-Sobolev Transference under Curvature Lower Bound] \label{thm:log-Sob-trans}
Let $\mu_2 \ll \mu_1$ denote two Borel probability measures on a common Riemannian manifold $(M,g)$,
and assume that $(M,g,\mu_2)$ satisfies our $\kappa$-semi-convexity assumptions ($\kappa \geq 0$).
Assume that $(M,g,\mu_1)$ satisfies a strong-enough log-Sobolev inequality:
\begin{equation} \label{eq:ls-assump}
\rho  = \rho_{LS}(M,g,\mu_1) > \frac{4p}{p-1} \kappa ~,
\end{equation}
for some $p > 1$, and that:
\begin{equation} \label{eq:L-cond}
\int \brac{\frac{d\mu_2}{d\mu_1}}^p d\mu_1 \leq L^p ~.
\end{equation}
Then $(M,g,\mu_2)$ satisfies a log-Sobolev inequality:
\begin{equation} \label{eq:ls-concl}
\rho_{LS}(M,g,\mu_2) \geq C(\rho,\kappa,L,p)  ~,
\end{equation}
where:
\[
C(\rho,\kappa,L,p) := c \; \rho \; \frac{p-1}{p} \; \exp(-C (1 + \log(L)) / \theta) ~~,~~ \theta := 1 - \frac{4p \kappa}{(p-1) \rho} ~,
\]
and $c,C>0$ are universal constants. Moreover, when $\kappa = 0$, one may in fact use:
\[
C(\rho,0,L,p) = c \; \rho \; \frac{p-1}{p} \frac{1}{1+ \log(L)} ~.
\]
\end{thm}

\begin{proof}
Let us denote the isoperimetric and concentration profiles on the corresponding spaces by $\I_i = \I_i(M,g,\mu_i)$ and $\K_i = \K_i(M,g,\mu_i)$, $i=1,2$, respectively.
By the Herbst argument (see Appendix), it is known that the log-Sobolev inequality (\ref{eq:ls-assump}) implies the following Laplace-functional inequality:
\[
 \int \exp(\lambda f) d\mu_1 \leq \exp(\lambda^2/(2\rho)) \;\;\; \forall \lambda \geq 0 \;\; \forall \text{ $1$-Lipschitz $f$ s.t. } \int f d\mu_1 = 0 ~.
\]
By passing from expectation to median in the above requirement that $\int f d\mu_1 = 0$ (employing the Markov-Chebyshev inequality),  and from $1$-Lipschitz functions to their level-sets,  it is easy to check (see e.g. \cite[Lemma 4.2]{EMilmanGeometricApproachPartII}) that this implies the following concentration inequality on $(M,g,\mu_1)$:
\[
\K_1(r) \leq \exp\brac{ - \frac{\rho}{2} \brac{r - \sqrt{\frac{2 \log 2}{\rho}}}_+^2 } \;\;\; \forall r > 0 ~.
\]
Using Proposition \ref{prop:conc-trans-integral}, we deduce that:
\[
\K_2(r) \leq 2 L \K_1(r/2)^{1-1/p} \leq \exp\brac{ \log (2L)  - \frac{\rho(p-1)}{8p} \brac{r - 2 \sqrt{\frac{2 \log 2}{\rho}}}_+^2 } \;\;\; \forall r > 0 ~.
\]
Since $\delta_0 := \frac{\rho(p-1)}{8p} / \kappa > 1/2$ by assumption, $\K_2$ clearly satisfies the decay condition (\ref{eq:alpha-cond}) required to apply Theorem \ref{thm:main-equiv}:
\[
 \exists \delta'_0 := \frac{1}{2}(\delta_0 + 1/2) > \frac{1}{2} \;\;\; \exists r'_0 = r'_0(\rho,\kappa,L) \;\;\; \forall r \geq r'_0 \;\;\; \K_2(r) \leq \exp(-\delta'_0 \kappa r^2) ~.
\]
Consequently, Theorem \ref{thm:main-equiv} implies that the following isoperimetric inequality is satisfied:
\begin{eqnarray*}
\tilde{\I}_2(v) & \geq & \min\brac{c(\delta_0') \sqrt{\rho} v \frac{\log 1/v}{\sqrt{\frac{8p}{p-1}} \sqrt{\log 1/v + \log (2 L)} + 2 \sqrt{2 \log 2}} , c_{\rho,\kappa,L}} \\
& \geq & c'(\rho,\kappa,L,p) v \sqrt{\log 1/v} \;\;\;\;\;\; \forall v \in [0,1/2] ~.
\end{eqnarray*}
(note that by Jensen's inequality applied to (\ref{eq:L-cond}), $L \geq 1$).
This means that $(M,g,\mu_2)$ satisfies a Gaussian isoperimetric inequality. As described in \cite{EMilmanGeometricApproachPartII} in greater detail, it is known by a result of M. Ledoux \cite{LedouxBusersTheorem}, refined by B. Beckner (see \cite{LedouxLectureNotesOnDiffusion}), that this implies the log-Sobolev inequality (\ref{eq:ls-concl}) with
$C(\rho,\kappa,L,p) = c (c'(\rho,\kappa,L,p))^2$, with $c>0$ a universal constant, concluding the proof. Note that when $\kappa = 0$, the remarks at the end of the formulation of Theorem \ref{thm:main-equiv} imply that one may use:
\[
c'(\rho,0,L,p) = c \; \sqrt{\rho} \sqrt{\frac{p-1}{p}} \sqrt{\frac{1}{1+ \log(L)}}
\]
above, with $c>0$ a universal constant. When $\kappa > 0$, a tedious inspection of the estimates obtained in the proof of Theorem 1.2 in Section 5 of \cite{EMilmanGeometricApproachPartI} verifies the asserted estimate on $C(\rho,\kappa,L,p)$, concluding the proof.
\end{proof}

\subsection{Spectral-Gap Transference under Non-Negative Curvature} \label{subsec:SG-trans}

First, we state   transference  principles  which were obtained by the second named author. In their most
general form they involve total variation estimates. Recall that   the total-variation distance $d_{TV}$ between two absolutely continuous probability measures $\mu_1,\mu_2$ on $(M,g)$ is defined as:
\[
d_{TV}(\mu_1,\mu_2) := \frac{1}{2} \int \abs{\frac{d\mu_1}{d\vol_M} - \frac{d\mu_2}{d\vol_M}} d\vol_M ~. 
\]
For notational simplicity, we formulate the results of \cite{EMilman-RoleOfConvexity}  in the Euclidean setting 
and  refer to that paper for the general case. Transference results  with respect to the $1$-Wasserstein and relative-entropy distances are proposed in \cite{EMilmanGeometricApproachPartII}.

\begin{thm}[\cite{EMilman-RoleOfConvexity}] \label{thm:old-SG-trans}
Let $\mu_1$ and $\mu_2$ denote two log-concave probability measures on Euclidean space $(\Real^n,|\cdot|)$, meaning that $\mu_i = \exp(-V_i(x))\, dx$ with $V_i : \Real^n \rightarrow \Real \cup \set{+\infty}$ convex. Assume that either of the following conditions is satisfied for some $L \in (1,+\infty)$:
\begin{enumerate}
\item $\mu_2 \ll \mu_1$ and $\norm{d\mu_2/d\mu_1}_{L^\infty} \leq L$; or,
\item $\mu_1 \ll \mu_2$ and $\norm{d\mu_1/d\mu_2}_{L^\infty} \leq L$; or,
\item $\displaystyle d_{TV}(\mu_1,\mu_2) \leq 1 - \frac1L$ .
\end{enumerate}
Then the $i$-th assumption ($i=1,2,3$) implies:
\[
\rho_{SG}(\Real^n,|\cdot|,\mu_2) \geq C_i(L)^2 \rho_{SG}(\Real^n,|\cdot|,\mu_1) ~,
\]
where, for some universal constant $c>0$:
\begin{equation} \label{eq:C-eps}
C_1(L) = \frac{c}{1+\log(L)} ~,~ C_2(L) = \frac{c}{L^2} ~,~ C_3(L) = \frac{c}{L^2(1+\log(L))} ~.
\end{equation}
\end{thm}

 In order to compare the above assumptions  it is worthwhile to record:
\begin{lem} \label{lem:weak-TV}
Let $\mu_2 \ll \mu_1 \ll\vol_M$ be Borel probability measures on a  Riemannian manifold $(M,g)$.
If $\displaystyle \mu_2\set{\frac{d\mu_2}{d\mu_1} \geq D} \leq \delta$, then 
$\displaystyle d_{TV}(\mu_1,\mu_2) \leq 1 - \frac{1-\delta}{\max(1,D)}$.
 
\end{lem}
\begin{proof}
\[
1 - d_{TV}(\mu_1,\mu_2) = \int \min\brac{\frac{d\mu_1}{d\vol_M},\frac{d\mu_2}{d\vol_M}} d\vol_M = \int \min\brac{\frac{1}{\frac{d\mu_2}{d\mu_1}},1} d\mu_2 \geq \frac{1-\delta}{\max(1,D)}~.
\]
\end{proof}
Consequently the hypothesis of Case 1  (or 2) implies the one of Case 3. However, 
applying Case 3 of Theorem \ref{thm:old-SG-trans} gives the poorest 
dependence in $L$ as $C_3(L)$ is much smaller than $C_1(L)$ and $C_2(L)$.
On the other hand, the total variation distance is robust and easy to control. For instance
if we assume that $\mu_2 \ll \mu_1$ are log-concave on $\Real^n$ and verify $\int (d\mu_2/d\mu_1)^p d\mu_1\le L^p$, then
  by the Markov-Chebyshev's inequality, for any $D > 0$:
\[
\mu_2\set{\frac{d\mu_2}{d\mu_1} \geq D} \leq \frac{1}{D^{p-1}} \int \brac{\frac{d\mu_2}{d\mu_1}}^p d\mu_1 \leq \frac{L^p}{D^{p-1}} ~.
\]
Consequently, we have by Lemma \ref{lem:weak-TV}:
\[
d_{TV}(\mu_1,\mu_2) \leq \inf_{D\ge 1}\left( 1 - \frac{1}{D} \left(1 - \frac{L^p}{D^{p-1}}\right)\right)= 1-\frac{p-1}{(pL)^{p/(p-1)}} ~.
\]
Therefore applying Case 3 of Theorem \ref{thm:old-SG-trans} yields that the spectral-gap of $\mu_2$
is at least that of $\mu_1$ divided by a power of $L$. This polynomial dependence would not be good
enough for the applications we have in mind. The goal of the next theorem is to provide
 a new quantitative transference principle for the spectral-gap, which generalizes Case 1 of Theorem \ref{thm:old-SG-trans} (corresponding to the case $p=\infty$ below):
\begin{thm}[Spectral-Gap Transference under Convexity Assumptions] \label{thm:SG-trans}
Let $\mu_2 \ll \mu_1$ denote two Borel probability measures on a common Riemannian manifold $(M,g)$,
and assume that $(M,g,\mu_2)$ satisfies our convexity assumptions. Assume that for some $p > 1$:
\begin{equation} \label{eq:L-cond-SG}
\int \brac{\frac{d\mu_2}{d\mu_1}}^p d\mu_1 \leq L^p ~.
\end{equation}
Then:
\[
\rho_{SG}(M,g,\mu_2) \geq C(L,p)^2 \rho_{SG}(M,g,\mu_1) ~,
\]
where:
\begin{equation} \label{eq:C-Lp}
C(L,p) := c \; \frac{p-1}{p} \; \frac{1}{1+ \log(L)} ~,
\end{equation}
and $c>0$ is a  universal constant.
\end{thm}

\begin{proof}[Proof of Theorem \ref{thm:SG-trans}]
Set $\rho = \rho_{SG}(M,g,\mu_1)$, $\K_1 = \K(M,g,\mu_1)$ and $\K_2 = \K(M,g,\mu_2)$. By a result of M. Gromov and V. Milman \cite{GromovMilmanLevyFamilies} (see also \cite[Corollary 2.7]{EMilman-RoleOfConvexity}), a spectral-gap inequality always implies the following exponential concentration:
\begin{equation}\label{eq:gromowmilman}
\K_1(r) \leq \exp\brac{ - c \sqrt{\rho} \,r} \;\;\; \forall r > 0 ~,
\end{equation}
where $c>0$ is a universal constant. Using Proposition \ref{prop:conc-trans-integral}, we deduce from (\ref{eq:L-cond-SG}) that:
\[
\K_2(r) \leq 2 L \K_1(r/2)^{1-1/p} \leq 2 L \exp\brac{- c \frac{p-1}{2 p} \sqrt{\rho}\,  r } \;\;\; \forall r > 0 ~.
\]
Using our convexity assumptions on $(M,g,\mu_2)$, Theorem \ref{thm:main-equiv} implies that the following isoperimetric inequality is satisfied:
\begin{eqnarray*}
\tilde{\I}(M,g,\mu_2)(v) & \geq & c \sqrt{\rho} \,\frac{p-1}{p} \min\brac{ v \frac{\log 1/v}{\log 1/v + \log (2 L)}, c_2} \\
& \geq & c' \sqrt{\rho}\,  \frac{p-1}{p} \frac{1}{1 + \log(L)} v \;\;\;\;\;\; \forall v \in [0,1/2]
\end{eqnarray*}
(note again that $L \geq 1$ by Jensen's inequality). This means that $(M,g,\mu_2)$ satisfies a linear (or Cheeger-type) isoperimetric inequality. As described in \cite{EMilmanGeometricApproachPartII} in greater detail, it is known by results of Maz'ya \cite{MazyaSobolevImbedding,MazyaCapacities} and independently Cheeger \cite{CheegerInq} that this implies a spectral-gap inequality:
\begin{equation} \label{eq:CheAgain}
\sqrt{\rho_{SG}(M,g,\mu_2)} \geq \frac{1}{2} \inf_{v \in (0,1/2]} \frac{\tilde{\I}(M,g,\mu_2)(v)}{v} \geq \frac{c'}{2} \frac{p-1}{p} \frac{1}{1 + \log(L)} \sqrt{\rho} ~,
\end{equation}
as asserted.
\end{proof}

We can actually prove a result stronger than Theorem~\ref{thm:SG-trans}, using stronger results from \cite{EMilman-RoleOfConvexity,EMilmanGeometricApproachPartI,EMilmanIsoperimetricBoundsOnManifolds}.
However, in practice a convenient way to check the hypothesis of the next theorem would be to 
verify the condition of the previous one.
\begin{thm}\label{thm:SG-trans2}
Let $\mu_2 \ll \mu_1$ denote two Borel probability measures on a common Riemannian manifold $(M,g)$,
and assume that $(M,g,\mu_2)$ satisfies our convexity assumptions. Let $L>0$. There is a universal constant $c>0$ such that
 \[
\mu_2 \left\{\frac{d\mu_2}{d\mu_1}> L\right\} \le \frac18 \quad \Longrightarrow 
\quad
\rho_{SG}(M,g,\mu_2) \geq \frac{c}{\log(8L)^2}\, \rho_{SG}(M,g,\mu_1) ~.
\]
\end{thm}
\begin{proof}
Let us adopt the notations of the previous proof. First note that our hypothesis ensures $L\ge 7/8$.
Next we apply \eqref{eq:conc-trans-1} for $\eps=\frac18$ and $M(\eps)=L$, to get for any $r>0$:
$$ 
\K_2\left(r+\K_1^{-1}\brac{\frac{3}{8L}} \right)\le \frac18+L \K_1(r) ~ .
$$
Applying this to $r = \log(8L)/(c\sqrt\rho)$ and using exponential concentration for $\mu_1$ in the form of \eqref{eq:gromowmilman}, we obtain:
$$  
\K_2 (r_2) \leq \frac14 ~,~ r_2 := \frac{1}{c\sqrt\rho} \log(8L)+\K_1^{-1}\brac{\frac{3}{8L}} \leq \frac{2}{c \sqrt{\rho}} \log(8L)  ~.
$$
At this point, we use a much stronger result than Theorem \ref{thm:main-equiv} for transferring concentration to linear isoperimetry when the convexity assumptions are satisfied, namely (see \cite{EMilman-RoleOfConvexity,EMilmanGeometricApproachPartI,EMilmanIsoperimetricBoundsOnManifolds}):
\[
\inf_{v \in (0,1/2]} \frac{\tilde{\I}(M,g,\mu_2)(v)}{v} \geq \sup_{r > 0} \frac{1-2\K_2(r)}{r} ~.
\]
Applying this inequality for $r_2$, and arguing as in (\ref{eq:CheAgain}), we obtain:
$$
\sqrt{\rho_{SG}(M,g,\mu_2)} \geq \frac{1}{2} \inf_{v \in (0,1/2]} \frac{\tilde{\I}(M,g,\mu_2)(v)}{v} \geq \frac{1}{4 r_2} \geq 
 \frac{c\sqrt\rho}{8 \log(8L)} ~.
$$
\end{proof}

\section{The Conservative Spin Model} \label{sec:model}

Let $\mu = \exp(-V(x)) dx$ denote a probability measure on $\Real$, and let $X_0$ be a random variable on $\Real$ with law $\mu$.
Let $\set{X_i}_{i=1,\ldots,n}$ be independent copies of $X_0$, so that $X = (X_1,\ldots,X_n)$ is a random-vector in $\Real^n$ distributed according to:
\[
\mu_n := \mu^{\otimes n} = \exp(-H(x)) \,dx ~,
\]
where $H(x)$ denotes the non-interacting Hamiltonian:
\[
H(x) = \sum_{i=1}^n V(x_i)~.
\]
The measure $\mu_n$ is the Gibbs measure corresponding to the grand canonical ensemble of non-interacting spins. We obtain from it the measure $\mu_E$ corresponding to the canonical ensemble, by conditioning the mean-spin $S = \frac{1}{n} \sum_{i=1}^n X_i$ to be $0$, namely:
\[
\mu_E = \frac{1}{Z_E} \exp(-H(x)) \, d\vol_E(x) ~,
\]
where we equip $\Real^n$ with its standard Euclidean structure $|\cdot|$, denote by $\vol_E$ the induced Lebesgue measure on the hyperplane $E = E_0$ given by $\sum_{i=1}^n x_i = 0$, and denote by $Z_E > 0$ a normalization term. Note that $Z_E$ is the density at $0$ of $\mu^D$, the push-forward of $\mu_n$ by the orthogonal projection on the diagonal line $D$ spanned by the vector $(1/\sqrt{n},\ldots,1/\sqrt{n})$. We denote by $\pi_F$ the orthogonal projection onto the subspace $F$.

Note that $\mu_E$ is not absolutely continuous with respect to $\mu_n$, so we define $\mu_{E,w}$ by ``thickening" it uniformly in the diagonal direction by a width of $w > 0$ from each side:
\[
d\mu_{E,w}(x) = \frac{1}{Z_{E,w}} \exp\brac{-H(\pi_E(x))} 1_{|\pi_D(x)| \leq w}\,  dx ~,~ Z_{E,w} = 2 w Z_E ~.
\]

Recall from the Introduction that we denote $M_p = M_p(\mu) := \E |X_0|^p$ and set 
$D_{1,\Psi_1} = D_{1,\Psi_1}(\mu) := \norm{V'(X_0)}_{L_{\Psi_1}}$, where $\norm{Y}_{L_{\Psi_1}}$ denotes the $\Psi_1$ norm of a random-variable $Y$.
Given $\delta>0$, we set $Y_0 := \sup_{\xi \in [X_0-\delta,X_0+\delta]} |V''(\xi)|$, and denote $D_{2,\Psi_1} = D^\delta_{2,\Psi_1} = D^\delta_{2,\Psi_1}(\mu) := \norm{Y_0}_{L_{\Psi_1}}$ and $D_2 = D^\delta_2 = D^\delta_2(\mu) := \E(Y_0)$.
 
 \medskip
 
 A central tool we will use in this section is the following large-deviation bound of Bernstein-type (e.g. \cite{LedouxTalagrand-Book}).
\begin{thm}[Bernstein] \label{thm:Bern}
Let $Y_1,\ldots,Y_n$ denote a sequence of independent random-variables on $\Real$, with $\E(Y_i) = 0$ and $\norm{Y_i}_{L_{\Psi_1}} \leq D_{\Psi_1}$ for each $i=1,\ldots,n$. Then there exists a universal numeric constant $c>0$, so that for any vector $a = (a_1,\ldots,a_n)$ and $t > 0$:
\[
\P(\sum a_i Y_i \geq t) \leq \exp\brac{-c \min \brac{ \frac{t^2}{\norm{a}_2^2 D_{\Psi_1}^2} , \frac{t}{\norm{a}_\infty D_{\Psi_1}} } } ~.
\]
\end{thm}

\subsection{Concentration Transference from $\mu_n$ to $\mu_{E,w}$ - One-Sided Approach} \label{subsec:model-approach-2}

In this subsection, we present a first way to transfer concentration estimates from $\mu_n$ to $\mu_{E,w}$. It  only requires controlling the second derivative of $V$ from one side.  To this end, we need the following \emph{local} version of the Berry--Esseen Theorem (see \cite{Petrov-SumsOfIndependentRVsBook}), and in fact, only at the center point $x = 0$:

\begin{thm} \label{thm:local-BE}
Let $X_1,\ldots,X_n$ denote i.i.d. copies of the random variable $X_0$, having absolutely-continuous law $\mu$. Assume that $\E(X_0) = 0$, that $M_3 = M_3(\mu) < \infty$ and that $\lambda := \norm{d\mu/dx}_{L^\infty} < \infty$. Then for any $n \geq 1$, the density $p_n(x)$ of $\frac{1}{\sqrt{n}} \sum_{i=1}^n X_i$ satisfies:
\[
\sup_{x \in \Real} \abs{p_n(x) - \frac{1}{\sqrt{2 \pi M_2}} \exp\brac{-\frac{x^2}{2 M_2}} } \leq \frac{C}{\sqrt{M_2}} \max\brac{\frac{M_3}{M_2^{3/2}} , M_3 \lambda^3} \frac{1}{\sqrt{n}} ~,
\]
for some universal constant $C>0$. 
\end{thm}

\begin{prop} \label{prop:conc-trans-model-2}
Let $p \in \Real$, and assume that:
\begin{itemize}
\item $V \in C^2(\Real)$ and $sign(p) V'' \geq -\kappa$, for some $\kappa \geq 0$.  
\item $D_{1,\Psi_1} = D_{1,\Psi_1}(\mu) := \norm{V'(X_0)}_{L_{\Psi_1}}  < \infty$.
\end{itemize}
Then for any $w > 0$ and integer $n$ greater than $\eta = \eta^{|p| w D_{1,\Psi_1}}$, the following estimate holds:
\[
\int_{\Real^n} \exp\brac{p (H(x) - H(\pi_E(x))} 1_{\abs{\pi_D(x)} \leq w} \, d\mu_n(x) \leq  C \exp\brac{C w^2 (|p| \kappa +  p^2 D_{1,\Psi_1}^2)} ~,
\]
for some universal constant $C > 1$.
\end{prop}

\begin{proof}
Assume that $p > 0$,  it will be evident from the proof that the case that $p<0$ is treated identically, after exchanging $H(X)$ and $H(\pi_E(x))$. We write:
\begin{eqnarray}
\nonumber & &  \int_{\Real^n} \exp(p (H(x) - H(\pi_E(x)))) 1_{\abs{\pi_D(x)} \leq w} \, d\mu_n(x) \\
\nonumber & = & \int_{-\infty}^{\infty} p \exp(pu) \mu_n\set{x \in \Real^n \; ; \; \abs{\pi_D(x)} \leq w ~,~ H(x) - H(\pi_E(x)) \geq u} du \\
\label{eq:integral} & \leq & e^{p u_0} + \int_{u_0}^{\infty} p e^{pu} \mu_n\set{x \in \Real^n \; ; \; \abs{\pi_D(x)} \leq w ~,~ H(x) - H(\pi_E(x)) \geq u} du\,,
\end{eqnarray}
for some $u_0$ to be determined. We proceed by roughly evaluating the integrand as follows:
\begin{eqnarray*}
&   & \mu_n\set{x \in \Real^n \, ; \; \abs{\pi_D(x)} \leq w ~,~ H(x) - H(\pi_E(x)) \geq u} \\
& = & \mu_n \set{x \in \Real^n \, ; \; \abs{\pi_D(x)} \leq w ~,~ \sum_{i=1}^n V(x_i) - \sum_{i=1}^n V(\pi_E(x)_i) \geq u } \\
& \leq & \mu_n \set{x \in \Real^n \, ; \; \exists t \in [-w,w] ~,~ \sum_{i=1}^n V(x_i) - \sum_{i=1}^n V(x_i + t/\sqrt{n}) \geq u} ~.
\end{eqnarray*}
Applying Taylor's theorem and using in addition that $V'' \geq -\kappa$, we deduce that:
\begin{eqnarray}
\nonumber &   & \mu_n\set{x \in \Real^n \; ; \; \abs{\pi_D(x)} \leq w ~,~ H(x) - H(\pi_E(x)) \geq u} \\
\nonumber & \leq & \mu_n \set{x \in \Real^n \; ; \; \exists t \in [-w,w] ~,~ -\frac{t}{\sqrt{n}} \sum_{i=1}^n V'(x_i) + \frac{t^2}{2} \kappa \geq u } \\
\label{eq:one-term} & \leq & \P \brac{\frac{w}{\sqrt{n}} \abs{\sum_{i=1}^n V'(X_i)} \geq u - \frac{w^2}{2} \kappa } ~.
\end{eqnarray}
It is easy to verify that $\E V'(X_0) = \int_{-\infty}^\infty V'(x) \exp(-V(x)) \,dx = -\int_{-\infty}^\infty d\exp(-V(x)) = 0$.
Consequently, whenever $u > u_0$ and $n$ tends to infinity, the expression in (\ref{eq:one-term}) will be governed by the Central-Limit Theorem, in accordance with the heuristic argument from the Introduction. To obtain exponentially decaying estimates on (\ref{eq:one-term}) as $u \rightarrow \infty$ for each individual $n$, we apply Theorem \ref{thm:Bern} , which together with (\ref{eq:integral}) yields:
\begin{eqnarray*}
& & \int_{\Real^n} \exp\brac{p (H(x) - H(\pi_E(x))} 1_{\abs{\pi_D(x)} \leq w} d\mu_n(x) \\
& \leq & \exp(p u_0) + 2 p \exp(p u_0) \int_{0}^\infty \exp\brac{pu - c \min\brac{ \frac{\sqrt{n}\, u}{w D_{1,\Psi_1}} , \frac{u^2}{w^2 D_{1,\Psi_1}^2} } } du ~.
\end{eqnarray*}
When $n$ is large enough, the above estimate is clearly bounded by:
\begin{eqnarray*}
&\leq & 2 \exp(p u_0) \brac{1 + p \int_{0}^\infty \exp\brac{p u - c \frac{u^2}{w^2 D_{1,\Psi_1}^2}} du} \\
&\leq & 2 \exp(p u_0) \brac{1 + p \sqrt{\pi/c}\,  w D_{1,\Psi_1} \exp(p^2 w^2 D_{1,\Psi_1}^2 / (4 c)) } ~.
\end{eqnarray*}
Recalling the definition of $u_0$ and using that $1 + \exp(x^2) x \leq C \exp(C x^2)$ for an appropriately chosen $C>1$, we obtain that the latter quantity does not exceed:
\[ 
C' \exp(C' w^2 (p \kappa + p^2 D_{1,\Psi_1}^2)) ~, 
\]
as asserted.
\end{proof}

\begin{cor} \label{cor:conc-trans-model-2}
In addition to the assumptions of Proposition \ref{prop:conc-trans-model-2} (for positive $p$), assume that the barycenter of $\mu$ is at the origin, that $M_3 = M_3(\mu) < \infty$, and that $\lambda := \norm{d\mu/dx}_{L^\infty} < \infty$. Then setting:
\begin{equation} w_0  := \sqrt{\min\brac{M_2,\frac{1}{\kappa + D^2_{1,\Psi_1}}}} ~,
\end{equation}
the following estimate holds:
\[
\brac{\int_{\Real^n} \brac{\frac{d\mu_{E,w_0}}{d\mu_n}}^4 d\mu_n}^{1/4} \leq C \max(1,\sqrt{M_2 (\kappa + D^2_{1,\Psi_1})}) ~,
\]
for any integer $n$ greater than $\eta^{M_3/M_2^{3/2} , M_3 \lambda^3}$, where $C>0$ is a universal constant.
\end{cor}
\begin{proof}
Recall that:
\[
d\mu_{E,w}(x) = \frac{1}{Z_{E,w}} \exp\brac{-H(\pi_E(x))} 1_{|\pi_D(x)| \leq w} \, dx ~.
\]
Consequently, applying Proposition \ref{prop:conc-trans-model-2} with $p=4$, we know that for any integer $n$ greater than $\eta^{w D_{1,\Psi_1}}$:
\begin{equation} \label{eq:cor-calc-2}
\brac{\int_{\Real^n} \brac{\frac{d\mu_{E,w}}{d\mu_n}}^4 d\mu_n}^{1/4} \leq \frac{C}{Z_{E,w}} \exp(C w^2 (\kappa + D_{1,\Psi_1}^2)) ~.
\end{equation}
Recall that $Z_{E,w} = 2 w Z_E$ and that $Z_E$ is the density at $0$ of $\frac{1}{\sqrt{n}} \sum_{i=1}^n X_i$. It follows by Theorem \ref{thm:local-BE} that  $Z_E \geq \frac{1}{2 \sqrt{\pi M_2}}$ whenever:
\[
\sqrt{n} \geq C' \max\brac{\frac{M_3}{M_2^{3/2}} , M_3 \lambda^3} ~,
\]
for an appropriately chosen universal constant $C' > 0$. Plugging this estimate into (\ref{eq:cor-calc-2}), we obtain:
\[
\brac{\int_{\Real^n} \brac{\frac{d\mu_{E,w}}{d\mu_n}}^4 d\mu_n}^{1/4} \leq C'' \frac{\sqrt{M_2}}{w} \exp(C w^2 (\kappa + D_{1,\Psi_1}^2)) ~,
\]
and the asserted estimate follows when setting $w = w_0$. 
\end{proof}

\subsection{Concentration Transference from $\mu_n$ to $\mu_{E,w}$ - Two-Sided Approach} \label{subsec:model-approach-1}

In this subsection, we present a variant of the procedure carried out above.
It  only requires the following version of the classical Berry--Esseen Theorem (e.g. \cite{Petrov-SumsOfIndependentRVsBook}):

\begin{thm}[Berry--Esseen] \label{thm:BE}
Let $X_1,\ldots,X_n$ denote a sequence of independent random-variables on $\Real$, with $\E(X_i) = 0$ and $\E(|X_i|^3) / E(X_i^2)^{3/2} \leq M$ for each $i=1,\ldots,n$.
Let $Z$ denote a standard Gaussian random variable on $\Real$ and denote:
\[
S_n = \frac{\sum_{i=1}^n X_i}{\sqrt{\sum_{i=1}^n \E(X_i^2)}} ~.
\]
Then:
\[
\forall t \in \Real \;\;\; \abs{\P(S_n \leq t) - \P(Z \leq t)} \leq \frac{C M}{\sqrt{n}} ~,
\]
for some universal constant $C>0$.
\end{thm}

\begin{prop} \label{prop:conc-trans-model}
Let $\mu = \exp(-V(x)) dx$ denote a probability measure on $\Real$, and assume that:
\begin{itemize}
\item $V \in C^2(\Real)$ and $\lim_{x \rightarrow \pm \infty} V(x) = +\infty$.
\item $D_{1,\Psi_1} = D_{1,\Psi_1}(\mu) < \infty$, and there exists $\delta>0$ so that $D_{2,\Psi_1} = D^\delta_{2,\Psi_1}(\mu) < \infty$.
\end{itemize}
Then for any $p \in \Real$, $w > 0$, and integer $n$ greater than $\eta = \eta^{w D_{1,\Psi_1}, w^2 D_{2,\Psi_1},w,|p|}_{\delta}$, the following estimate holds:
\[
\int_{\Real^n} \exp\brac{p \big(H(x) - H(\pi_E(x))\big)} 1_{\abs{\pi_D(x)} \leq w} \,d\mu_n(x) \leq  C \exp\brac{C w^2 (|p| D_2 + p^2 D_{1,\Psi_1}^2)} ~,
\]
for some universal constant $C > 1$.
\end{prop}
\begin{proof}
Let us assume that $p > 0$, it will be evident from the proof that the case that $p<0$ is treated identically, after exchanging $H(x)$ and $H(\pi_E(x))$. Repeating verbatim the relevant parts of the proof of Proposition \ref{prop:conc-trans-model-2}, we verify that:
\begin{eqnarray}
\nonumber & &  \int_{\Real^n} \exp(p (H(x) - H(\pi_E(x)))) 1_{\abs{\pi_D(x)} \leq w}\,  d\mu_n(x) \\
\label{eq:integral-2} & \leq & e^{p u_0} + \int_{u_0}^{\infty} p e^{pu} \mu_n\set{x \in \Real^n \; ; \; \abs{\pi_D(x)} \leq w ~,~ H(x) - H(\pi_E(x)) \geq u} du~,
\end{eqnarray}
for some $u_0 \in \Real$ to be determined, and that:
 \begin{eqnarray*}
&   & \mu_n\set{x \in \Real^n \; ; \; \abs{\pi_D(x)} \leq w ~,~ H(x) - H(\pi_E(x)) \geq u} \\
& \leq & \mu_n \set{x \in \Real^n \, ; \; \exists t \in [-w,w] ,~ -\frac{t}{\sqrt{n}} \sum_{i=1}^n V'(x_i) - \frac{t^2}{2n} \sum_{i=1}^{n} \inf_{\xi_i \in [x_i,x_i + \frac{t}{\sqrt{n}}]} V''(\xi_i) \geq u } \\
& \leq & \mu_n \set{x \in \Real^n ; \frac{w}{\sqrt{n}} \abs{\sum_{i=1}^n V'(x_i)} + \frac{w^2}{2n} \sum_{i=1}^{n} \sup_{\xi_i \in [x_i-w/\sqrt{n},x_i + w/\sqrt{n}]} |V''(\xi_i)| \geq u } ~.
\end{eqnarray*}
When $n$ is greater than $(w/\delta)^2$, we evaluate this by:
\begin{eqnarray}
\nonumber & \leq &
 \mu_n \set{ x \in \Real^n ; \frac{w}{\sqrt{n}} \abs{\sum_{i=1}^n V'(x_i)} + \frac{w^2}{2n} \sum_{i=1}^{n} \brac{\sup_{\xi_i \in [x_i-\delta,x_i + \delta]} |V''(\xi_i)| - D_2} \geq u - \frac{D_2 w^2}{2}  } \\
\label{eq:two-terms} & \leq &
 \P \brac{\frac{w}{\sqrt{n}} \abs{\sum_{i=1}^n V'(X_i)} \geq \frac{u - u_0}{2} } + \P \brac{ \frac{w^2}{2n} \sum_{i=1}^{n} (Y_i - D_2) \geq \frac{u - u_0}{2}},
\end{eqnarray}
where $Y_1,\ldots,Y_n$ denote independent copies of $Y_0$ and $u_0 := D_2 w^2 / 2$. 

Observe that  by definition $\E (Y_0 - D_2) = 0$, and that $\E V'(X_0)  = 0$ as before. Consequently, whenever $u > u_0$ and $n$ tends to infinity, the first and second term in (\ref{eq:two-terms}) will be governed by the Central-Limit Theorem and Law of Large Numbers, respectively, in accordance with the heuristic argument from the Introduction. 
To obtain exponentially decaying estimates as $u \rightarrow \infty$ for each individual $n$, we apply Theorem \ref{thm:Bern} to each of the terms,  which together with (\ref{eq:integral-2}) yields:
\begin{eqnarray*}
& & \int_{\Real^n} \exp\brac{p (H(x) - H(\pi_E(x))} 1_{\abs{\pi_D(x)} \leq w}\, d\mu_n(x) \leq \exp(p u_0) \\
&+& 2 p \exp(p u_0) \int_{0}^\infty \exp\brac{pu - c' \min\brac{ \frac{\sqrt{n} \, u}{2w D_{1,\Psi_1}} , \frac{ u^2}{4w^2 D_{1,\Psi_1}^2} } } du \\
&+& p \exp(p u_0) \int_{0}^\infty \exp\brac{pu - c' n \min\brac{ \frac{ u}{w^2 (D_{2,\Psi_1}+D_2)} , \frac{ u^2}{w^4 (D_{2,\Psi_1}+D_2)^2}    }} du ~,
\end{eqnarray*}
where we have used that $\norm{Y_0 - D_2}_{L_{\Psi_1}} \leq \norm{Y_0}_{L_{\Psi_1}} + \norm{D_2}_{L_{\Psi_1}} = D_{2,\Psi_1}+D_2$.
We see that when $n$ is large enough, the above estimate is clearly bounded by:
\begin{eqnarray*}
&\leq & 2\exp(pu_0)\left(1+ p \exp(p u_0) \int_{0}^\infty \exp\brac{p u - c \frac{u^2}{w^2 D_{1,\Psi_1}^2}} du\right) \\
&\leq & 2\exp(p u_0) \brac{ 1 + p \exp(p^2 w^2 D_{1,\Psi_1}^2 / (4 c)) w D_{1,\Psi_1} \sqrt{\pi/c}  } ~.
\end{eqnarray*}
Recalling that $u_0 = D_2 w^2 / 2$ and using that $1 + \exp(x^2) x \leq C \exp(C x^2)$ for an appropriately chosen $C>1$, we obtain that the latter quantity is at most:
\[
 C' \exp(p D_2 w^2 / 2 + C' p^2 w^2 D_{1,\Psi_1}^2) ~,
\]
and the assertion follows.
\end{proof}

\begin{cor} \label{cor:conc-trans-model}
In addition to the assumptions of Proposition \ref{prop:conc-trans-model}, assume that the barycenter of $\mu$ is at the origin and that $M_3 = M_3(\mu) < \infty$. Then setting:
\begin{equation} \label{eq:w_0}
w_0  := \sqrt{\min\brac{M_2,\frac{1}{D_2 + D_{1,\Psi_1}^2}}} ~,
\end{equation}
the following estimate holds:
\[
\brac{\int_{\Real^n} \brac{\frac{d\mu_{E,w_0}}{d\mu_n}}^4 d\mu_n}^{1/4} \leq C \max(1,M_2 (D_2 + D_{1,\Psi_1}^2)) ~,
\]
for any integer $n$ greater than $\eta = \eta^{M_2,M_3 / M_2^{3/2},D_{1,\Psi_1},D_{2,\Psi_1}}_{\delta}$,
where $C>0$ is a universal constant.
\end{cor}
\begin{proof}
Recall that:
\[
d\mu_{E,w}(x) = \frac{1}{Z_{E,w}} \exp\brac{-H(\pi_E(x))} 1_{|\pi_D(x)| \leq w}\,  dx ~.
\]
Consequently, applying Proposition \ref{prop:conc-trans-model} with $p=4$, we know that for any integer $n$ greater than $\eta^{D_{1,\Psi_1},D_{2,\Psi_1},w,4}_{\delta}$:
\begin{equation} \label{eq:cor-calc}
\brac{\int_{\Real^n} \brac{\frac{d\mu_{E,w}}{d\mu_n}}^4 d\mu_n}^{1/4} \leq \frac{C'}{Z_{E,w}} \exp(C' w^2 (D_2 + D_{1,\Psi_1}^2)) ~.
\end{equation}
Since $Z_{E,w} = 2 w Z_E$ and $Z_E$ is the density at $0$ of $\frac{1}{\sqrt{n}} \sum_{i=1}^n X_i$, it is clear by the Central-Limit principle that when $n$ is large enough, $Z_E$ should approximate the density at zero of a Gaussian random variable with zero mean and variance $\sigma^2 = \frac{1}{n} \sum_{i=1}^n \E(X_i^2)$, i.e. $1/(\sqrt{2 \pi} \sigma)$. It is possible to make this rigorous by invoking a local Central-Limit Theorem, as in the previous subsection; however, anticipating future situations where the Central-Limit Theorem is not available due to dependencies, we proceed as follows, even though this incurs a quadratic penalty in our final estimate. Applying the Cauchy-Schwartz inequality and using Proposition \ref{prop:conc-trans-model} again with $p=-1$, we obtain:
\begin{eqnarray*}
Z_{E,w} & = & \int_{\Real^n} \exp\brac{H(x) - H(\pi_E(x)} 1_{\abs{\pi_D(x)} \leq w}\,  d\mu_n(x) \\
& \geq & \frac{ \brac{\int_{\Real^n} 1_{\abs{\pi_D(x)} \leq w} \, d\mu_n(x)}^2}{\int_{\Real^n} \exp\brac{-(H(x) - H(\pi_E(x))} 1_{\abs{\pi_D(x)} \leq w} \, d\mu_n(x)} \\
& \geq & \frac{ \brac{\mu^D\set{[-w,w]}}^2 }{ C'' \exp(C'' w^2 (D_2 + D_{1,\Psi_1}^2))} ~.
\end{eqnarray*}
We can now estimate $\mu^D\set{[-w,w]}$ by invoking the Berry--Esseen Theorem \ref{thm:BE}, which yields:
\[
\forall w > 0 \;\;\; \abs{\mu^D\set{[-w,w]} - \P(Z \in [-w/\sigma,w/\sigma])} \leq \frac{C M_3}{M_2^{3/2} \sqrt{n}} ~,
\]
for some universal constant $C>0$, where $Z$ is a standard Gaussian random-variable on $\Real$. Consequently, when $n$ is large enough, we obtain:
\[
Z_{E,w} \geq c \min(1,w/\sigma)^2 \exp(-C'' w^2 (D_2 + D_{1,\Psi_1}^2)) ~.
\]
Now plugging this back into (\ref{eq:cor-calc}) and using $w = w_0$, we obtain:
\[
\brac{\int_{\Real^n} \brac{\frac{d\mu_{E,w_0}}{d\mu_n}}^4 d\mu_n}^{1/4} \leq C (\sigma/w_0)^2 ~,
\]
as asserted.
\end{proof}

\subsection{Log-Sobolev Transference from $\mu$ to $\mu_E$}

We now translate the above transference results from the concentration to the log-Sobolev level. 

\begin{thm} \label{thm:model-zero-spin}
Let $\mu= \exp(-V(x)) dx$ denote a probability measure on $\Real$ with barycenter at $0$ so that $D_{1,\Psi_1} = D_{1,\Psi_1}(\mu) < \infty$ and $D_{2,\Psi_1} = D^\delta_{2,\Psi_1}(\mu) < \infty$ for some $\delta > 0$. Assume in addition that $(\Real,|\cdot|,\mu)$ satisfies $LSI(\rho)$, and that:
\begin{equation} \label{eq:model-kappa}
-\kappa := \inf_{x \in \Real} V''(x) \geq -\frac{\rho}{8} ~.
\end{equation}
Then for any integer $n$ greater than $\eta = \eta^{M_3/M_2^{3/2},D_{1,\Psi_1},D_{2,\Psi_1}}_{\rho,\delta}$, the conservative zero-mean spin system $(E,|\cdot|,\mu_E)$ satisfies LSI with constant:
\[
\rho_{LS}(E,|\cdot|,\mu_E) \geq c \frac{\rho}{Q^C} ~,
\]
where $c,C>0$ are universal constants and $Q$ is the following scale-invariant quantity:
\[
Q:= \max(1,M_2 (D_2 + D_{1,\Psi_1}^2)) ~.
\]
\end{thm}
\begin{proof}
First, note that the Herbst argument (see Appendix) implies the sub-Gaussian decay of Lipschitz functions on a space satisfying LSI, which in particular ensures the existence of all finite moments of $X_0$, and implies $M_2 \leq C / \rho$, for some universal constant $C>0$.

It is well-known (see Appendix) that the LS inequality tensorizes with respect to the Euclidean  norm, and so $(\Real^n,|\cdot|,\mu_n)$ satisfies a LS inequality with the same constant $\rho$. Since $d\mu_n(x) = \exp(-H(x)) dx$ with $H(x) = \sum_{i=1}^n V(x_i)$, it follows that $ \mathrm{Hess} H \geq -\kappa Id$ as tensors in $\Real^n$. The same bound holds for the restriction of these tensors onto any linear subspace, and so it follows that $(E,|\cdot|,\mu_E)$ satisfies our $\kappa$-semi-convexity assumptions. Moreover, the uniform thickening of $\mu_E$ in the direction $D$ orthogonal to $E$ only adds $0$ as an eigenvalue to the Hessian matrix in that direction, and hence $(\Real^n,|\cdot|,\mu_{E,w})$ also satisfies our $\kappa$-semi-convexity assumptions for any $w > 0$.

We now transfer the log-Sobolev inequality on $(\Real^n,|\cdot|,\mu_n)$ onto $(\Real^n,|\cdot|,\mu_{E,w_0})$ by applying Theorem \ref{thm:log-Sob-trans} with $p=4$. Note that (\ref{eq:model-kappa}) implies that $\rho > (16/3) \kappa$, which is required for applying Theorem \ref{thm:log-Sob-trans}, and consequently the parameter $\theta$ in that theorem satisfies $\theta \geq 1/3$. Estimating $\int_{\Real^n} \brac{\frac{d\mu_{E,w_0}}{d\mu_n}}^4 d\mu_n$ using Corollary \ref{cor:conc-trans-model}, it follows that for any integer $n$ larger than $\eta^{M_2,M_3 / M_2^{3/2},D_{1,\Psi_1},D_{2,\Psi_1}}_{\delta}$:
\[
\rho_{LS}(\Real^n,|\cdot|,\mu_{E,w_0}) \geq c \frac{\rho}{Q^C} ~,
\]
for some universal constants $c,C>0$. By the tensorization property of the LS inequality:
\[
\rho_{LS}(\Real^n,|\cdot|,\mu_{E,w_0}) = \min(\rho_{LS}(E,|\cdot|,\mu_E),\rho_{LS}(\Real,|\cdot|,\nu_{[-w_0,w_0]})) ~,
\]
where $\nu_{[-w_0,w_0]}$ denotes the uniform measure on $[-w_0,w_0]$. It follows that:
\[
\rho_{LS}(E,|\cdot|,\mu_E) \geq \rho_{LS}(\Real^n,|\cdot|,\mu_{E,w_0}) \geq c \frac{\rho}{Q^C} ~,
\]
as asserted.
\end{proof}

Repeating the above argument and replacing Corollary \ref{cor:conc-trans-model} by Corollary \ref{cor:conc-trans-model-2}, we obtain:

\begin{thm} \label{thm:model-zero-spin-2}
Let $\mu= \exp(-V(x)) dx$ denote a probability measure on $\Real$ with barycenter at $0$ so that $D_{1,\Psi_1} = D_{1,\Psi_1}(\mu) < \infty$ and
$\lambda := \norm{d\mu/dx}_{L^\infty} < \infty$. Assume in addition that $(\Real,|\cdot|,\mu)$ satisfies $LSI(\rho)$, and that:
\begin{equation} \label{eq:model-kappa-2}
-\kappa := \inf_{x \in \Real} V''(x) \geq -\frac{\rho}{8} ~.
\end{equation}
Then for any integer $n$ greater than $\eta = \eta^{M_3/M_2^{3/2},M_3 \lambda^3}$, the conservative zero-mean spin system $(E,|\cdot|,\mu_E)$ satisfies LSI with constant:
\[
\rho_{LS}(E,|\cdot|,\mu_E) \geq c \frac{\rho}{Q^C} ~,
\]
where $c,C>0$ are universal constants and $Q$ is the following scale-invariant quantity:
\[
Q:= \max(1,M_2 (\kappa + D_{1,\Psi_1}^2)) ~.
\]
\end{thm}

\begin{rem} \label{rem:constants}
We did not attempt to optimize over numeric constants above. In particular, the constant $8$ in the conditions (\ref{eq:model-kappa}) and (\ref{eq:model-kappa-2}) may easily be improved down to $4+\eps$ by using a $p$ greater than $4$ in Theorem \ref{thm:log-Sob-trans} and Corollaries \ref{cor:conc-trans-model} and \ref{cor:conc-trans-model-2}. Moreover, it should be possible to improve it all the way down to $1+\eps$ (and similarly, the constant $4$ in (\ref{eq:ls-assump}) should be pushed down to $1$) by carefully revisiting Lemma \ref{lem:conc-trans}, and replacing $\K_1(r/2)$ in (\ref{eq:conc-trans-2}) by $\K_1((1-\xi) r)$ for some arbitrarily small $\xi > 0$. We refrain here from pushing these numeric constants to their limit, since this seems irrelevant for applications.
\end{rem}

\section{Weakly Interacting Conservative Model} \label{sec:inter}

In this section, we modify our non-interacting Hamiltonian $H(x) = \sum_{i=1}^n V(x_i)$ by adding some weak interaction term $-I_{A}(x)$ corresponding to a weighted $n$ by $n$ symmetric  matrix $A = \set{a_{i,j}}$ with zero diagonal ($a_{i,j} = a_{j,i}$ and $a_{i,i} = 0$):
\[
H_{A}(x) = H(x) - I_{A}(x) ~,~ I_{A}(x) := \sum_{i,j =1}^n a_{i,j} x_i x_j ~.
\]
Define the corresponding Gibbs probability measure:
\[
\mu_{A} :=  \frac{1}{Z_A} \exp(-H_A(x)) dx ~,
\]
where $Z_{A} > 0$ is a normalization term. To make sure that the Gibbs measure is indeed well defined (at least for weak enough interactions, see below), we will require throughout this section the following:

\medskip
\noindent \textbf{Assumptions for Interaction}:
\begin{itemize}
\item $\mu$ has barycenter at the origin: $\int x d\mu(x) = 0$.
\item $\mu$ has sub-Gaussian tail decay:
\begin{equation} \label{eq:sub-Gaussian}
\exists \rho > 0 \;\;\; \forall \lambda \in \Real \;\;\; \int \exp(\lambda x) d\mu(x) \leq \exp(\frac{\lambda^2}{2 \rho})  ~.
\end{equation}
\end{itemize}

Let $X = (X_1,\ldots,X_n)$ denote a random-vector in $\Real^n$ distributed according to $\mu_{A}$, set $S = \sum_{i=1}^n X_i$, and let $\mu_{A,E}$ denote the law of $X$ conditioned on $S = 0$, i.e.:
\[
\mu_{A,E} = \frac{1}{Z_{A,E}} \exp\brac{-H_{A}(x)} d\vol_E(x) ~.
\]
Again, $\mu_{A,E}$ is not absolutely continuous with respect to $\mu_A$, so we define $\mu_{A,E,w}$ by ``thickening" it uniformly in the diagonal direction by a width of $w > 0$ from each side:
\[
d\mu_{A,E,w}(x) = \frac{1}{2w Z_{A,E}} \exp\brac{-H_A(\pi_E(x))} 1_{|\pi_D(x)| \leq w} dx ~.
\]

To see that the Gibbs measure is indeed well defined for weak interactions, and to establish all of our estimates in this section, we require the following theorem, communicated to us by Rafal Lata{\l}a, to whom we are indebted; its proof is deferred to the Appendix. Recall that we denote the operator and Hilbert-Schmidt norms of $A$ by $\norm{A}_{op}$ and $\norm{A}_{HS}$, respectively.

\begin{thm}[Lata{\l}a] \label{thm:Latala}
Let $X_1,\ldots,X_n$ denote a sequence of independent random-variables on $\Real$, so that for each $i=1,\ldots,n$, the law $\mu_i$ of $X_i$ satisfies the Assumptions for Interaction. Then there exist universal constants $C_2,c_2 > 0$ so that for any integer $n \geq 1$ and $n$ by $n$ symmetric matrix $A = \set{a_{i,j}}$ with zero diagonal:
\[
\P\brac{ \abs{\sum_{i,j=1}^n a_{i,j} X_i X_j} \geq t } \leq C_2 \exp\brac{-c_2 \min\brac{\frac{\rho^2 t^2}{\norm{A}_{HS}^2} , \frac{\rho t}{\norm{A}_{op}}}} \;\;\; \forall t > 0 ~.
\]
\end{thm}

\subsection{Concentration Transference from $\mu_n$ to $\mu_{E,w}$}

The main new calculation in this section is given in the following:

\begin{prop} \label{prop:conc-trans-inter}
Let $p \in \Real$, and assume that:
\begin{itemize}
\item The Assumptions for Interaction are satisfied.
\item For some appropriately chosen universal constants $C_3,C_4 > 0$:
\begin{equation} \label{eq:Aop-small}
\norm{A}_{op} \leq \frac{C_3}{1 + C_4 \sqrt{\rho} w} \frac{\rho}{|p|} ~.
\end{equation}
\end{itemize}
Then:
\[
\int_{\Real^n} \exp(p I_A(\pi_E(x))) 1_{\abs{\pi_D(x)} \leq w} d\mu_n(x) \leq C_5 \exp(|p| \norm{A}_{op} w^2 + C_6 p^2 (1 + C_4 \sqrt{\rho} w)^2 \norm{A}_{HS}^2 / \rho^2 )  ~.
\]
\end{prop}
\begin{proof}
We may assume that $p > 0$, since otherwise we may replace $A$ by $-A$. Furthermore, all assumptions and both sides of the desired inequality are invariant under the transformations $X_0 \mapsto X_0 / \sqrt{\rho}$, $A / \rho \mapsto A$, $w \sqrt{\rho} \mapsto w$ and $\rho \mapsto 1$, and so we may assume that $\rho = 1$; nevertheless, we proceed in full generality. We evaluate:
\begin{eqnarray}
\nonumber & & \int_{\Real^n} \exp(p I_A(\pi_E(x))) 1_{\abs{\pi_D(x)} \leq w}\,  d\mu_n(x) \\
\nonumber & = &  \int_{-\infty}^\infty p \exp(pu) \mu_n\set{x \in \Real^n \; ; \; \abs{\pi_D(x)} \leq w ~,~ I_A(\pi_E(x)) \geq u} du \\
\label{eq:inter-integral} & \leq & e^{p u_0} + \int_{u_0}^{\infty} p e^{p u} \mu_n\set{x \in \Real^n \; ; \; \abs{\pi_D(x)} \leq w ~,~ I_A(\pi_E(x)) \geq u} du  ~,
\end{eqnarray}
for some $u_0$ to be determined. We proceed by roughly evaluating the integrand as follows:
\begin{eqnarray*}
&   & \mu_n\set{x \in \Real^n \; ; \; \abs{\pi_D(x)} \leq w ~,~ I_A(\pi_E(x)) \geq u}\\
& = & \mu_n \set{x \in \Real^n \; ; \; \abs{\pi_D(x)} \leq w ~,~ \sum_{i,j=1}^n a_{i,j} (\pi_E(x))_i (\pi_E(x))_j \geq u } \\
& \leq & \mu_n \set{x \in \Real^n \; ; \; \exists t \in [-w,w] ~  \sum_{i,j=1}^n a_{i,j} (x_i + t/\sqrt{n}) (x_j + t/\sqrt{n}) \geq u } \\
& \leq & \mu_n \set{x \in \Real^n \; ; \; \sum_{i,j=1}^n a_{i,j} x_i x_j + \frac{2 w}{\sqrt{n}} \abs{\sum_{i,j=1}^n a_{i,j} x_i} + \frac{w^2}{n} \abs{\sum_{i,j=1}^n a_{i,j}} \geq u } ~.
\end{eqnarray*}
Since:
\[
\frac{1}{n} \abs{\sum_{i,j=1}^n a_{i,j}} \leq \norm{A}_{op} ~,
\]
setting $u_0 = \norm{A}_{op} w^2$, we obtain:
\begin{eqnarray}
\nonumber & \leq & \P \brac{ \sum_{i,j=1}^n a_{i,j} X_i X_j + \frac{2 w}{\sqrt{n}} \abs{\sum_{i,j=1}^n a_{i,j} X_i} \geq u - u_0 } \\
\label{eq:inter-two-terms} & \leq & \P \brac{\sum_{i,j=1}^n a_{i,j} X_i X_j \geq \alpha (u - u_0)} +  \P \brac{ \frac{2 w}{\sqrt{n}} \abs{\sum_{i,j=1}^n a_{i,j} X_i} \geq (1-\alpha)(u - u_0)} ~,
\end{eqnarray}
for some $\alpha \in (0,1)$ to be determined. The first term above is immediately estimated using Theorem \ref{thm:Latala}. Bounding the second term is elementary, since using (\ref{eq:sub-Gaussian}) and independence:
\[
\forall t > 0 \;\;\; \P \brac{\sum_{i=1}^n \alpha_i X_i \geq t} \leq \inf_{\lambda > 0} \exp\brac{\frac{\lambda^2}{2 \rho} \sum_{i=1}^n \alpha_i^2 - \lambda t} = \exp\brac{-\frac{\rho t^2}{2 \sum_{i=1}^n \alpha_i^2}} ~.
\]
Since:
\[
\sum_{i=1}^n \brac{\frac{1}{\sqrt{n}} \sum_{j=1}^n a_{i,j}}^2 = | A (1/\sqrt{n},\ldots,1/\sqrt{n})|^2 \leq \norm{A}_{op}^2 ~,
\]
we obtain:
\[
\forall t > 0 \;\;\; \P \brac{ \frac{2 w}{\sqrt{n}} \abs{\sum_{i=1}^n (\sum_{j=1}^n a_{i,j}) X_i} \geq t} \leq 2 \exp\brac{-\frac{\rho t^2}{8 w^2 \norm{A}_{op}^2}} ~.
\]
Combining the estimates on both terms in (\ref{eq:inter-two-terms}), we conclude that:
\begin{eqnarray*}
& & \forall u \geq u_0, \;  \forall \alpha \in (0,1), \;\mu_n\set{x \in \Real^n \; ; \; \abs{\pi_D(x)} \leq w ~,~ I_A(\pi_E(x)) \geq u} \\
& \leq &
C_2 \exp\brac{-c_2 \min\brac{\frac{\alpha^2 \rho^2 (u-u_0)^2}{\norm{A}_{HS}^2} , \frac{\alpha \rho (s-s_0)}{\norm{A}_{op}}}} + 2 \exp\brac{-\frac{\rho (1-\alpha)^2 (u-u_0)^2}{8 w^2 \norm{A}_{op}^2}} ~.
\end{eqnarray*}
Plugging this into (\ref{eq:inter-integral}), we obtain:
\begin{eqnarray*}
& & \int_{\Real^n} \exp(p I_A(\pi_E(x))) 1_{\abs{\pi_D(x)} \leq w} d\mu_n(x) \leq \exp(p u_0) \\
& + & C_2 p \exp(p u_0) \int_{0}^\infty \exp\brac{p u - c_2 \min \brac{\frac{\alpha^2 \rho^2 u^2}{\norm{A}_{HS}^2} , \frac{\alpha \rho u}{\norm{A}_{op}}} } du \\
& + & 2 p \exp(p u_0) \int_{0}^\infty \exp\brac{p u - \frac{(1-\alpha)^2 \rho u^2}{8 w^2 \norm{A}_{op}^2}} du ~.
\end{eqnarray*}
Setting:
\[
\alpha = \frac{1}{1 + \sqrt{8 c_2 \rho} w \frac{\norm{A}_{op}}{\norm{A}_{HS}}} ~,
\]
we obtain:
\begin{eqnarray*}
&\leq& \exp(p u_0) + (C_2 + 2) p \exp(p u_0) \int_0^\infty \exp\brac{p u - \frac{c_2 \rho^2 u^2}{\brac{1 +\sqrt{8 c_2\rho} w \frac{\norm{A}_{op}}{\norm{A}_{HS}}}^2 \norm{A}_{HS}^2} } du \\
&+& C_2 p \exp(p u_0) \int_{0}^\infty  \exp\brac{p u - \frac{c_2 \rho u}{\brac{1 + \sqrt{8 c_2 \rho} w\frac{\norm{A}_{op}}{\norm{A}_{HS}}} \norm{A}_{op}}} du ~.
\end{eqnarray*}
Since $\norm{A}_{op} \leq \norm{A}_{HS}$, we see that as soon as (\ref{eq:Aop-small}) is satisfied with say $C_3 = c_2 / 2$ and $C_4 =\sqrt{8 c_2}$, we obtain:
\begin{eqnarray*}
&\leq& e^{p u_0}\brac{1 + (C_2 + 2) p \int_0^\infty \exp\brac{p u - \frac{c_2 \rho^2 u^2}{\brac{1 + C_4 \sqrt{\rho} w}^2 \norm{A}_{HS}^2}} du +
C_2 p \int_{0}^\infty e^{- pu} du }\\
&\leq& e^{p u_0} \brac{1 + C_2 + (C_2 + 2) p \exp\brac{\frac{p^2 (1+C_4 \sqrt{\rho} w)^2 \norm{A}_{HS}^2}{4 c_2 \rho^2}} \frac{ \sqrt{\pi} \norm{A}_{HS} (1 + C_4 \sqrt{\rho} w)}{\sqrt{c_2} \rho} } ~.
\end{eqnarray*}
Noting that $1 + x \exp(x^2) \leq C \exp(C x^2)$ for some appropriately chosen constant $C>1$, and recalling that $u_0 = \norm{A}_{op} w^2$, the assertion follows.
\end{proof}

\begin{prop} \label{prop:conc-trans-inter-full}
Let the Assumptions for Interaction and the assumptions specified in Proposition \ref{prop:conc-trans-model} be satisfied. Assume further that for an appropriate universal constant $c>0$:
\begin{equation} \label{eq:Aop-small-simple}
\norm{A}_{op} \leq c \rho ~.
\end{equation}
Then setting:
\[
w_0 := \sqrt{\min\brac{M_2 , \frac{1}{D_2 + D_{1,\Psi_1}^2}}} ~,
\]
The following estimate holds:
\[
\brac{\int_{\Real^n} \brac{\frac{d\mu_{A,E,w_0}}{d\mu_n}}^4 d\mu_n}^{1/4} \leq C \max\brac{1 , M_2 (D_2 + D_{1,\Psi_1}^2)} \exp(C \norm{A}_{HS}^2 / \rho^2) ~,
\]
for all integers $n$ greater than $\eta = \eta^{M_2, M_3 / M_2^{3/2},D_{1,\Psi_1},D_{2,\Psi_1}}_{\delta}$.
\end{prop}
\begin{proof}
Given $w > 0$, we write:
\begin{eqnarray*}
\int \brac{\frac{d\mu_{A,E,w}}{d\mu_n}}^4 d\mu_n = \frac{\int \brac{\frac{d\mu_{A,E,w}}{d\mu_n}}^4 d\mu_n}{\brac{\int \frac{d\mu_{A,E,w}}{d\mu_n} d\mu_n}^4}  = \frac{\int (f g h)^4 d\mu_n}{\brac{\int f g h \; d\mu_n}^4} ~,
\end{eqnarray*}
where:
\[
f(x) = \exp(H(x) - H(\pi_E(x))) ~,~ g(x) = \exp(I_A(\pi_E(x))) ~,~ h(x) = 1_{\abs{\pi_D(x)} \leq w} ~.
\]

Since we do not care about numerical constants (even inside exponents), we proceed by applying the Cauchy-Schwartz inequality several times, and obtain:
\begin{eqnarray}
\nonumber & & \int \brac{\frac{d\mu_{A,E,w}}{d\mu_n}}^4 d\mu_n  \leq \brac{\int f^8 h \; d\mu_n}^{1/2} \brac{\int g^8 h \; d\mu_n}^{1/2} \brac{\int \frac{1}{fg} h \; d\mu_n}^4 \brac{\int h \; d\mu_n}^{-8} \\
\label{eq:fgh} & \leq & \brac{\int f^8 h \; d\mu_n}^{1/2} \brac{\int g^8 h \; d\mu_n}^{1/2} \brac{\int f^{-2} h \; d\mu_n}^2 \brac{\int g^{-2} h \; d\mu_n}^2 \brac{\int h \; d\mu_n}^{-8} ~.
\end{eqnarray}

The integrals above involving $f$ and $g$ were estimated in Propositions \ref{prop:conc-trans-model} and \ref{prop:conc-trans-inter}, respectively. And since $\int h \; d\mu_n = \mu^D\set{[-w,w]}$, the Berry--Esseen Theorem \ref{thm:BE} implies as in the proof of Corollary \ref{cor:conc-trans-model} that:
\[
\forall w > 0 \;\;\; \abs{\mu^D\set{[-w,w]} - \P(Z \in [-w/\sigma,w/\sigma])} \leq \frac{C M_3}{M_2^{3/2} \sqrt{n}} ~,
\]
for some universal constant $C>0$, where $Z$ is a standard Gaussian random-variable on $\Real$ and $\sigma^2 = \frac{1}{n} \sum_{i=1}^n \E(X_i^2)$. Consequently, when $n$ is large enough, we obtain:
\[
\int h \; d\mu_n \geq c \min(1,w/\sigma) ~.
\]
Plugging all of these estimates into (\ref{eq:fgh}), we obtain for all integers $n$ also larger than $\eta^{D_{1,\Psi_1},D_{2,\Psi_1},w}_{\delta}$ from Proposition \ref{prop:conc-trans-model}, and all interaction matrices $A$ satisfying (\ref{eq:Aop-small}), that:
\begin{eqnarray*}
\brac{\int \brac{\frac{d\mu_{A,E,w}}{d\mu_n}}^4 d\mu_n}^{1/4} & \leq & C \max(1,\sigma/w)^2 \exp\brac{C \brac{w^2 (D_2 + D_{1,\Psi_1}^2) + w^2 \norm{A}_{op}}} \\
 & & \exp\brac{C (1 + C_4 \sqrt{\rho} w)^2 \frac{\norm{A}_{HS}^2}{\rho^2} }  ~.
\end{eqnarray*}
Setting $w = w_0$ and noting that $\sigma^2 = M_2 \leq C / \rho$ thanks to the sub-Gaussian decay assumption (\ref{eq:sub-Gaussian}), it follows that $\sqrt{\rho} w_0 \leq \sqrt{\rho} \sigma \leq \sqrt{C}$. The smallness assumption (\ref{eq:Aop-small}) then translates into the assumption (\ref{eq:Aop-small-simple}), and the assertion consequently readily follows.
\end{proof}

\subsection{log-Sobolev Transference from $\mu$ to $\mu_{E}$}

We can now obtain:

\begin{thm} \label{thm:inter-zero-spin}
Let $\mu = \exp(-V(x)) dx$ denote a probability measure on $\Real$ with barycenter at the origin, so that $D_{1,\Psi_1} = D_{1,\Psi_1}(\mu) < \infty$ and $D_{2,\Psi_1} = D^\delta_{2,\Psi_1}(\mu) < \infty$ for some $\delta > 0$. Let $A$ denote an $n$ by $n$ symmetric matrix with zero diagonal.
Assume that $(\Real,|\cdot|,\mu)$ satisfies $LSI(\rho)$, and that:
\begin{itemize}
\item \begin{equation} \label{eq:inter-kappa}
-\kappa := \inf_{x \in \Real} V''(x) \geq -\frac{\rho}{8} ~.
\end{equation}
\item \begin{equation} \label{eq:Aop-small-simple2}
\norm{A}_{op} \leq c \rho ~,
\end{equation}
where $c>0$ is an appropriate universal constant.
\end{itemize}
Then for any integer $n$ greater than $\eta = \eta^{M_3 / M_2^{3/2}, D_{1,\Psi_1},D_{2,\Psi_1}}_{\rho,\delta}$, the weakly-interacting conservative zero-mean spin system
$(E,|\cdot|,\mu_{A,E})$ satisfies LSI with constant:
\[
\rho_{LS}(E,|\cdot|,\mu_{A,E}) \geq c \frac{\rho}{Q^C} ~,
\]
where $c,C>0$ are universal constants and $Q$ is the following scale-invariant quantity:
\[
Q := \max\brac{1 , M_2 (D_2 + D_{1,\Psi_1}^2)} \exp(\norm{A}_{HS}^2 / \rho^2) ~.
\]
\end{thm}
\begin{proof}
The proof is almost identical to that of Theorem \ref{thm:model-zero-spin}.

Since $\mu$ satisfies $LSI(\rho)$, it follows by the Herbst argument (see Appendix) that the sub-Gaussian tail decay assumption (\ref{eq:sub-Gaussian}) is satisfied (with the same constant $\rho$). In particular, it follows that $M_2 \leq C / \rho$, for some universal constant $C>0$.
It is well-known (see Appendix) that LSI tensorizes with respect to the Euclidean ($\ell_2$) norm, and so $(\Real^n,|\cdot|,\mu_n)$ also satisfies $LSI(\rho)$. Since $ \mathrm{Hess} H \geq -\kappa Id$ and $ \mathrm{Hess} I_A \equiv A \leq \norm{A}_{op} Id$ as tensors in $\Real^n$, it follows that $ \mathrm{Hess} H_A =  \mathrm{Hess} H -  \mathrm{Hess} I_A \geq -\kappa_A Id$, where $\kappa_A = \kappa + \norm{A}_{op}$. The same bound holds for the restriction of these tensors onto any linear subspace, and so it follows that $(E,|\cdot|,\mu_{A,E})$ satisfies our $\kappa_A$-semi-convexity assumptions. Moreover, the uniform thickening of $\mu_{A,E}$ in the direction $D$ orthogonal to $E$ only adds $0$ as an eigenvalue to the Hessian matrix in that direction, and hence $(\Real^n,|\cdot|,\mu_{A,E,w})$ also satisfies our $\kappa_A$-semi-convexity assumptions for any $w > 0$. Note that by ensuring that the constant $c$ in (\ref{eq:Aop-small-simple2}) is not greater than $1/24$, it follows from (\ref{eq:Aop-small-simple2}) and (\ref{eq:inter-kappa}) that $\kappa_A \leq \frac{\rho}{6}$.

We now transfer the log-Sobolev inequality on $(\Real^n,|\cdot|,\mu_n)$ onto $(\Real^n,|\cdot|,\mu_{A,E,w_0})$ by applying Theorem \ref{thm:log-Sob-trans} with $p=4$.
To this end, we invoke Proposition \ref{prop:conc-trans-inter-full} for estimating $\int_{\Real^n} \brac{\frac{d\mu_{A,E,w_0}}{d\mu_n}}^4 d\mu_n$.
Note that $\rho \geq 6 \kappa_A > (16/3) \kappa_A$, which is required for applying Theorem \ref{thm:log-Sob-trans}, and consequently the parameter $\theta$ in that Theorem satisfies $\theta \geq 1/9$. Also note that all of the assumptions for applying Proposition \ref{prop:conc-trans-inter-full} are indeed satisfied, including the sub-Gaussian tail decay assumption (\ref{eq:sub-Gaussian}) and the smallness condition (\ref{eq:Aop-small-simple}) (by appropriately choosing $c$ in (\ref{eq:Aop-small-simple2})). Estimating $\int_{\Real^n} \brac{\frac{d\mu_{E,w_0}}{d\mu_n}}^4 d\mu_n$ using Proposition \ref{prop:conc-trans-inter-full}, it follows that for any integer $n$ larger than $\eta^{M_2,M_3 / M_2^{3/2},D_{1,\Psi_1},D_{2,\Psi_1}}_{\delta}$:
\[
\rho_{LS}(\Real^n,|\cdot|,\mu_{A,E,w_0}) \geq c \frac{\rho}{Q^C} ~,
\]
where $c,C>0$ are universal constants. By the tensorization property of the log-Sobolev inequality:
\[
\rho_{LS}(\Real^n,|\cdot|,\mu_{A,E,w_0}) = \min(\rho_{LS}(E,|\cdot|,\mu_{A,E}),\rho_{LS}(\Real,|\cdot|,\nu_{[-w_0,w_0]})) ~,
\]
where $\nu_{[-w_0,w_0]}$ denotes the uniform measure on $[-w_0,w_0]$. It follows that:
\[
\rho_{LS}(E,|\cdot|,\mu_{A,E}) \geq \rho_{LS}(\Real^n,|\cdot|,\mu_{A,E,w_0}) \geq c \frac{\rho}{Q^C} ~,
\]
as asserted.
\end{proof}

\newcommand{\leftexp}[2]{{#2}^{\wedge #1}}

\section{Uniform Bounds for Arbitrary Mean-Spins} \label{sec:mean-spin}

\subsection{Conservative Spin Model}

We have thus far only treated canonical ensembles obtained from conditioning the grand canonical ensemble $\mu_n$ on having the mean-spin $S = \frac{1}{n} \sum X_i$ fixed at $0$, the assumed barycenter of $\mu$. Consider now the canonical ensemble $\mu_{E_s}$ obtained from conditioning the grand canonical ensemble $\mu_n$ on having a mean-spin $S = s$, namely:
\[
d\mu_{E_s} = \frac{1}{Z_{E_s}} \exp(-H(x))\,  d\vol_{E_s}(x) ~,
\]
where $E_s$ denotes the affine hyperplane $\frac{1}{n} \sum_{i=1}^n x_i = s$ and $Z_{E_s} > 0$ is a normalization term. Naturally, we will only consider values of $s$ which lie in $\mathrm{isupp}(\mu)$, the interior of the support of $\mu$.

To handle arbitrary mean-spin values, we use the well-known Cram\'er trick, which is the key ingredient in the Cram\'er Theorem on Large Deviations, a central tool in the approaches of \cite{GOVWTwoScaleApproachForLSI,MenzOttoLSIForPerturbationsOfSuperQuadraticPotential,MenzLSIwithWeakInteraction}. Denote by $\leftexp{a}{\mu}$ the probability measure on $\Real$ obtained from $\mu$ by multiplying its density by $\exp(a x)$ and renormalizing:
\[
\leftexp{a}{\mu} := \frac{1}{Z^a} \frac{d\mu}{dx} \exp(a x)\, dx ~.
\]
Note that in our setting $\leftexp{a}{\mu}$ is indeed well defined, since $\mu$ is always assumed to satisfy LSI, and hence (by the Herbst argument) it must have sub-Gaussian tail-decay, and so $Z^a = \int \exp(a x) d\mu(x) < \infty$ for all $a \in \Real$. The key observation is that the densities of $\mu_n$ and $(\leftexp{a}{\mu})_n$ coincide up to a constant multiple on each hyperplane $E_s$, and consequently $\mu_{E_s} = (\leftexp{a}{\mu})_{E_s}$ for all $a,s \in \Real$.

Let $\leftexp{a}{X}$ be a random variable distributed according to $\leftexp{a}{\mu}$. Given $a \in \Real$, we denote $s(a) := \E(\leftexp{a}{X})$ (note that this expectation always exists in our setting). The function $\Real \ni a \mapsto s(a) \in \mathrm{isupp}(\mu)$ is well known to be increasing and onto, and we denote its inverse by $a(s)$. Denoting by $T^s$ the function $T^s(x) := x-s$ translating by $s$ to the left, we denote $\mu^a := (T^{s(a)})_*(\leftexp{a}{\mu})$, the translation of $\leftexp{a}{\mu}$ having barycenter at $0$. We denote by $X^a$ the random-variable with law $\mu^a$. Finally, observe that the measure-metric space $(E_s,\abs{\cdot},(\leftexp{a(s)}{\mu})_{E_s})$ is isometrically isomorphic to the measure-metric space $(E_0,\abs{\cdot},(\mu^{a(s)})_{E_0})$, since the Euclidean structure is compatible with translations. We conclude that the best constants in a LSI on these spaces coincide, and so to obtain uniform estimates on $\rho_{LS}(E_s,\abs{\cdot},(\leftexp{a(s)}{\mu})_{E_s})$ in $s \in \mathrm{isupp}(\mu)$, we must obtain uniform estimates on $\rho_{LS}(E_0,\abs{\cdot},(\mu^{a})_{E_0})$ in $a \in \Real$. Using Theorem \ref{thm:model-zero-spin}, we obtain:

\begin{thm} \label{thm:uniform-model}
Let $\mu = \exp(-V(x)) \, dx$ denote a probability measure on $\Real$ with $V \in C^2(\Real)$ and $\lim_{x \rightarrow \pm \infty} V(x) = \infty$. Fixing $\delta > 0$, denote $\rho^a := \rho_{LS}(\Real,\abs{\cdot},\mu^a)$, $M_p^a := M_p(\mu^a)$, $D_{1,\Psi_1}^a := D_{1,\Psi_1}(\mu^a)$ and $D_{2,\Psi_1}^a := D_{2,\Psi_1}^\delta(\mu^a)$. Assume that:
\[
\rho^a \geq \bar{\rho} > 0 ~,~ \frac{M_3^a}{(M_2^a)^{3/2}} \leq \bar{M} < \infty ~,~ D_{1,\Psi_1}^a \leq \bar{D_1} < \infty ~,~ D_{2,\Psi_1}^a \leq \bar{D_2} < \infty ~,
\]
uniformly in $a \in \Real$, and that:
\[
-\kappa := \inf_{x \in \Real} V''(x) \geq -\frac{\bar{\rho}}{8} ~.
\]
Then for $n \geq \eta^{\bar{M},\bar{D_1},\bar{D_2}}_{\bar{\rho},\delta}$, the canonical ensemble $(E_s,\abs{\cdot},\mu_{E_s})$ satisfies a LSI, uniformly in the system size $n$ and mean-spin value $s \in \Real$, depending solely on $\bar{\rho}$, $\bar{M}$, $\bar{D_1}$ and $\bar{D_2}$.
\end{thm}
\begin{proof}
By the preceding discussion, it is enough to verify a uniform lower bound on $\rho_{LS}(E_0,\abs{\cdot},\mu^{a}_{E_0})$ in $a \in \Real$. Write $\mu^a = \exp(-V_a(x))\, dx$, and apply Theorem \ref{thm:model-zero-spin} to $\mu^a$, which is possible thanks to the requirement that $-\rho_a / 8 \leq -\kappa \leq V_a''$ for all $a \in \Real$.
Note that always $M_2^a \leq C / \rho^a \leq C / \bar{\rho}$ by the sub-Gaussian tail decay (\ref{eq:sub-Gaussian}) guaranteed by the Hebst argument (since $\E(X^a) = 0$).
The proof is complete.
\end{proof}

To demonstrate the desired uniformity for a concrete class of measures $\mu$, we recall the following definition given in the Introduction:
\begin{dfn*}
A probability measure $\mu = \exp(-V(x))\,  dx$ is called $(\alpha,\beta,\omega)$ \emph{weakly Gaussian} if we may decompose $V = V_{\text{conv}} + V_{\text{pert}}$ so that:
\begin{itemize}
\item $V_{\text{conv}}, V_{\text{pert}} \in C^2(\Real)$.
\item $V_{\text{conv}}'' \geq \alpha > 0$.
\item $\sup V_{\text{pert}} - \inf V_{\text{pert}} \leq \omega < \infty$.
\item $- \kappa := - \frac{1}{8} \alpha \exp(-\omega) \leq V'' \leq \beta < \infty$. \end{itemize}
\end{dfn*}

\begin{lem} \label{lem:weakly-Gaussian}
Let $\mu$ be $(\alpha,\beta,\omega)$ weakly Gaussian. Then for any $a \in \Real$:
\begin{enumerate}
\item $\rho^a := \rho_{LS}(\Real,\abs{\cdot},\mu^a) \geq \alpha \exp(-\omega)$.
\item $\mu^a = \exp(-V_a(x)) dx$ with $-\rho^a / 8 \leq V_a'' \leq \beta$.
\item For any $\delta > 0$ and $a \in \Real$, the parameters $M_3^a/ (M_2^a)^{3/2}$, $D^a_{1,\Psi_1}$ and $D^a_{2,\Psi_1}$ associated to the measure $\mu^a$ are uniformly bounded above by functions of $\alpha$, $\beta$ and $\omega$.
\end{enumerate}
\end{lem}
\begin{proof}
The first assertion is an immediate consequence of the Bakry--\'Emery condition for LSI (see Appendix) and the Holley--Stroock perturbation argument (\ref{eq:HS}). The second assertion is immediate from the definition of a weakly Gaussian measure and the first assertion.

For the third assertion, $D^a_{2,\Psi_1}$ are trivially bounded above by $\sup_{x \in \Real} |V_a''(x)| \leq D_{2,\infty} := \max(\beta,\kappa)$. As for $D^a_{1,\Psi_1}$, since $V_a'$ is Lipschitz (with constant $D_{2,\infty}$) and $\E(V_a'(X^a)) = \int V_a'(x) d\mu_a(x) = 0$, it follows by the Herbst argument (see Appendix) that:
\[
\E(\exp(\lambda |V_a'(X^a)|)) \leq \E(\exp(\lambda V_a'(X^a))) + \E(\exp(-\lambda V_a'(X^a))) \leq 2 \exp\brac{ \frac{\lambda^2 D_{2,\infty}^2}{2 \rho^a} } ~.
\]
Consequently, choosing $\lambda > 0$ small enough (uniformly in $a \in \Real$ thanks to the uniform bounds on $D_{2,\infty}$ and $\rho^a$), we see that the right-hand side is bounded by $e$, implying the asserted uniform upper bound on $D^a_{1,\Psi_1}$.

Lastly, we use the fact that $V_a$ may be written as $W_1 - W_2$, where $W_1$ is (strictly) convex, $\int \exp(-W_1(x)) dx = 1$, and $\sup W_2 - \inf W_2 \leq \omega$. Consequently, $\sup W_2 \geq 0$ and $\inf W_2 \leq 0$. We estimate:
\begin{eqnarray*}
\frac{M_3(\mu^a)}{M_2(\mu^a)^{3/2}} &=& \frac{\int |x|^3 \exp(-W_1(x) + W_2(x)) dx}{\brac{\int x^2 \exp(-W_1(x) + W_2(x)) dx}^{3/2}} \\
&\leq & \exp(\sup W_2 - \frac{3}{2} \inf W_2) \frac{\int |x|^3 \exp(-W_1(x)) dx}{\brac{\int x^2 \exp(-W_1(x)) dx}^{3/2}} ~.
\end{eqnarray*}
The first term on the right hand side above is bounded by $\exp((3/2) \omega)$, and the second term is bounded by a universal constant thanks to Theorem \ref{thm:log-concave} in the Appendix, which entails a reverse H\"{o}lder inequality for moments of log-concave measures such as $\exp(-W_1(x)) dx$. The proof is complete.
\end{proof}

Combining Lemma \ref{lem:weakly-Gaussian} with Theorem \ref{thm:uniform-model}, we obtain:
\begin{thm}
Let $\mu$ be a $(\alpha,\beta,\omega)$ weakly Gaussian measure. Then the canonical ensemble $(E_s,\abs{\cdot},\mu_{E_s})$ with mean-spin $s$ satisfies a LSI, uniformly in the system size $n \geq 2$ and mean-spin value $s \in \Real$, depending solely on a positive lower bound on $\alpha$ and upper bounds on $\beta$ and $\omega$.
\end{thm}
\begin{proof}
The theorem follows immediately from Lemma \ref{lem:weakly-Gaussian} and Theorem \ref{thm:uniform-model} when $n \geq \eta^{\beta,\omega}_{\alpha}$. For smaller $n$, it is an easy consequence of the Bakry--\'Emery criterion (see Appendix) together with the Holley--Stroock perturbation argument (\ref{eq:HS}): indeed, the strictly convex part $V_{\text{conv}}$ of the weakly-Gaussian potential implies that its restriction onto $E_s$ satisfies $LSI(\alpha)$, and the bounded perturbation $V_{\text{pert}}$ can only change the Hamiltonian by at most $n \omega$, which is bounded when $n < \eta^{\beta,\omega}_{\alpha}$.
\end{proof}

\subsection{Weakly Interacting Conservative Spin Model}

Analogously, let $\mu_{A,E_s}$ denote the conditioning of the grand canonical weakly-interacting ensemble $\mu_A$ to $S = s$, namely:
\[
d\mu_{A,E_s} = \frac{1}{Z_{A,E_s}} \exp(-H_A(x))\,  d\vol_{E_s}(x) ~,
\]
where $Z_{A,E_s} > 0$ is a normalization term. In fact, in applications, it is useful to consider our $n$-particle configuration $(x_1,\ldots,x_n)$ as a subset of a larger $N$-particle configuration $(x_1,\ldots,x_n,y_{n+1},\ldots,y_N)$, with which it interacts via an additional term $I_N(x,y)$ in the Hamiltonian:
\[
H_{A,(y)}(x) := H(x) - I_A(x) - I_N(x,y) ~,~ I_N(x,y) = \sum_{i=1}^n \sum_{j=n+1}^N a_{i,j} x_i y_j ~.
\]
Setting $b_i = - \sum_{j=n+1}^N a_{i,j} y_j$ and $b = (b_1,\ldots,b_n)$, we define:
\[
H_{A,b}(x) := H_{A,(y)}(x) = \sum_{i=1}^n (V(x_i) + b_i x_i) - I_A(x) ~,
\]
and introduce the canonical weakly-interacting ensemble with mean-spin $s$ and boundary contribution $b=(b_1,\ldots,b_n)$:
\[
d\mu_{A,E_s,b} = \frac{1}{Z_{A,E_s,b}} \exp(-H_{A,b}(x))\,  d\vol_{E_s}(x) ~.
\]

As before, note that:
\begin{equation} \label{eq:inter-invariance1}
(\mu^{\wedge u})_{A,E_s,b} = \mu_{A,E_s,b}  \;\;\; \forall u \in \Real ~.
\end{equation}
 However, contrary to the non-interacting case, the commutation with the operation of translation is not as nice. Indeed, recall that we denote the translation function on $\Real^n$ by $T^t(x) = x-t$, and note that given $t = (t_1,\ldots,t_n)$:
\begin{equation} \label{eq:inter-invariance2}
\frac{d (T^t)_*(\mu_{A,E_s,b})}{dx} = \prod_{i=1}^n \frac{d (T^{t_i})_*(\mu^{\wedge (b_i + 2 \sum_{j=1}^n a_{i,j} t_j)})}{d x_i} \frac{\exp(I_A(x))}{Z}  ~.
\end{equation}
Consequently, to apply Cram\'er's trick, we need a much more delicate argument than in the non-interacting case:

\begin{lem} \label{lem:inter-invariance}
Let $X^a$ denote a random-variable distributed according to $\mu^{a}$, and assume that $\bar{M_2} := \sup_{a \in \Real} \Var(X^a) < \infty$. Then for any $n$ by $n$ symmetric matrix with zero diagonal so that $\norm{A}_{op} < 1/(2\bar{M_2})$, and for any mean-spin $s \in \Real$ and boundary contribution $b = (b_1,\ldots,b_n) \in \Real^n$,
there exists a tilt $u_0 \in \Real$ and translation vector $(t_1,\ldots,t_n)$, so that $\frac{1}{n} \sum_{i=1}^n t_i = s$ and so that the barycenter of $(T^{t_i})_*(\mu^{\wedge (b_i+ u_0 + 2 \sum_{j=1}^n a_{i,j} t_j )})$ is at $0$ for all $i=1,\ldots,n$.
\end{lem}
\begin{proof}
Let $F(a)$ (previously denoted $s(a)$) denote the barycenter of $\mu^{\wedge a}$, and recall that this is an increasing function. Moreover, observe that $F'(a) = \Var(X^a) \in (0, \bar{M_2}]$.
We would like to show that given $A$, $s$ and $b$ as in the assumption, there exists a solution $u_0 \in \Real$, $t = (t_1,\ldots,t_n) \in \Real^n$ to the following system of $n+1$ non-linear equations:
\begin{eqnarray*}
F_i(u_0,t) := F(z_i(u_0,t)) := F(b_i + u_0 + 2 \sum_{j=1}^n a_{i,j} t_j) & = & t_i ~,~ i=1,\ldots,n ~, \\
\sum_{i=1}^n t_i & = & s n ~.
\end{eqnarray*}

Recall the definition of the hyperplane $E_s = \set{t \in \Real^n ; \sum_{i=1}^n t_i =  s n }$. Given $t \in \Real^n$, we define $u_0 = u_0(t)$ to be the unique element of $\Real$ so that 
$$G(t) := (F_1(u_0(t),t),\ldots,F_n(u_0(t),t)) \in E_s.$$
 Since $F$ is strictly increasing, it is immediate to verify that $u_0(t)$ is indeed well-defined. Our goal is to find a solution $t \in E_s$ to the equation $G(t) = t$.
We will show that under our assumptions, $G$ is a strict contraction on $E_s$, when the latter space is equipped with the induced Euclidean structure from
$(\Real^n,\abs{\cdot})$, and hence the existence of a (unique)
solution will follow immediately from Banach's fixed point theorem.

To show that $G$ contracts Euclidean distance on $E_s$, we calculate the Jacobian matrix $dG/dt = \set{\partial G_i/ \partial t_j}_{i,j=1,\ldots,n}$ at $t$:
\begin{equation} \label{eq:main-partial-matrix}
\frac{\partial G_i}{\partial t_j} = F'(z_i) \brac{ \frac{\partial u_0}{\partial t_j}  + 2 a_{i,j} } ~.
\end{equation}
Next, observe that since:
\[
\sum_{k=1}^n F(z_k(u_0(t),t)) = s n ~,
\]
differentiating in $t_j$ reveals that at $(u_0(t),t)$:
\begin{equation} \label{eq:main-partial-matrix2}
\frac{\partial u_0}{\partial t_j} \sum_{k=1}^n F'(z_k) +  \sum_{k=1}^n F'(z_k) 2 a_{k,j} = 0 ~,~ j=1,\ldots,n ~.
\end{equation}
Denoting $F'(z) = (F'(z_1),\ldots,F'(z_n)) \in \Real^n$ and combining (\ref{eq:main-partial-matrix2}) with (\ref{eq:main-partial-matrix}), we obtain in matrix form:
\begin{eqnarray*}
\nonumber \frac{dG}{dt} &=& \mathrm{Diag}(F'(z)) \brac{-\frac{1}{\norm{F'(z)}_{\ell_1}} \mathrm{Row}(F'(z))\cdot 2 A + 2 A } \\
\label{eq:main-matrix} &=& \brac{\mathrm{Diag}(F'(z))  -  \frac{1}{\norm{F'(z)}_{\ell_1}} F'(z) \otimes F'(z)} 2 A ~,
\end{eqnarray*}
where $\mathrm{Diag}(F'(z))$ denotes the diagonal matrix with $F'(z)$ as its diagonal, $\mathrm{Row}(F'(z))$ denotes the $n$ by $n$ matrix having identical rows equal to $F'(z)$, and we have used that $F' > 0$ to dispense of the absolute values in $\norm{F'(z)}_{\ell_1}$. Note that as positive definite matrices:
\[
0 \leq \mathrm{Diag}(F'(z)) \leq \norm{F'(z)}_{\ell_\infty} Id \leq \bar{M_2} Id ~,
\]
and that by H\"{o}lder's inequality:
\[
0 \leq \frac{1}{\norm{F'(z)}_{\ell_1}} F'(z) \otimes F'(z) \leq \frac{\norm{F'(z)}_{\ell_2}^2}{\norm{F'(z)}_{\ell_1}} Id \leq \norm{F'(z)}_{\ell_\infty} Id \leq \bar{M_2} Id ~.
\]
and so consequently:
\[
-\bar{M_2} Id \leq B := \mathrm{Diag}(F'(z))  -  \frac{1}{\norm{F'(z)}_{\ell_1}} F'(z) \otimes F'(z) \leq \bar{M_2} Id ~.
\]
It follows from this, (\ref{eq:main-matrix}) and our assumption that $\norm{A}_{op} < 1/(2 \bar{M_2})$ that:
\[
\norm{dG/dt}_{op}(t) \leq 2 \norm{B}_{op} \norm{A}_{op} \leq 2 \bar{M_2} \norm{A}_{op} \leq \lambda < 1 ~,~ \forall t \in E_s ~,
\]
and the desired contraction property is demonstrated, concluding the proof.
\end{proof}

Given $u = (u_1,\ldots,u_n) \in \Real^n$, let us denote:
\[
\mu^u := \otimes_{i=1}^n \mu^{u_i} ~,~
\frac{d\mu^{u}_{A,E_0}}{dx} := \frac{d\mu^u}{dx} \frac{\exp(I_A(x))}{Z^u_{A,E_0}}\, d\vol_{E_0}(x)  ~,
\]
where recall $\mu^{a}$ denotes the translated $\mu^{\wedge a}$ with barycenter at the origin. An immediate consequence of (\ref{eq:inter-invariance1}), (\ref{eq:inter-invariance2}) and Lemma \ref{lem:inter-invariance} is:
\begin{cor} \label{cor:inter-reduction-to-zero-spin}
With the assumptions of Lemma \ref{lem:inter-invariance}, for any mean-spin $s \in \Real$ and boundary contribution $b = (b_1,\ldots,b_n) \in \Real^n$, there exist tilts $u = (u_1,\ldots,u_n) \in \Real^n$ so that $(E_s,|\cdot|,\mu_{A,E_s,b})$ is isometrically isomorphic as a measure-metric space to $(E_0,|\cdot|,\mu^u_{A,E_0})$. In particular:
\[
\inf_{s \in \Real , b \in \Real^n} \rho_{LS}(E_s,|\cdot|,\mu_{A,E_s,b}) \geq \inf_{u \in \Real^n} \rho_{LS}(E_0,|\cdot|,\mu^u_{A,E_0}) ~.
\]
\end{cor}

Although the results described in Sections \ref{sec:model} and \ref{sec:inter} were proved for the case of \emph{identically distributed} independent random variables $X_1,\ldots,X_n$, each having law $\mu$, all of the proofs carry through mutatis mutandis to the case of non-identically distributed independent r.v.'s. This is thanks to the general formulation of the central tools we have used - Theorems \ref{thm:Bern}, \ref{thm:BE} and \ref{thm:Latala} - which did not assume identical distribution, only uniform upper bounds on the appropriate parameters. Consequently, all of the results of those sections carry through to the case when $\mu_n$, $\mu_{E_0}$ and $\mu_{A,E_0}$ are replaced by $\mu^u$, $\mu^u_{0,E_0}$ and $\mu^u_{A,E_0}$, respectively, when the parameters:
\[
M_2(\mu) ~,~ \frac{M_3(\mu)}{M_2(\mu)^{3/2}} ~,~ D_{1,\Psi_1}(\mu) ~,~ D^\delta_{2,\Psi_1}(\mu) ~,~ \rho(\mu) ~,
\]
are replaced by:
\[
\sup_{a \in \Real} M_2(\mu^a) ~,~ \sup_{a \in \Real} \frac{M_3(\mu^a)}{M_2(\mu^a)^{3/2}} ~,~ \sup_{a \in \Real} D_{1,\Psi_1}(\mu^a) ~,~ \sup_{a \in \Real} D^\delta_{2,\Psi_1}(\mu^a) ~,~ \inf_{a \in \Real} \rho(\mu^a) ~,
\]
respectively. Combining this with Corollary \ref{cor:inter-reduction-to-zero-spin}, we obtain:

\begin{thm} \label{thm:uniform-inter}
Let $\mu = \exp(-V(x)) \, dx$ be a probability measure on $\Real$ with $V \in C^2(\Real)$ and $\lim_{x \rightarrow \pm \infty} V(x) = \infty$. Fixing $\delta > 0$, denote $\rho^a := \rho_{LS}(\Real,\abs{\cdot},\mu^a)$, $M_p^a := M_p(\mu^a)$, $D_{1,\Psi_1}^a := D_{1,\Psi_1}(\mu^a)$ and $D_{2,\Psi_1}^a := D_{2,\Psi_1}^\delta(\mu^a)$. Assume that:
\[
\rho^a \geq \bar{\rho} > 0 ~,~ \frac{M_3^a}{(M_2^a)^{3/2}} \leq \bar{M} < \infty ~,~ D_{1,\Psi_1}^a \leq \bar{D_1} < \infty ~,~ D_{2,\Psi_1}^a \leq \bar{D_2} < \infty ~,
\]
uniformly in $a \in \Real$, and that:
\[
-\kappa := \inf_{x \in \Real} V''(x) \geq -\frac{\bar{\rho}}{8} ~.
\]
Let $A$ denote an $n$ by $n$ symmetric matrix with zero diagonal satisfying:
\begin{equation} \label{eq:inter-general-Aop-small}
\norm{A}_{op} \leq c \bar{\rho} ~,
\end{equation}
for some appropriately chosen universal constant $c > 0$. Then for any $n \geq \eta^{\bar{M},\bar{D_1},\bar{D_2}}_{\bar{\rho},\delta}$,
the canonical ensemble with weak-interaction $(E_s,\abs{\cdot},\mu_{A,E_s,b})$ satisfies LSI, uniformly in the system size $n$, mean-spin value $s \in \Real$ and boundary contribution $b \in \Real^n$, depending solely on $\bar{\rho}$, $\bar{M}$, $\bar{D_1}$, $\bar{D_2}$ and $\norm{A}_{HS} / \bar{\rho}$.
\end{thm}
\begin{proof}
Choosing the constant $c>0$ in (\ref{eq:inter-general-Aop-small}) small enough, the sub-Gaussian tail decay (\ref{eq:sub-Gaussian}) guaranteed by the Hebst argument (since $\E(X^a) = 0$) ensures that:
\[
\norm{A}_{op} \leq c \bar{\rho} \leq c \frac{C}{\bar{M_2}} < \frac{1}{2 \bar{M_2}} ~.
\]
Consequently, the assumptions of Corollary \ref{cor:inter-reduction-to-zero-spin} are satisfied, and so it is enough to verify a uniform lower bound on $\rho_{LS}(E_0,\abs{\cdot},\mu^{u}_{A,E_0})$ in $u \in \Real^n$. We furthermore require that $c>0$ in (\ref{eq:inter-general-Aop-small}) is smaller than what is required in (\ref{eq:Aop-small-simple2}) to apply Theorem \ref{thm:inter-zero-spin}. Write $\mu^a = \exp(-V_a(x)) dx$, and
apply Theorem \ref{thm:inter-zero-spin} to $\mu^u$, which is possible thanks to the requirement that $-\bar{\rho} / 8 \leq -\kappa \leq V_a''$ for all $a \in \Real$. Note again that $\bar{M_2} \leq C / \bar{\rho}$ by the sub-Gaussian tail decay (\ref{eq:sub-Gaussian}).
The proof is complete.
\end{proof}

Combining Lemma \ref{lem:weakly-Gaussian} with Theorem \ref{thm:uniform-inter}, we obtain:
\begin{thm}
Let $\mu$ be a $(\alpha,\beta,\omega)$ weakly Gaussian measure, and let $A$ denote an $n$ by $n$ symmetric matrix with zero diagonal satisfying:
\[
\norm{A}_{op} \leq c \alpha \exp(-\omega) ~,
\]
for an appropriate universal constant $c > 0$. Then the canonical ensemble $(E_s,\abs{\cdot},\mu_{A,E_s,b})$ with weak-interaction $A$, satisfies a LSI, uniformly in the system size $n \geq 2$, mean-spin value $s \in \Real$ and boundary contribution $b \in \Real^n$, depending solely on a positive lower bound on $\alpha$ and upper bounds on $\beta$, $\omega$ and $\norm{A}_{HS} / (\alpha \exp(-\omega))$.
\end{thm}
\begin{proof}
The theorem follows from Lemma \ref{lem:weakly-Gaussian} and Theorem \ref{thm:uniform-inter} when $n \geq \eta^{\beta,\omega}_{\alpha}$. For smaller $n$, it is an easy consequence of the Bakry--\'Emery criterion (see Appendix) together with the Holley--Stroock perturbation argument (\ref{eq:HS}): indeed, note that
$\norm{A}_{op} \leq c \alpha$, and so the strictly convex part $V_{\text{conv}}$ of the weakly-Gaussian potential, is still strictly convex  even after adding the linear boundary term and the weak interactions if we assume that $c < 1$, since $\alpha Id - A \geq (1-c) \alpha Id$ as positive definite matrices. Consequently, so is the restriction onto $E_s$, and by the Bakry--\'Emery criterion, satisfies $LSI((1-c)\alpha)$. The bounded perturbation $V_{\text{pert}}$ can only change the Hamiltonian by at most $n \omega$, which is bounded when $n < \eta^{\beta,\omega}_{\alpha}$.
\end{proof}

\section{Spectral-Gap of Conservative Spin Model with Convex Potential} \label{sec:SG}

We now turn to study the \emph{spectral-gap} of a canonical ensemble, having a \emph{convex potential}. The prime example we have in mind, which was suggested to us by Pietro Caputo, is the two-sided exponential measure $\nu = \frac{1}{2} \exp(-|x|) \, dx$. More generally, recall that a measure $\mu = \exp(-V(x)) \, dx$ on $\Real^n$ is called log-concave when $V : \Real^n \rightarrow \Real \cup \set{\infty}$ is convex. 
Note that for a one-dimensional  log-concave probability measure $\mu$ on $\Real$, it is known (see Theorem \ref{thm:log-concave} in the Appendix) that:
\begin{equation} \label{eq:SG-Var}
 \rho_{SG}(\Real,|\cdot|,\mu) \simeq \frac{1}{\Var(\mu)} ~,
\end{equation}
where recall we use $A \simeq B$ to denote that $c_1 \leq A/B \leq c_2$ for some two universal constants $c_1,c_2 > 0$. 

\subsection{Zero Mean-Spin Case - Bounded System Size}

First, we take care of the zero mean-spin case when the system size $n$ is bounded above:

\begin{prop} \label{prop:SG-small-n}
Let $\mu$ denote a log-concave probability measure on $\Real$ with barycenter at $0$. Then for any integer $n \geq 2$:
\[
 C \, \rho_{SG}(\Real,|\cdot|,\mu) \ge \rho_{SG}(E,|\cdot|,\mu_E) \geq \frac{c}{n} \, \rho_{SG}(\Real,|\cdot|,\mu) ~,
\]
where $C,c>0$  are  universal constants.
\end{prop}

The recent paper \cite{BartheCorderoVariance} gives a better  estimate involving a factor $\log(n)^{-2}$
instead of $n^{-1}$. However this improvement is not essential for our purpose. The proof of the proposition is based on the following:

\begin{lem} \label{lem:SG-small-n}
Let $\mu$ denote a log-concave probability measure on $\Real$ with barycenter at $0$.
Let $X = (X^1,\ldots,X^n)$ denote the random vector in $\Real^n$ distributed according to the product measure $\mu_n$, and let 
$X_E = (X_E^1,\ldots,X_E^n)$ denote its conditioning on the hyperplane $\sum_{i=1}^n X^i = 0$, i.e. having law $\mu_{E}$. Then:
\[
\Var(X^1) \simeq \Var(X_E^1) ~.
\]
In other words, the variance of the original measure $\mu$ is universally equivalent to the variance of the marginal of the conditioned measure $\mu_E$. 
\end{lem}
\begin{proof}
Write $\mu = f(x)\,  dx$. As usual, we denote the law of $\frac{1}{\sqrt{n}} \sum_{i=1}^n X^i$ by $\mu^D_n = g_n(x) dx$, and let $\mu_E^1 = h_n(x) dx$ denote the law of $X_E^1$. 
Since $\mu_n$ is log-concave, then so is $\mu_E$, and by the Pr\'ekopa--Leindler Theorem (e.g. \cite{GardnerSurveyInBAMS}), so are their marginals $\mu^D_n$ and $\mu_E^1$. In addition, the barycenter of $\mu^D_n$ is at the origin (by linearity of the projection), and so is the barycenter of $\mu_E^1$ (by symmetry: $\E(X_E^1) = E(X^1 | \sum_{i=1}^n X^i = 0) = 0$).
 
 It is known (see Theorem \ref{thm:log-concave}) that  for a log-concave probability measure $\nu = w(x) dx$ with barycenter at the origin:
 \begin{equation} \label{eq:SG-small-n-fact}
 \Var(\nu) \simeq \frac{1}{w(0)^2} ~.
 \end{equation}
 Consequently, we just need to show that $f(0) \simeq h_n(0)$. Indeed:
 \begin{eqnarray*}
 h_n(0) &=& \lim_{\eps \rightarrow 0+} \frac{1}{2\eps} \P(|X_E^1| \leq \eps) = \lim_{\eps,\delta \rightarrow 0+} \frac{1}{2\eps} \P\brac{ |X^1| \leq \eps \; \left | \; \abs{\sum_{i=1}^n X^i} \leq \delta \right . } \\
 &= &\lim_{\eps,\delta \rightarrow 0+} \frac{1}{2\eps} \frac{\P\brac{ |X^1| \leq \eps \wedge \abs{\sum_{i=1}^n X^i} \leq \delta }}{\P\brac{ \abs{\sum_{i=1}^n X^i} \leq \delta }} = f(0) \lim_{\delta \rightarrow 0+} \frac{\P\brac{ \abs{\sum_{i=2}^n X^i} \leq \delta }}{\P\brac{ \abs{\sum_{i=1}^n X^i} \leq \delta }} \\
 & = & f(0) \lim_{\delta \rightarrow 0+} \frac{\P\brac{ \abs{\frac{1}{\sqrt{n-1}}\sum_{i=2}^n X^i} \leq \frac{\delta}{\sqrt{n-1} }}}{\P\brac{ \abs{\frac{1}{\sqrt{n}}\sum_{i=1}^n X^i} \leq \frac{\delta}{\sqrt{n}} }} = f(0) \frac{\sqrt{n}}{\sqrt{n-1}} \frac{g_{n-1}(0)}{g_{n}(0)} ~,
 \end{eqnarray*}
and so it remains to show that $g_{n-1}(0) \simeq g_n(0)$. Since $g_k$ is log-concave itself, it follows from (\ref{eq:SG-small-n-fact}) that for any integer $k \geq 1$:
 \[
 \frac{1}{g_k(0)^2} \simeq \Var(\mu_k^D) = \Var\brac{\frac{1}{\sqrt{k}} \sum_{i=1}^k X^i} = \Var(X^1) ~,
 \]
 thereby completing the proof.
 \end{proof}

\begin{proof}[Proof of Proposition \ref{prop:SG-small-n}]
Since $\mu_E$ is log-concave, it follows by the isoperimetric bound of Kannan--Lov\'asz--Simonovits \cite{KLS} coupled with Cheeger's inequality \cite{CheegerInq} (see e.g. \cite{EMilman-RoleOfConvexity} for more information) that:
\[
\rho_{SG}(E,|\cdot|,\mu_E) \geq \frac{c}{\int |x|^2 d\mu_E(x)} = \frac{c}{\E(\sum_{i=1}^n (X^i_E)^2)} = \frac{c}{n \Var(X^1_E)}  ~,
\]
for some universal constant $c > 0$. Appealing to Lemma \ref{lem:SG-small-n} and (\ref{eq:SG-Var}), we observe that $\Var(X^1_E) \simeq \Var(X^1) \simeq 1 / \rho_{SG}(\Real,|\cdot|,\mu)$, thereby concluding the proof of the right-hand side inequality.

 Denoting as usual by $X_E$ the random vector distributed according to $\mu_E$, and testing the spectral-gap inequality on $(E,|\cdot|,\mu_E)$ with the function $E \ni x = (x_1,\ldots,x_n) \mapsto f(x) = x_1$, we have by Lemma \ref{lem:SG-small-n} and (\ref{eq:SG-Var}) that:
\[
\rho_{SG}(E,|\cdot|,\mu_{E}) \leq \frac{\E |\nabla f (X_E) |^2}{\Var(f(X_E))} = \frac{n-1}{n} \frac{1}{\Var(X_E^1)} \simeq \frac{1}{\Var(X^1)} \simeq \rho_{SG}(\Real,|\cdot|,\mu) ~,
\]
as asserted.
\end{proof}

\begin{rem}\label{rem:KLS}
In \cite{KLS}, Kannan, Lov\'asz and Simonovits proposed a daring conjecture, now commonly referred to as the \emph{KLS conjecture}. It predicts that
for log-concave probability measures on any Euclidean space, the spectral-gap can be evaluated up to dimension free constants
by just testing the inequality on linear functions.
In the case of the measure $\mu_E$, since the coordinates play symmetric roles,
the covariance of $\mu_E$ is a multiple of the identity of $E$. Hence all (non-zero) linear functions give rise to the same Rayleigh quotient.
Therefore if $\mu$ is log-concave, and with the notation $f(x)=x_1$  of the latter proof, the KLS conjecture predicts that:
$$ 
 \rho_{SG}(E,|\cdot|,\mu_{E}) \simeq \frac{\E |\nabla f (X_E) |^2}{\Var(f(X_E))} = \frac{n-1}{n} \frac{1}{\Var(X_E^1)} \simeq \frac{1}{\Var(X^1)} \simeq \rho_{SG}(\Real,|\cdot|,\mu) ~.
 $$
\end{rem}

\subsection{Zero Mean-Spin Case - General System Size}

\begin{thm} \label{thm:SG-model-zero-spin}
Let $\mu = \exp(-V(x)) dx$ denote a log-concave probability measure with barycenter at $0$ so that $V \in C^2(\Real)$ and $D_{1,\Psi_1} = D_{1,\Psi_1}(\mu) < \infty$, and let $\rho = \rho_{SG}(\Real,|\cdot|,\mu)$. Then for any integer $n \geq 2$:
\[
\rho_{SG}(E,|\cdot|,\mu_E) \geq c \frac{\rho}{\brac{ 1 + \log(Q)}^2} ~,
\]
where $c>0$ is a  universal constant and $Q$ is the following scale-invariant quantity:
\begin{equation} \label{eq:SG-Q}
Q:= \max(1,\sqrt{\Var(\mu)}  D_{1,\Psi_1}(\mu)) ~.
\end{equation}
\end{thm}

\begin{proof} It is well-known (see the Appendix) that the spectral-gap inequality tensorizes with respect to the Euclidean ($\ell_2$) norm, and so $(\Real^n,|\cdot|,\mu_n)$ satisfies a spectral-gap inequality with the same constant $\rho$. Since $d\mu_n(x) = \exp(-H(x))\,  dx$ with $H(x) = \sum_{i=1}^n V(x_i)$, it follows that $ \mathrm{Hess} H \geq 0$ as a tensor field in $\Real^n$. The same bound holds for its restriction onto any linear subspace, and so it follows that $(E,|\cdot|,\mu_E)$ satisfies our convexity assumptions. Moreover, the uniform thickening of $\mu_E$ in the direction $D$ orthogonal to $E$ only adds $0$ as an eigenvalue to the Hessian matrix in that direction, and hence $(\Real^n,|\cdot|,\mu_{E,w})$ also satisfies our convexity assumptions for any $w > 0$.

We now transfer the spectral-gap inequality on $(\Real^n,|\cdot|,\mu_n)$ onto $(\Real^n,|\cdot|,\mu_{E,w_0})$ by applying Theorem \ref{thm:SG-trans} with $p=4$. 
Note that necessarily $\lambda := \norm{d\mu/dx}_{L^\infty} < \infty$ and that $M_3 / M_2^{3/2}$ and $M_3 \lambda^3$ are bounded above by universal constants, thanks to Theorem \ref{thm:log-concave}. Applying Corollary \ref{cor:conc-trans-model-2} with $\kappa = 0$, it follows that whenever $n$ exceeds some universal constant $\eta_0$, the following estimate holds:
\[
\brac{\int_{\Real^n} \brac{\frac{d\mu_{E,w_0}}{d\mu_n}}^4 d\mu_n}^{1/4} \leq C \max(1,\sqrt{M_2} D_{1,\Psi_1}) = C Q ~,
\]
where $C>0$ is some universal constant and:
\[
w_0  := \min\brac{\sqrt{M_2},\frac{1}{D_{1,\Psi_1}}} ~.
\]
Theorem \ref{thm:SG-trans} therefore implies that:
\[
\rho_{SG}(\Real^n,|\cdot|,\mu_{E,w_0}) \geq c \frac{\rho}{\brac{ 1 + \log(Q)}^2} ~,
\]
for some universal constant $c>0$. By the tensorization property of the spectral-gap inequality:
\[
\rho_{SG}(\Real^n,|\cdot|,\mu_{E,w_0}) = \min(\rho_{SG}(E,|\cdot|,\mu_E),\rho_{SG}(\Real,|\cdot|,\nu_{[-w_0,w_0]})) ~,
\]
where $\nu_{[-w_0,w_0]}$ denotes the uniform measure on $[-w_0,w_0]$. It follows that when $n \geq \eta_0$:
\[
\rho_{SG}(E,|\cdot|,\mu_E) \geq \rho_{SG}(\Real^n,|\cdot|,\mu_{E,w_0}) \geq  c \frac{\rho}{\brac{ 1 + \log(Q)}^2}   ~,
\]
as asserted. The case when $n < \eta_0$ is handled by Proposition \ref{prop:SG-small-n}. This concludes the proof. 
\end{proof}

\begin{rem} \label{rem:D12}
Observe that the quantity $D_{1,\Psi_1}(\mu)$ may be infinite, e.g. for $V(x)=e^{x^2}+c$, and so in particular 
the scale invariant quantity $\sqrt{\Var(\mu)} D_{1,\Psi_1}(\mu)$  appearing in (\ref{eq:SG-Q}) is not bounded above in the class of log-concave probability measures. However, this quantity is always bounded away from $0$: 
to see this, set $D_{1,2}(\mu) := \norm{V'}_{L^2(\mu)}$, use $2\exp(|t|)\ge t^2$, Cauchy-Schwartz and integrate by parts: 
\begin{eqnarray*}
\sqrt{2e} \, \sqrt{\Var(\mu)} D_{1,\Psi_1}(\mu) &\ge&   \sqrt{\Var(\mu)} D_{1,2}(\mu) 
=  \left(\int_{-\infty}^{\infty} t^2 d\mu(t) \, \int (V')^2 d\mu\right)^{\frac12}\\  &\ge& \int_{-\infty}^{\infty} t V'(t) \, d\mu(t) =1 ~.
\end{eqnarray*}
Lastly, we mention that even the scale invariant quantity $\sqrt{\Var(\mu)} D_{1,2}(\mu)$ is not uniformly bounded above in the class 
of log-concave probability measures: for potentials of the form $V(x)=|x|^p+c_p$, it is equivalent  (up to constants) to $p^{1/2}$ when $p$ is large. 
\end{rem}

\begin{cor} \label{cor:SG-mean-zero}
Let $\mu = \exp(-V(x)) \, dx$ denote a log-concave probability measure with barycenter at $0$. Assume that either:
\begin{enumerate}
\item $V$ is Lipschitz with constant $L$. 
\item $V \in C^1(\Real)$ and $V'$ is Lipschitz with constant $L^2$. 
\end{enumerate}
Let $\rho = \rho_{SG}(\Real,|\cdot|,\mu)$. Then for any integer $n \geq 2$:
\[
C \rho \ge \rho_{SG}(E,|\cdot|,\mu_E) \geq c \frac{\rho}{\log(2 + L^2/\rho)^2} ~,
\]
where $C,c>0$  are  universal constants. 
\end{cor}
\begin{proof}
It is easy to verify that it suffices to prove the claim when $V \in C^2(\Real)$. Indeed, any convex $V$ as in the first (respectively, second) case may be approximated in the maximum norm by convex functions $V_m \in C^2(\Real)$ so that the $\norm{V_m'}_{L^\infty}$ (respectively, $\norm{V_m''}_{L^\infty}$) converges to at most $L$ (respectively, $L^2$). Since the spectral-gap is stable under convergence of the potential in the maximum norm, the reduction follows.

Consequently, by Theorem \ref{thm:SG-model-zero-spin}, it is enough to bound $Q$ given by (\ref{eq:SG-Q}) from above. First, observe that $c_1 \leq \Var(\mu) \rho \leq c_2$ according to Theorem \ref{thm:log-concave}. In the first case, we obviously have $D_{1,\Psi_1}(\mu) = \norm{V'(X_0)}_{\Psi_1} \leq \norm{V'(X_0)}_{L^\infty} \leq L$, and hence $Q \leq \max(1,\sqrt{c_2} L / \sqrt{\rho})$. In the second case, we know by the result of M. Gromov and V. Milman \cite{GromovMilmanLevyFamilies} that spectral-gap implies exponential concentration of Lipschitz functions about their mean, and since $V'$ is $L^2$-Lipschitz with $\E(V'(X_0)) = 0$, it follows that $D_{1,\Psi_1}(\mu) = \norm{V'(X_0)}_{\Psi_1} \leq C L^2/\sqrt{\rho}$, and hence $Q \leq \max(1,\sqrt{c_2} C  L^2 / \rho)$. In either case,
the asserted lower  bound follows from Theorem \ref{thm:SG-model-zero-spin}. The upper bound follows from Proposition~\ref{prop:SG-small-n}. 
\end{proof}

\subsection{Dependence on Mean-Spin}

Given a log-concave probability measure $\mu = \exp(-V(x)) \, dx$ on $\Real$, we denote $a_{\pm} = \lim_{x \rightarrow \pm \infty} (V(x)-V(0)) / x$ (the latter function is monotone and hence the limits exist in the wide sense). Consequently, the probability measure $\mu^{\wedge a}$ is well defined for all $a \in (a_{-},a_{+})$, and as usual, we denote by $X^{\wedge a}$ the random variable distributed according to $\mu^{\wedge a}$. Again, we denote $s(a) := \E(X^{\wedge a})$, which always exists since $\mu^{\wedge a}$ is still a log-concave probability measure, and hence has exponential tail-decay. The function $(a_{-},a_{+})\ni a \mapsto s(a) \in \mathrm{isupp}(\mu)$ is increasing and onto, and we denote its inverse by $a(s)$.

Applying the Cram\'er trick as in Section \ref{sec:mean-spin}, we immediately see that:
\begin{equation} \label{eq:SG-Cramer}
 \rho_{SG}(E_s,|\cdot|,\mu_{E_s}) = \rho_{SG}(E_0,|\cdot|,\mu^{a(s)}_{E_0})  \;\;\; \forall s \in \mathrm{isupp}(\mu) ~,
\end{equation}
where recall $\mu^{a}$ is the translation of $\mu^{\wedge a}$ having barycenter at the origin. As a consequence of Corollary \ref{cor:SG-mean-zero}, we obtain:
\begin{thm} \label{thm:SG-any-mean}
Let $\mu = \exp(-V(x))\,  dx$ denote a log-concave probability measure on $\Real$. Assume that either:
\begin{enumerate}
\item $V$ is Lipschitz with constant $L$. 
\item $V \in C^1(\Real)$ and $V'$ is Lipschitz with constant $L^2$. 
\end{enumerate}
Given $s \in \Real$, denote $\rho_s := \rho_{SG}(\Real,|\cdot|,\mu^{a(s)})$. Then for any integer $n \geq 2$ and mean-spin $s \in \Real$:
\begin{equation} \label{eq:SG-bound}
C\rho_s \ge \rho_{SG}(E_s,|\cdot|,\mu_{E_s}) \geq c \frac{\rho_s}{\log(2 + L^2/\rho_s)^2} ~,
\end{equation}
where $C, c>0$ are  universal constants.
\end{thm}
\begin{proof}
Write $\mu^{a} = \exp(-V_a(x)) dx$ for $a \in (a_{-},a_{+})$. In the first case, since obviously $|a_{-}|,a_{+} \leq L$, then $V_a$ is Lipschitz with constant $2L$. In the second case, since by definition $V_a(x-s(a)) = V(x) - a x + c_a$, we see that the Lipschitz constant of $V_a'$ is identical to that of $V'$. In either case, the assertion follows from (\ref{eq:SG-Cramer}) and Corollary \ref{cor:SG-mean-zero}. 
\end{proof}

\begin{rem} \label{rem:KLS2}
The left-hand bound in (\ref{eq:SG-bound}), namely $C  \rho_{SG}(\Real,|\cdot|,\mu^{a(s)}) \ge \rho_{SG}(E_s,|\cdot|,\mu_{E_s})$, is true without any restriction on the log-concave measure $\mu$, as follows from Proposition~\ref{prop:SG-small-n} and  \eqref{eq:SG-Cramer}. 
The KLS conjecture (see Remark~\ref{rem:KLS}) predicts that the logarithmic term in the right-hand bound in 
\eqref{eq:SG-bound} and the technical restrictions on the log-concave measure $\mu$ may be removed. 
\end{rem}

We conclude this subsection with an estimate of the one-dimensional spectral-gaps appearing in the previous results.
It is optimal up to constants, as witnessed by the example of the two-sided exponential measure, studied in the next subsection.

\begin{prop}
Let $\mu $ be a log-concave measure on $\Real$ and let $X$ be a random variable distributed according to $\mu$.
Then for all $s\in \mathrm{isupp}(\mu)$:
$$  \rho_{SG}(\Real,|\cdot|,\mu^{a(s)}) \ge \frac{c}{\Var(X)+s^2}~,$$
where $c>0$ is a universal constant.
\end{prop}
\begin{proof}
Assume as we may that 0 is the barycenter of $\mu$. By (\ref{eq:SG-Var}), our goal is to prove that for all $s\in \mathrm{issup}(\mu)$:
$$ \Var(X^{\wedge a(s)}) \le C \big( \Var(X)+ s^2 \big) ~,$$
for some universal constant $C$. This is equivalent to showing that for all $a\in (a_-,a_+)$:
\begin{equation}\label{eq:goal-s}
 \Var(X^{\wedge a}) \le C \big( \Var(X)+ s(a)^2 \big) = C \big( \Var(X)+ (\E X^{\wedge a})^2 \big) ~.
 \end{equation}
It suffices to prove this for all $a\in [0, a^+)$, as we may apply it to $\mu\circ (-\mathrm{Id})$.
We claim that for all $a\in[0,a^ +)$, the following holds:
\begin{equation} \label{eq:claim-s}
\int |t| e^{ta}d\mu(t) \le \sqrt{\E(X^ 2)}\int  e^{ta} d\mu(t)+\int t e^{ta} d\mu(t) ~.
\end{equation}
Indeed, this is true for $a=0$ since $\E|X|\le \sqrt{\E(X^ 2)}$ and $\int t \, d\mu(t)=\E X=0$.
Now set $\varphi(a)=\int t e^{ta} d\mu(t)$, and note that $\varphi(0)=0$ and $\varphi'(a)=\int t^2 e^{ta} d\mu(t) \ge 0$. Hence for all $a\in [0, a^+)$,
$\varphi(a)\ge 0$. For $a\in[0,a^+)$, the derivative of the right-hand side of \eqref{eq:claim-s} is equal to:
  $$ \sqrt{\E(X^ 2)} \varphi(a) + \int t^2 e^{ta}d\mu(t) ~.$$
This no less than the derivative of the left-hand side of \eqref{eq:claim-s}: $\int t^2 \mathrm{sign}(t) e^{ta}d\mu(t)$.  
Hence we have proved \eqref{eq:claim-s}, which can be rephrased as  $\E |X^{\wedge a}| \le   \sqrt{\E(X^ 2)} + \E X^{\wedge a}$.
Using $(a+b)^2\le 2(a^2+b^2)$ and $ \E \big((X^{\wedge a})^2\big) \leq C (\E |X^{\wedge a}|)^2 $ (see Theorem \ref{thm:log-concave}), we obtain:
$$  \E \big((X^{\wedge a})^2\big)\le C' \left(\E(X^2)+ \big( \E X^{\wedge a} \big)^2 \right) ~.$$
This obviously implies \eqref{eq:goal-s} and consequently completes the proof.
\end{proof}

\subsection{The Two-Sided Exponential Measure}

In the special case of the two-sided exponential measure $\nu = \frac{1}{2} \exp(-|x|)\,  dx$, we can indeed confirm the KLS conjecture for $(E_s,|\cdot|,\nu_{E_s})$. Note that it is easy to check that if $a \in (-1,1)$
then 
$\Var(\mu^{a}) \simeq (1 + s(a)^2)$, and therefore $\rho_{SG}(\Real,|\cdot|,\mu^{a(s)}) \simeq 1 / (1 + s^2)$. 

\begin{thm}
Let $\nu = \frac{1}{2} \exp(-|x|)\,  dx$. Then the canonical ensemble $(E_s,|\cdot|,\nu_{E_s})$ satisfies:
\[
\rho_{SG}(E_s,|\cdot|,\nu_{E_s})  \simeq \frac{1}{1+s^2} ~,
\]
uniformly in the system-size $n \geq 2$ and mean-spin $s \in \Real$. 
\end{thm}
\begin{proof}
The upper bound on $\rho_{SG}(E_s,|\cdot|,\nu_{E_s})$ follows e.g. from Remark~\ref{rem:KLS2}. As for the lower bound,
since $V(x) = |x| + \log(2)$ is $1$-Lipschitz on $\Real$, Theorem \ref{thm:SG-any-mean} implies that:
\[
\rho_{SG}(E_s,|\cdot|,\nu_{E_s}) \geq \frac{c}{\brac{(2+|s|) \log(2+|s|)}^2} ~.
\]
Consequently, the asserted lower bound follows when $|s| < C_2$, for any fixed constant $C_2>0$, and it remains to verify it (by symmetry) when $s \geq C_2$.

Observe that $\nu_{E_s}$ has constant density on the simplex $\Delta_{n,s}$, defined as:
\[
\Delta_{n,s} := E_s \cap \Real^n_+ ~,
\]
where $\Real^n_+$ denotes the positive orthant. Denoting by $\lambda_{\Delta_{n,s}}$ the uniform measure on $\Delta_{n,s}$, we conclude that:
\[
\lambda_{\Delta_{n,s}} = \frac{\nu_{E_s}|_{\Delta_{n,s}}}{\nu_{E_s}(\Delta_{n,s})} ~.
\]
It is well-known (see e.g. \cite{BartheWolffGammaDistributions}) that uniformly in $n \geq 2$:
\[
\rho_{SG}(\Delta_{n,1},|\cdot|,\lambda_{\Delta_{n,1}}) \geq c ~,
\]
where $c>0$ is a universal constant, and hence by scaling:
\[
\rho_{SG}(\Delta_{n,s},|\cdot|,\lambda_{\Delta_{n,s}}) \geq \frac{c}{s^2} ~.
\]
We will deduce the desired assertion by transferring this bound on the spectral-gap from $\lambda_{\Delta_{n,s}}$ onto the entire $\nu_{E_s}$. Since both measures are log-concave on $E_s$, we employ the transference principle given by case 2 of Theorem \ref{thm:old-SG-trans}, which requires control over $\norm{d\lambda_{\Delta_{n,s}} / d\nu_{E_s}}_{L^\infty}$. We conclude that the remaining part of the assertion will follow if we show for instance that:
\begin{equation} \label{eq:simplex}
\nu_{E_s}(\Delta_{n,s}) \geq 1/2 \;\;\; \forall n \geq 2 \;\;\; \forall s \geq C_2 > 0 ~.
\end{equation}
This boils down to a simple calculation, which we now verify.

Indeed, denote:
\[
\Lambda(n-1,r) := \set{ x \in \Real^{n-1}_+ ; \sum_{i=1}^{n-1} x_i \leq r}  ~.
\]
Now write:
\[
\frac{\nu_{E_s}(E_s \setminus \Delta_{n,s})}{\nu_{E_s}(\Delta_{n,s})} = \frac{J \frac{1}{2^n} \int_{\Real^{n-1} \setminus \Lambda_{n-1,s n}} \exp\big(- \sum_{i=1}^{n-1} |x_i| - |s n - \sum_{i=1}^{n-1} x_i|\big) dx}{J \frac{1}{2^n} \int_{\Lambda_{n-1,s n}} \exp(- sn) dx} ~,
\]
where $J$ is an appropriate Jacobian factor. It is immediate to verify that $Vol(\Lambda(m,1)) = 1 / m!$, and hence by scaling:
\begin{eqnarray*}
&=&  \frac{(n-1)!}{(sn)^{n-1}} \int_{\Real^{n-1} \setminus \Lambda_{n-1,s n}} \exp\left( s n -  \sum_{i=1}^{n-1} |x_i| - \Big|s n - \sum_{i=1}^{n-1} x_i\Big|\right) dx \\
&=& \frac{(n-1)!}{(sn)^{n-1}} \int_{\Real^{n-1} \setminus \Lambda_{n-1,s n}} \exp\left(- 2 \Big(s n - \sum_{i=1}^{n-1} x_i\Big)_{-} - \sum_{i=1}^{n-1} 2 (x_i)_{-}\right) dx ~, \\
\end{eqnarray*}
where $a_{-} := (|a| - a)/2 = \max(-a,0)$. Since:
\[
\Real^{n-1} \setminus \Lambda_{n-1,s n} \subset \bigcup_{i=1}^{n-1} \set{x_i \leq 0} \cup \set{\sum_{i=1}^{n-1} x_i \geq s n} ~,
\]
it follows by the union bound that:
\[ 
\frac{\nu_{E_s}(E_s \setminus \Delta_{n,s})}{\nu_{E_s}(\Delta_{n,s})} \leq \frac{(n-1)!}{(sn)^{n-1}} \; n \int_0^\infty e^{-2 t}dt = \frac{n!}{2 (sn)^{n-1}} ~.
\]
In particular, the latter ratio is bounded above by $1$ (in fact, $1/2$) whenever e.g. $s \geq 1$, yielding (\ref{eq:simplex}), as desired, thereby concluding the proof. 
\end{proof}

\subsection{Further Remarks}

As witnessed by Theorems \ref{thm:old-SG-trans}, \ref{thm:SG-trans} and \ref{thm:SG-trans2}, there are many possibilities for transferring the spectral-gap estimate from the product measure $\mu_n$ onto the thickened conditioned measure $\mu_{E,w}$. We chose to present above the most convenient possibility for handling the two-sided exponential measure, which was to employ Theorem \ref{thm:SG-trans} coupled with the estimate already obtained in Proposition \ref{prop:conc-trans-model-2}, leading to a dependence on the parameter $D_{1,\Psi_1}(\mu)$. However, several other possibilities, each having its own advantages and disadvantages, are possible. In particular, we mention that it is possible to naively estimate $\mu_n\set{\frac{d\mu_{E,w}}{d\mu_n } \geq t}$ by repeating the argument of Proposition \ref{prop:conc-trans-model-2} and simply using Chebyshev's inequality (instead of Bernstein's Theorem \ref{thm:Bern}), resulting in a dependence on the parameter $D_{1,2} = \norm{V'}_{L^2(\mu)} = \sqrt{\norm{V''}_{L^1(\mu)}}$, instead of the more complicated $D_{1,\Psi_1}(\mu)$ (cf. Remark \ref{rem:D12}). However, this does not lead to a bound on $d_{TV}(\mu_n,\mu_{E,w})$, since Lemma \ref{lem:weak-TV} requires control over $\mu_n\set{(1_{\abs{\pi_D(x)} \leq w} d\mu_{n})/d\mu_{E,w} \geq t}$, to which end one would also need to control $D^\delta_{2}(\mu)$. And controlling $V''$ from above is problematic since this prevents approximating non-smooth densities such as that of the two-sided exponential measure.

Consequently, we also point out another method for directly handling non $C^2$ densities, by proposing an alternative to the second-order Taylor expansion of $V(x + \eps) $ which was crucially used in the proof of Proposition \ref{prop:conc-trans-model}.

\begin{dfn*}
Let $V : \Real \rightarrow \Real$ denote a locally Lipschitz function. Given $\eps \neq 0$, define:
\begin{itemize}
\item $W^\eps (x) = V(x+\eps) - V(x)$.
\item $W^\eps_1(x) = \frac{1 - \exp(-W^\eps(x))}{\eps}$.
\item $W^\eps_2(x) = \frac{\exp(-W^\eps(x)) - 1 + W^\eps(x)}{\eps^2/2}$.
\end{itemize}
\end{dfn*}

It is immediate to verify that:
\[
V(x + \eps) - V(x) = \eps W^\eps_1(x) + \frac{\eps^2}{2} W^\eps_2(x) ~,
\]
and that:
\[
\int W^\eps_1(x) \exp(-V(x)) dx = 0 \;\;\; \forall \eps \neq 0 ~.
\]
It is also easy to check that if $V$ is Lipschitz with constant $L$, then:
\begin{equation} \label{eq:global-L}
\forall \abs{\eps} \in (0,1/L] \;\;\; \forall x \in \Real \;\;\;\;\; \abs{ W^\eps_1(x) } \leq C L  ~,~ \abs{ W^\eps_2(x)} \leq C L^2 ~,
\end{equation}
for some universal constant $C>0$. If $V$ is only locally Lipschitz, then we have the less useful:
\begin{equation} \label{eq:LexpL}
\forall \abs{\eps} \in (0,1] \;\;\; \forall x \in \Real \;\;\;\;\; \abs{ W^\eps_1(x) } \leq L(x) \exp(L(x))  ~,~ \abs{ W^\eps_2(x)} \leq 2 L(x)^2 \exp(L(x)) ~,
\end{equation}
where $L(x) := \sup_{\abs{\eps} \in (0,1]} \frac{\abs{V(x+\eps) - V(x)}}{\abs{\eps}}$.

It is now immediate to verify that we may use the above decomposition instead of the second-order Taylor expansion in the proof of Proposition \ref{prop:conc-trans-model}, with the constants $D_{1,\Psi_1}$, $D^\delta_{2,\Psi_1}$ and $D_2$ replaced by the following ones (respectively):
\[
D^{W,\delta}_{1,\Psi_1} := \sup_{|\eps| \in (0,\delta]} \norm{W^\eps_1(X_0)}_{L_{\Psi_1}} ~,~ D^{W,\delta}_{2,\Psi_1} := \sup_{|\eps| \in (0,\delta]} \norm{W^\eps_2(X_0)}_{L_{\Psi_1}} ~,~
D^{W,\delta}_2 := \sup_{|\eps| \in (0,\delta]} \E(|W^\eps_2(X_0)|) ~,
\]
where as usual $X_0$ is distributed with law $\mu = \exp(-V(x)) dx$. Since we need our bounds to hold uniformly on $|\eps| \in (0,\delta]$, and since the ones given by (\ref{eq:LexpL}) are very bad, it is most convenient to consider the simplest case of a potential $V$ having global Lipschitz constant $L>0$, in which case (\ref{eq:global-L}) implies that:
\[
\forall \delta \leq 1/L \;\;\;\;\; D^{W,\delta}_{1,\Psi_1} \leq C L ~,~ D^{W,\delta}_{2,\Psi_1} \leq C L^2 ~,~ D^{W,\delta}_2 \leq C L^2 ~.
\]
The rest of the proof of Proposition \ref{prop:conc-trans-model} remains unchanged.

\section*{Appendix}
\renewcommand{\thesection}{A}
\setcounter{thm}{0}
\setcounter{equation}{0}  \setcounter{subsection}{0}

\subsection{Useful facts about log-Sobolev and spectral-gap inequalities}

\begin{thm}[Bakry--\'Emery \cite{BakryEmery}]
Let $(M,g)$ be a complete smooth oriented Riemannian manifold, equipped with its geodesic distance metric. Let $\mu = \exp(-\psi(x))\, d\vol_M(x)$ denote a probability measure on $(M,g)$, with $\psi \in C^2(M)$, and assume that as tensor fields:
\[
\exists \rho > 0 \;\;\;\;   \mathrm{Ric}_g +  \mathrm{Hess}_g \psi \geq \rho g ~.
\]
Then $(M,g,\mu)$ satisfies $LSI(\rho)$.  
\end{thm}

\begin{thm}[Herbst Argument (e.g. \cite{Ledoux-Book})]
Assume that $(\Omega,d,\mu)$ satisfies $LSI(\rho)$. Then the following Laplace-functional inequality holds:
\[
 \int \exp(\lambda f) d\mu\leq \exp(\lambda^2/(2\rho)) \;\;\; \forall \lambda \geq 0 \;\; \forall \text{ $1$-Lipschitz $f$ s.t. } \int f d\mu = 0 ~.
\]
\end{thm}

\begin{thm}[Tensorization (e.g. \cite{Ledoux-Book})]
Let $(\Omega_i,d_i,\mu_i)$, $i=1,2$ be Riemannian manifolds, equipped with their geodesic distance and 
an absolutely continuous probability measure. Observe that in this case, the geodesic distance on the product
manifols is given by the $\ell_2$ product metric
\[
\forall (x_i,y_i) \in \Omega_1 \times \Omega_2 \;\;\;\; d_1 \otimes d_2 ((x_1,y_1) , (x_2,y_2)) := \sqrt{ d_1(x_1,x_2) ^2 + d_2(y_1,y_2)^2}  ~.
 \]
Assume that $(\Omega_i,d_i,\mu_i)$ satisfies $LSI(\rho_i)$, $i=1,2$. Then the product measure-metric space $(\Omega_1 \times \Omega_2,d_1 \otimes d_2 , \mu_1 \otimes \mu_2)$ satisfies $LSI(\min(\rho_1,\rho_2))$. 
 The same statement holds for the spectral-gap inequality, with $LSI$ replaced by $SG$ in all occurrences above. 
\end{thm}

\subsection{Useful facts about log-concave measures}

\begin{thm} \label{thm:log-concave}
Let $\mu = f(x) dx$ denote a log-concave probability measure on $\Real$. Then:
\begin{enumerate}
\item $M_q(\mu)^{\frac{1}{q}} \leq C \frac{q}{p} M_p(\mu)^{\frac{1}{p}}$ for all $1 \leq p \leq q < \infty$.
\item $c_1 \leq \rho_{SG}(\Real,|\cdot|,\mu) \Var(\mu) \leq c_2$. 
\item If the barycenter of $\mu$ is at the origin then $f(0) \geq \norm{f}_{L^\infty} / e$.  
\item $c_1 \leq \norm{f}^2_{L^\infty} \Var(\mu) \leq c_2$. 
\end{enumerate}
Here $C,c_1,c_2 > 0$ are universal numeric constants. 
\end{thm}
\begin{proof}
The first assertion is a well-known Kahane--Khintchine inequality, which entails a reverse H\"{o}lder inequality for moments of linear (and more generally, homogeneous convex) functionals on the class of log-concave measures (see Berwald \cite{BerwaldMomentComparison}), or deduce this from Borell's lemma \cite{Borell-logconcave} as in \cite[Appendix III]{Milman-Schechtman-Book}). Note that the barycenter of $\mu$ may not be at the origin, but this is not required for obtaining the asserted reverse H\"{o}lder inequality. The second assertion is due to Bobkov \cite{BobkovGaussianIsoLogSobEquivalent}.
The third assertion may be found in \cite{FradeliziCentroid}, and the fourth one may be found in \cite{Milman-Pajor-LK} when $f$ is even and in \cite[Lemmas 2.5 and 2.6]{KlartagPerturbationsWithBoundedLK} or \cite[Lemma 3.2 and (3.12)]{PaourisSmallBall} in the general case.
\end{proof}

\subsection{Order Two Sub-Gaussian Chaos}

Let us now present the proof of Theorem \ref{thm:Latala}, communicated to us by Rafal Lata{\l}a. We greatly thank him for allowing us to include it here. It will be convenient to employ the following:

\begin{dfn*}
A real-valued random variable $X$ is called $\alpha$-sub-Gaussian ($\alpha > 0$), if $\norm{X}_{2k} \leq \alpha \norm{G}_{2k}$ for all integers $k \geq 1$, where $G$ denotes a standard Gaussian random-variable, and $\norm{Y}_{p} := (\E |Y|^p)^{1/p}$. 
\end{dfn*}
An elementary calculation verifies that there is a universal constant $C>0$ with:
\[
\E(\exp(\lambda X)) \leq \exp(\lambda^2 / (2 \rho)) \;\;\; \forall \lambda \in \Real \;\; \Rightarrow \;\;  \text{$X$ is $C / \sqrt{\rho}$-sub-Gaussian} ~.
\]
 Hence Theorem \ref{thm:Latala} follows from (and in fact is equivalent to) the following:

\begin{thm}[Lata{\l}a] \label{thm:Latala2}
Let $\alpha > 0$,  and let $X_1,\ldots,X_n$ denote a sequence of independent random-variables on $\Real$, so that for each $i=1,\ldots,n$, $\E(X_i) = 0$ and $X_i$ is $\alpha$-sub-Gaussian. There exists a universal constant $c > 0$ so that for any integer $n \geq 1$ and $n$ by $n$ symmetric matrix $A = \set{a_{i,j}}$ with zero diagonal:
\[
\P\brac{ \abs{\sum_{i,j=1}^n a_{i,j} X_i X_j} \geq t } \leq 2 \exp\brac{-c \min\brac{\frac{t^2}{\alpha^4 \norm{A}_{HS}^2} , \frac{t}{\alpha^2 \norm{A}_{op}}}} \;\;\; \forall t > 0 ~.
\]
\end{thm}
 
For the proof, we require the following two intermediate steps:
\begin{lem}
Let $X_1,\ldots,X_n$ be as in Theorem \ref{thm:Latala2} and let $G_1,\ldots,G_n$ denote independent standard Gaussian random-variables. 
Then for any $(a_1,\ldots,a_n) \in \Real^n$ and $p \geq 2$:
\[
\norm{\sum_{i=1}^n a_i X_i}_{p} \leq 4 \alpha \norm{\sum_{i=1}^n a_i G_i}_{p} ~.
\]
\end{lem}
\begin{proof}
Let $\set{X_i'}_{i=1}^n$ denote independent copies of $\set{X_i}_{i=1}^n$, and let $\set{\eps_i}_{i=1}^n$ denote an independent sequence of (symmetric) Bernoulli $\pm 1$ random variables. Then by the Contraction Principle (e.g. \cite{LedouxTalagrand-Book}):
\[
\norm{\sum_{i=1}^n a_i X_i}_{p} \leq \norm{\sum_{i=1}^n a_i (X_i - X_i')}_{p} = \norm{\sum_{i=1}^n a_i \eps_i (X_i - X_i')}_{p} \leq 2  \norm{\sum_{i=1}^n a_i \eps_i X_i }_{p} ~.
\]
Since $\set{\eps_i X_i}_{i=1}^n$ is a sequence of independent symmetric $\alpha$-sub-Gaussian r.v.'s, it follows immediately by algebraic expansion and vanishing of the odd coefficients that for any integer $k \geq 1$:
\[
\norm{\sum_{i=1}^n a_i \eps_i X_i }_{2k} \leq \alpha  \norm{\sum_{i=1}^n a_i G_i}_{2k} ~,
\]
and so the assertion follows for $p = 2k$. When $p \geq 2$  is not an integer, we simply choose an integer $k$ so that $2k-2 < p \leq 2k$, and evaluate:
\[
\norm{\sum_{i=1}^n a_i \eps_i X_i }_{p} \leq \norm{\sum_{i=1}^n a_i \eps_i X_i }_{2k} \leq \alpha \norm{\sum_{i=1}^n a_i G_i}_{2k} \leq \sqrt{\frac{2k-1}{p-1 }} \alpha \norm{\sum_{i=1}^n a_i G_i}_{p} \leq 2 \alpha \norm{\sum_{i=1}^n a_i G_i}_{p}  ~.
\]
\end{proof}

\begin{prop}
Let $\set{X_i}_{i=1}^n$, $\set{G_i}_{i=1}^n$ and $\set{a_{i,j}}_{i,j=1}^n$ be as in the Lemma and Theorem above. Then for any $p \geq 2$:
\[
\norm{\sum_{i,j=1}^n a_{i,j} X_i X_j}_{p} \leq C \alpha^2 \norm{\sum_{i,j=1}^n a_{i,j} G_i G_j}_{p} ~,
\]
where $C>0$ is a universal constant. 
\end{prop}
\begin{proof}
Let $\set{X_i'}_{i=1}^n$ and $\set{G_i'}_{i=1}^n$ denote independent copies of $\set{X_i}_{i=1}^n$ and $\set{G_i}_{i=1}^n$, respectively. Then:
\begin{eqnarray*}
\norm{\sum_{i,j=1}^n a_{i,j} X_i X_j}_{p} &\leq& C_1 \norm{\sum_{i,j=1}^n a_{i,j} X_i X_j'}_{p} \leq  16 C_1 \alpha^2 \norm{\sum_{i,j=1}^n a_{i,j} G_i G_j'}_{p}\\
& \leq& 
16 C_1 C_2 \alpha^2 \norm{\sum_{i,j=1}^n a_{i,j} G_i G_j}_{p} ~,
\end{eqnarray*}
where the first and third inequalities are decoupling inequalities due to de la Pe\~{n}a and Montgomery-Smith \cite{DeLaPenaMontgomerySmith-UStatistics}, and the middle one follows from separability and the Lemma. 
\end{proof}

\begin{proof}[Proof of Theorem \ref{thm:Latala2}]
By the above Proposition and the Hanson--Wright estimates on moments of Gaussian Chaoses of order 2 \cite{HansonWright}, we have for any $p \geq 2$:
\[
\norm{\sum_{i,j=1}^n a_{i,j} X_i X_j}_{p} \leq C \alpha^2 \norm{\sum_{i,j=1}^n a_{i,j} G_i G_j}_{p} \leq C' \alpha^2 (p \norm{A}_{op} + \sqrt{p} \norm{A}_{HS}) ~.
\]
The asserted tail decay estimate now follows by a standard application of Chebyshev's inequality and optimization on $p$. 

Note that the Hanson--Wright estimates in \cite{HansonWright} are in fact valid for any \emph{symmetric} sub-Gaussian distribution, and yield a dependence on $\norm{\;|A|\;}_{op}$ (where $|A|$ is the matrix with entries $\set{|a_{i,j}|}$ if $A = \set{a_{i,j}}$). However, it is important for our purposes to apply this result to non-symmetric distributions, explaining the symmetrization procedure carried out above. Furthermore, a better reference for the dependence on $\norm{A}_{op}$ as opposed to $\norm{\;|A|\;}_{op}$ above is e.g. Lata{\l}a \cite{LatalaGaussianChaos}. 
\end{proof}

\setlinespacing{1.0}
\bibliographystyle{plain}
\bibliography{../../ConvexBib}

\def\cprime{$'$}
\begin{thebibliography}{10}

\bibitem{BakryEmery}
D.~Bakry and M.~{\'E}mery.
\newblock Diffusions hypercontractives.
\newblock In {\em S\'eminaire de probabilit\'es, XIX, 1983/84}, volume 1123 of
  {\em Lecture Notes in Math.}, pages 177--206. Springer, Berlin, 1985.

\bibitem{BartheCorderoVariance}
F.~Barthe and D.~Cordero-Erausquin.
\newblock Invariances in variance estimates.
\newblock To appear in J. London Math. Soc., 2011.

\bibitem{BartheWolffGammaDistributions}
F.~Barthe and P.~Wolff.
\newblock Remarks on non-interacting conservative spin systems: the case of
  gamma distributions.
\newblock {\em Stochastic Process. Appl.}, 119(8):2711--2723, 2009.

\bibitem{BerwaldMomentComparison}
L.~Berwald.
\newblock Verallgemeinerung eines {M}ittelwertsatzes von {J}. {F}avard f\"ur
  positive konkave {F}unktionen.
\newblock {\em Acta Math.}, 79:17--37, 1947.

\bibitem{BobkovGaussianIsoLogSobEquivalent}
S.~G. Bobkov.
\newblock Isoperimetric and analytic inequalities for log-concave probability
  measures.
\newblock {\em Ann. Probab.}, 27(4):1903--1921, 1999.

\bibitem{BobkovHoudre}
S.~G. Bobkov and C.~Houdr{\'e}.
\newblock Isoperimetric constants for product probability measures.
\newblock {\em Ann. Probab.}, 25(1):184--205, 1997.

\bibitem{BodineauHelfferSurvey}
T.~Bodineau and B.~Helffer.
\newblock Correlations, spectral gap and log-{S}obolev inequalities for
  unbounded spins systems.
\newblock In {\em Differential equations and mathematical physics
  ({B}irmingham, {AL}, 1999)}, volume~16 of {\em AMS/IP Stud. Adv. Math.},
  pages 51--66. Amer. Math. Soc., Providence, RI, 2000.

\bibitem{Borell-logconcave}
Ch. Borell.
\newblock Convex measures on locally convex spaces.
\newblock {\em Ark. Mat.}, 12:239--252, 1974.

\bibitem{CaputoUniformPoincareForPerturbedStricltyConvexPonentials}
P.~Caputo.
\newblock Uniform {P}oincar\'e inequalities for unbounded conservative spin
  systems: the non-interacting case.
\newblock {\em Stochastic Process. Appl.}, 106(2):223--244, 2003.

\bibitem{SOS}
P.~Caputo, F.~Martinelli, and Toninelli~F. L.
\newblock Mixing time of monotone surfaces and {S}{O}{S} interfaces: a mean
  curvature approach.
\newblock arXiv:1101.4190v1, 2011.

\bibitem{ChafaiConservativeSpinSystems}
D.~Chafa{\"{\i}}.
\newblock Glauber versus {K}awasaki for spectral gap and logarithmic {S}obolev
  inequalities of some unbounded conservative spin systems.
\newblock {\em Markov Process. Related Fields}, 9(3):341--362, 2003.

\bibitem{CheegerInq}
J.~Cheeger.
\newblock A lower bound for the smallest eigenvalue of the {L}aplacian.
\newblock In {\em Problems in analysis (Papers dedicated to Salomon Bochner,
  1969)}, pages 195--199. Princeton Univ. Press, Princeton, N. J., 1970.

\bibitem{DeLaPenaMontgomerySmith-UStatistics}
V.~H. de~la Pe{\~n}a and S.~J. Montgomery-Smith.
\newblock Bounds on the tail probability of {$U$}-statistics and quadratic
  forms.
\newblock {\em Bull. Amer. Math. Soc. (N.S.)}, 31(2):223--227, 1994.

\bibitem{Federbush-LogSobolev}
P.~Federbush.
\newblock Partially alternate derivation of a result of {N}elson.
\newblock {\em J. Math. Physics}, 10(1):50--52, 1969.

\bibitem{FradeliziCentroid}
M.~Fradelizi.
\newblock Sections of convex bodies through their centroid.
\newblock {\em Arch. Math. (Basel)}, 69(6):515--522, 1997.

\bibitem{GardnerSurveyInBAMS}
R.~J. Gardner.
\newblock The {B}runn-{M}inkowski inequality.
\newblock {\em Bull. Amer. Math. Soc. (N.S.)}, 39(3):355--405, 2002.

\bibitem{GromovMilmanLevyFamilies}
M.~Gromov and V.~D. Milman.
\newblock A topological application of the isoperimetric inequality.
\newblock {\em Amer. J. Math.}, 105(4):843--854, 1983.

\bibitem{GrossIntroducesLogSob}
L.~Gross.
\newblock Logarithmic {S}obolev inequalities.
\newblock {\em Amer. J. Math.}, 97(4):1061--1083, 1975.

\bibitem{GOVWTwoScaleApproachForLSI}
N.~Grunewald, F.~Otto, C.~Villani, and M.~G. Westdickenberg.
\newblock A two-scale approach to logarithmic {S}obolev inequalities and the
  hydrodynamic limit.
\newblock {\em Ann. Inst. Henri Poincar\'e Probab. Stat.}, 45(2):302--351,
  2009.

\bibitem{HansonWright}
D.~L. Hanson and F.~T. Wright.
\newblock A bound on tail probabilities for quadratic forms in independent
  random variables.
\newblock {\em Ann. Math. Statist.}, 42:1079--1083, 1971.

\bibitem{HolleyStroockPerturbationLemma}
R.~Holley and D.~Stroock.
\newblock Logarithmic {S}obolev inequalities and stochastic {I}sing models.
\newblock {\em J. Statist. Phys.}, 46(5-6):1159--1194, 1987.

\bibitem{KLS}
R.~Kannan, L.~Lov{\'a}sz, and M.~Simonovits.
\newblock Isoperimetric problems for convex bodies and a localization lemma.
\newblock {\em Discrete Comput. Geom.}, 13(3-4):541--559, 1995.

\bibitem{KlartagPerturbationsWithBoundedLK}
B.~Klartag.
\newblock On convex perturbations with a bounded isotropic constant.
\newblock {\em Geom. and Funct. Anal.}, 16(6):1274--1290, 2006.

\bibitem{LandimPanizoYauLSIForPerturbationOfQuadraticPotential}
C.~Landim, G.~Panizo, and H.~T. Yau.
\newblock Spectral gap and logarithmic {S}obolev inequality for unbounded
  conservative spin systems.
\newblock {\em Ann. Inst. H. Poincar\'e Probab. Statist.}, 38(5):739--777,
  2002.

\bibitem{LatalaGaussianChaos}
R.~Lata{\l}a.
\newblock Estimates of moments and tails of {G}aussian chaoses.
\newblock {\em Ann. Probab.}, 34(6):2315--2331, 2006.

\bibitem{LedouxBusersTheorem}
M.~Ledoux.
\newblock A simple analytic proof of an inequality by {P}. {B}user.
\newblock {\em Proc. Amer. Math. Soc.}, 121(3):951--959, 1994.

\bibitem{LedouxLectureNotesOnDiffusion}
M.~Ledoux.
\newblock The geometry of {M}arkov diffusion generators.
\newblock {\em Ann. Fac. Sci. Toulouse Math. (6)}, 9(2):305--366, 2000.

\bibitem{Ledoux-Book}
M.~Ledoux.
\newblock {\em The concentration of measure phenomenon}, volume~89 of {\em
  Mathematical Surveys and Monographs}.
\newblock American Mathematical Society, Providence, RI, 2001.

\bibitem{LedouxSpinSystemsRevisited}
M.~Ledoux.
\newblock Logarithmic {S}obolev inequalities for unbounded spin systems
  revisited.
\newblock In {\em S\'eminaire de {P}robabilit\'es, {XXXV}}, volume 1755 of {\em
  Lecture Notes in Math.}, pages 167--194. Springer, Berlin, 2001.

\bibitem{LedouxConcentrationToIsoperimetryUsingSemiGroups}
M.~Ledoux.
\newblock From concentration to isoperimetry: Semigroup proofs.
\newblock In C.~Houdr\'e, M.~Ledoux, E.~Milman, and M.~Milman, editors, {\em
  Concentration, Functional Inequalities and Isoperimetry}, volume 545 of {\em
  Contemporary Mathematics}, pages 155--166. Amer. Math. Soc., 2011.

\bibitem{LedouxTalagrand-Book}
M.~Ledoux and M.~Talagrand.
\newblock {\em Probability in {B}anach spaces}, volume~23 of {\em Ergebnisse
  der Mathematik und ihrer Grenzgebiete (3) [Results in Mathematics and Related
  Areas (3)]}.
\newblock Springer-Verlag, Berlin, 1991.
\newblock Isoperimetry and processes.

\bibitem{Lichnerowicz1970GenRicciTensorCRAS}
A.~Lichnerowicz.
\newblock Vari\'et\'es riemanniennes \`a tenseur {C} non n\'egatif.
\newblock {\em C. R. Acad. Sci. Paris S\'er. A-B}, 271:A650--A653, 1970.

\bibitem{LuYauLSIForKawasakiAndGlauber}
S.~L. Lu and H.-T. Yau.
\newblock Spectral gap and logarithmic {S}obolev inequality for {K}awasaki and
  {G}lauber dynamics.
\newblock {\em Comm. Math. Phys.}, 156(2):399--433, 1993.

\bibitem{MazyaSobolevImbedding}
V.~G. Maz{\cprime}ja.
\newblock Classes of domains and imbedding theorems for function spaces.
\newblock {\em Dokl. Acad. Nauk SSSR}, 3:527--530, 1960.
\newblock Engl. transl. Soviet Math. Dokl., 1 (1961) 882--885.

\bibitem{MazyaCapacities}
V.~G. Maz{\cprime}ja.
\newblock {$p$}-conductivity and theorems on imbedding certain functional
  spaces into a {$C$}-space.
\newblock {\em Dokl. Akad. Nauk SSSR}, 140:299--302, 1961.
\newblock Engl. transl. Soviet Math. Dokl., 2 (1961) 1200�-1203.

\bibitem{MenzLSIwithWeakInteraction}
G.~Menz.
\newblock {L}{S}{I} for {K}awasaki dynamics with weak interaction.
\newblock {\em Comm. Math. Phys.}, 307(3):817--860, 2011.

\bibitem{MenzOttoLSIForPerturbationsOfSuperQuadraticPotential}
G.~Menz and F.~Otto.
\newblock Uniform logarithmic {S}obolev inequalities for conservative spin
  systems with super-quadratic single-site potential.
\newblock to appear in Annals of Probability, 2011.

\bibitem{EMilman-RoleOfConvexity}
E.~Milman.
\newblock On the role of convexity in isoperimetry, spectral-gap and
  concentration.
\newblock {\em Invent. Math.}, 177(1):1--43, 2009.

\bibitem{EMilmanGeometricApproachPartI}
E.~Milman.
\newblock Isoperimetric and concentration inequalities - equivalence under
  curvature lower bound.
\newblock {\em Duke Math. J.}, 154(2):207--239, 2010.

\bibitem{EMilmanGeometricApproachPartII}
E.~Milman.
\newblock Properties of isoperimetric, functional and transport-entropy
  inequalities via concentration.
\newblock to appear in Prob. Theor. Rel. Fields, arxiv.org/abs/0909.0207, 2010.

\bibitem{EMilmanIsoperimetricBoundsOnManifolds}
E.~Milman.
\newblock Isoperimetric bounds on convex manifolds.
\newblock In C.~Houdr\'e, M.~Ledoux, E.~Milman, and M.~Milman, editors, {\em
  Concentration, Functional Inequalities and Isoperimetry}, volume 545 of {\em
  Contemporary Mathematics}, pages 195--208. Amer. Math. Soc., 2011.

\bibitem{EMilmanSharpIsopInqsForCDD}
E.~Milman.
\newblock Sharp isoperimetric inequalities and model spaces for
  curvature-dimension-diameter condition.
\newblock submitted, arxiv.org/abs/1108.4609, 2011.

\bibitem{MilmanSurveyOnConcentrationOfMeasure1988}
V.~D. Milman.
\newblock The heritage of {P}.\ {L}\'evy in geometrical functional analysis.
\newblock {\em Ast\'erisque}, 157/158:273--301, 1988.
\newblock Colloque Paul L{\'e}vy sur les Processus Stochastiques (Palaiseau,
  1987).

\bibitem{Milman-Pajor-LK}
V.~D. Milman and A.~Pajor.
\newblock Isotropic position and interia ellipsoids and zonoids of the unit
  ball of a normed $n$-dimensional space.
\newblock In {\em Geometric Aspects of Functional Analysis}, volume 1376 of
  {\em Lecture Notes in Mathematics}, pages 64--104. Springer-Verlag,
  1987-1988.

\bibitem{Milman-Schechtman-Book}
V.~D. Milman and G.~Schechtman.
\newblock {\em Asymptotic theory of finite-dimensional normed spaces}, volume
  1200 of {\em Lecture Notes in Mathematics}.
\newblock Springer-Verlag, Berlin, 1986.
\newblock With an appendix by M. Gromov.

\bibitem{OttoReznikoffLSICriterion}
F.~Otto and M.~G. Reznikoff.
\newblock A new criterion for the logarithmic {S}obolev inequality and two
  applications.
\newblock {\em J. Funct. Anal.}, 243(1):121--157, 2007.

\bibitem{PaourisSmallBall}
G.~Paouris.
\newblock Small ball probability estimates for log-concave measures.
\newblock To appear in Trans. Amer. Math. Soc., 2010.

\bibitem{Petrov-SumsOfIndependentRVsBook}
V.~V. Petrov.
\newblock {\em Sums of independent random variables}.
\newblock Springer-Verlag, New York, 1975.
\newblock Translated from the Russian by A. A. Brown, Ergebnisse der Mathematik
  und ihrer Grenzgebiete, Band 82.

\bibitem{Stam-LogSobolev}
A.~J. Stam.
\newblock Some inequalities satisfied by the quantities of information of
  {F}isher and {S}hannon.
\newblock {\em Information and Control}, 2:101--112, 1959.

\bibitem{StroockZegarlinskiLogSobImpliesDSCondition}
D.~W. Stroock and B.~Zegarli{\'n}ski.
\newblock The equivalence of the logarithmic {S}obolev inequality and the
  {D}obrushin-{S}hlosman mixing condition.
\newblock {\em Comm. Math. Phys.}, 144(2):303--323, 1992.

\bibitem{VaradhanNonlinearDiffusionLimit}
S.~R.~S. Varadhan.
\newblock Nonlinear diffusion limit for a system with nearest neighbor
  interactions. {II}.
\newblock In {\em Asymptotic problems in probability theory: stochastic models
  and diffusions on fractals ({S}anda/{K}yoto, 1990)}, volume 283 of {\em
  Pitman Res. Notes Math. Ser.}, pages 75--128. Longman Sci. Tech., Harlow,
  1993.

\bibitem{YoshidaLogSobEquivToDecayOfCorrelations}
N.~Yoshida.
\newblock The equivalence of the log-{S}obolev inequality and a mixing
  condition for unbounded spin systems on the lattice.
\newblock {\em Ann. Inst. H. Poincar\'e Probab. Statist.}, 37(2):223--243,
  2001.

\end{thebibliography}
\end{document}